\documentclass[aps,prd,twocolumn,superscriptaddress,nofootinbib,longbibliography]{revtex4-2}

\usepackage{amsmath,amssymb,amsthm}
\usepackage{xcolor}
\newcommand{\rev}[1]{#1}
\usepackage{physics}
\usepackage{mathrsfs}
\usepackage{graphicx}
\usepackage{booktabs}
\usepackage{placeins}
\usepackage[colorlinks=true,linkcolor=blue,citecolor=blue,urlcolor=blue]{hyperref}
\usepackage[capitalise,nameinlink,noabbrev]{cleveref}

\crefname{section}{Sec.}{Secs.}
\Crefname{section}{Section}{Sections}
\crefname{subsection}{Sec.}{Secs.}
\Crefname{subsection}{Section}{Sections}
\crefname{equation}{Eq.}{Eqs.}
\Crefname{equation}{Equation}{Equations}
\crefname{figure}{Fig.}{Figs.}
\Crefname{figure}{Figure}{Figures}
\crefname{table}{Table}{Tables}
\Crefname{table}{Table}{Tables}
\crefname{appendix}{App.}{Apps.}

\newtheorem{theorem}{Theorem}
\newtheorem{corollary}{Corollary}
\newtheorem{proposition}{Proposition}
\newtheorem{remark}{Remark}
\newtheorem{conjecture}{Conjecture}

\graphicspath{{./}}
\begin{document}

 \title{Gravitational Enstrophy: Local Geometric Origin and Energy-Cascade Constraints}

\author{Luis Lehner}
\affiliation{Perimeter Institute for Theoretical Physics, 31 Caroline Street North, Waterloo, ON N2L 2Y5, Canada}

\begin{abstract}
Two-dimensional incompressible fluids conserve energy and enstrophy in the inviscid limit. Fj{\o}rtoft's argument uses this pair to constrain nonlinear energy transfer. We examine the corresponding question in General Relativity using the magnetic Weyl functional $\mathcal Z=\int B_{ab}B^{ab}\sqrt{\gamma}\,d^3x$. Its transverse radiative part carries an additional frequency-squared weight relative to gravitational-wave energy: $\delta\mathcal Z_k\simeq\omega_k^2W_k$ in the radiation zone, with $W=8\pi G E_{\rm GW}$. We derive its balance law from the Bianchi identities, obtaining the algebraic cancellations, bulk sources and boundary fluxes. Conservation of positive energy and curvature norm through nonlinear transfer gives a Fj{\o}rtoft constraint when their weights are ordered. Resonant weak-wave interactions generally change the curvature norm; $2\leftrightarrow2$ dynamics instead conserves energy and wave action. We examine the competition between transfer and linear losses in Kerr and AdS backgrounds, where long-lived or confined modes provide time for nonlinear evolution. In the hydrodynamic regime of AdS$_4$, a fluid-weighted magnetic Weyl functional inherits fluid enstrophy conservation in the inviscid limit. We compute its proportionality coefficient to boundary enstrophy in closed form at leading gradient order, giving a bulk diagnostic for holographic turbulence.

\end{abstract}

\maketitle

\section{Introduction}
\label{sec:introduction}

In ideal two-dimensional incompressible flow, energy and enstrophy conservation constrain how energy moves between scales. Fj{\o}rtoft's classical argument makes this particularly clear: energy leaving a source scale must be shared between smaller and larger scales~\cite{Fjortoft1953,1980RPPh...43..547K}. Recent observations of inverse energy transfer in nonlinear black-hole perturbations~\cite{YangZimmermanLehner2015,iuliano2024extremalblackholeweather,MaLehnerYangKidderPfeifferScheel2026} and in holographic turbulence~\cite{CarrascoLehnerMyersReulaSingh2012,Adams:2013vsa} naturally raise the question: \emph{does General Relativity possess a corresponding local conservation structure? If so, what form does it take?}

We approach this question through the electric--magnetic decomposition of the Weyl tensor. For a spacelike hypersurface \(\Sigma\) with future-directed unit normal \(u^a\), consider
\[
    \mathcal Z[\Sigma]
    =
    \int_\Sigma
    B_{ab}B^{ab}\sqrt{\gamma}\,d^3x,
\]
where \(B_{ab}={}^\star C_{acbd}u^cu^d\) is the magnetic Weyl tensor measured by \(u^a\). This functional measures magnetic curvature. Its connection with enstrophy can be made in two ways: through the frequency weighting of gravitational radiation, and through the relation between bulk curvature and boundary vorticity in fluid/gravity.\footnote{Earlier constructions of gravitational enstrophy have drawn on holography and curvature--vorticity relations, e.g.~\cite{CarrascoLehnerMyersReulaSingh2012,Eling:2013sna,marjieh2021enstrophysymmetry,Donnelly:2020xgu,Redondo-Yuste:2022czt}.}

For radiative perturbations of dimensionless amplitude $\epsilon$, we use the quadratic norm of the first-order transverse curvature,
\[
 \delta\mathcal Z=\int_\Sigma
 \delta B^{\rm rad}_{ab}\delta B_{\rm rad}^{ab}\,dV_{\bar\gamma}.
\]
The field $\delta B^{\rm rad}_{ab}=O(\epsilon)$ is reconstructed from the transverse NP scalars on the background geometry. We suppress the superscript ``rad'' when the sector is clear; the tetrad conventions and the full variation of the magnetic Weyl functional are discussed in \cref{app:np-stf,app:coulomb-nonradiative}.

Curvature contains two derivatives of the metric perturbation, while the Isaacson energy is quadratic in its first derivatives~\cite{Isaacson1968b}. For nearly monochromatic, weakly damped radiation this gives
\[
 \eta_k\equiv\frac{\delta\mathcal Z_k}{\omega_k^2 W_k}\simeq1
 \qquad\text{(Proposition~\ref{thm:spectral-weighting})},
\]
in the asymptotically flat radiation zone. We use $W=8\pi G E_{\rm GW}$ throughout the radiative discussion. Thus high-frequency waves contribute more strongly to $\delta\mathcal Z$ than to the energy, just as short-wavelength vortical structures contribute more strongly to fluid enstrophy. Near the black hole, background curvature also enters the relation between the two norms. For low-$m$ zero-damping modes (ZDMs), concentrated in the near-horizon throat, this contribution can be comparable to the derivative terms. Their spectral weights therefore require the strong-field calculation discussed in \cref{app:spectral-ratio-proof}.

The next question is how $\mathcal Z$ evolves. Starting from the vacuum Bianchi identities, we derive its balance law and identify the bulk sources and boundary fluxes. Several simplifications follow directly from the Weyl algebra. The vorticity contraction vanishes identically. For co-propagating type-N radiation the electric and magnetic norms agree, allowing their curl terms to combine into a boundary flux. In the zero-angular-momentum observer (ZAMO) frame, the shear contraction vanishes for a field transverse in that frame. The remaining terms describe acceleration and exchange with the other curvature components, including the sources introduced by linearization and projection. Their sum gives the evolution of the radiative norm.

The Fj{\o}rtoft argument then concerns the pair
\[
    W=\sum_kW_k,
    \qquad
    \mathcal Z
    =
    \sum_k\delta\mathcal Z_k
    \simeq
    \sum_k\omega_k^2W_k.
\]
With positive energy and ordered spectral weights, conservation of both quantities during nonlinear transfer constrains the receiving modes (Theorem~\ref{thm:fjortoft}). Energy leaving a single source has to reach both sides of it in spectral weight. For receivers equidistant in frequency, the lower-frequency side receives more energy; in general, the division follows from their distances in spectral weight. This gives the gravitational counterpart of the fluid argument. The couplings determine which of the allowed transfers occur and how fast they proceed.

Which pair is conserved depends on the dynamics. Ideal incompressible flow conserves fluid enstrophy, and in the hydrodynamic regime of AdS$_4$ we find a fluid-weighted magnetic Weyl functional proportional to it. We compute the coefficient $\tilde{\mathcal C}(T)$ in closed form at leading order in the gradient expansion. This supplies a quantitative bulk diagnostic for holographic turbulence.\footnote{Fully nonlinear numerical simulations exhibit the fluid/gravity turbulent correspondence~\cite{CarrascoLehnerMyersReulaSingh2012,Adams:2013vsa,GreenCarrascoLehner2014,Westernacher-Schneider:2015gfa}.} Weak radiative waves behave differently: resonant $2\leftrightarrow2$ interactions conserve $(N,W)$, with $N=\sum_kn_k$ and $W=\sum_k\omega_kn_k$, while generally changing the radiation-zone curvature norm $\mathcal Z=\sum_k\omega_k^3n_k$. Three-wave processes also change $N$. The examples in \cref{subsec:fjortoft-interpretation} show this directly from the resonance relations; the coupling calculation will be presented elsewhere.

Nonlinear transfer must also outpace damping. For an isolated linear mode on a stationary background, the quadratic envelope is proportional to $e^{2\omega_I t}$, so a finite stationary norm changes little over the transfer time when $2|\omega_I|\tau_{\rm nl}\ll1$. Near-extremal Kerr is a natural place to examine this competition because the ZDM lifetimes grow without bound. For fundamental ZDMs, the reference rate $\tau_{\rm nl}^{-1}\sim\epsilon^2\omega_R$ gives a linear-loss parameter $\alpha_{\rm loss}\sim\sqrt{1-\chi}/\epsilon^2$; we also discuss the cubic alternative. Reflecting Anti--de Sitter boundaries suppress outer leakage, while a black-hole horizon continues to absorb or amplify radiation. In generic weakly perturbed black-hole ringdown, damping instead competes strongly with nonlinear transfer.

The analysis below is perturbative. Sections~\ref{sec:weyl-geometry}--\ref{sec:evolution} develop the quadratic Weyl structure, its spectral weighting and its balance law. Section~\ref{sec:fjortoft} gives the transfer constraint and examines its nonlinear premise. We then turn to Kerr and AdS backgrounds in Sec.~\ref{sec:flux-regimes}, and derive the fluid/gravity map in Sec.~\ref{sec:holographic-fluid}. Section~\ref{sec:conclusion} discusses the implications and open questions. The appendices give the Bianchi and NP conventions, the bulk and boundary calculations, and extensions of the transfer argument.

\section{Electric--Magnetic Weyl Geometry}
\label{sec:weyl-geometry}

Let \((\mathcal{M}, g_{ab})\) be a four-dimensional spacetime, and \(u^a\) a future-directed unit timelike vector field (\(g_{ab}u^a u^b = -1\)). Define the spatial projector and volume form by
\[
    \gamma_{ab} = g_{ab} + u_a u_b, \qquad
    \epsilon_{abc} = u^d\epsilon_{dabc}.
\]
The electric and magnetic parts of the Weyl tensor are
\begin{gather*}
    E_{ab} = C_{acbd} u^c u^d, \\
    B_{ab} = {}^{\star}C_{acbd} u^c u^d = \frac{1}{2} \epsilon_{ac}{}^{ef} C_{efbd} u^c u^d.
\end{gather*}
Both \(E_{ab}\) and \(B_{ab}\) are symmetric, trace-free, and spatial (e.g. \(B_{ab} = B_{\langle ab \rangle}\), \(B_{ab} u^b = 0\)). This is a local decomposition requiring only the Weyl tensor at a point and the observer \(u^a\) (just as electromagnetic fields depend on the observer, so do \(E_{ab}\) and \(B_{ab}\)). For applications below, \(u^a\) is chosen as a natural observer for the background (e.g., ZAMOs in Kerr~\cite{bardeen1972rotating}), but the algebraic structure here is independent of global spacetime properties.

\subsection{Weyl Invariants and Quadratic Bilinears}
\label{subsec:weyl-bilinears}

The electric–magnetic decomposition naturally introduces the quadratic Weyl densities and associated quantities. In particular, the sum
\[
    E_{ab}E^{ab} + B_{ab}B^{ab},
\]
is the Bel--Robinson super-energy density measured by \(u^a\)~\cite{Bel1958,Misner1973}. The natural local quadratic basis is,
\[
    q_E = E_{ab}E^{ab}, \quad q_B = B_{ab}B^{ab}, \quad q_H = E_{ab}B^{ab}.
\]
At this stage, no fluid analogy is invoked: these are purely geometric objects arising from the Weyl decomposition. Gravitational enstrophy will emerge later as a property of this family: in the transverse radiative sector, its positive members share the spectral weighting of fluid enstrophy, with \(\mathcal{Z} \equiv \mathcal{Q}_B\) as the natural representative (\cref{sec:spectral-weighting}).

On a spacelike hypersurface \(\Sigma\) orthogonal to \(u^a\), with induced metric \(\gamma_{ab}\), we define the integrated functionals
\[
    \mathcal{Q}_{\{E,B,H\}}[\Sigma] = \int_\Sigma q_{\{E,B,H\}} \,\sqrt{\gamma} \, d^3x.
\]
Such hypersurfaces exist whenever the congruence of \(u^a\) is vorticity-free, as the ZAMO congruence used below is by the Frobenius theorem (\cref{app:zamo-shear}).

\subsection{Bianchi Evolution}
\label{subsec:bianchi-evolution}

On a four-dimensional Einstein background (\(G_{ab} + \Lambda g_{ab} = 0\)), the Bianchi identity implies
\[
    \nabla^a C_{abcd} = 0.
\]
Thus, the Weyl tensor obeys the same source-free Bianchi equation on all vacuum Einstein backgrounds, regardless of \(\Lambda\). While \(\Lambda\) affects the geometry and perturbation spectrum, it does not alter the local algebraic form of the Weyl Bianchi equations.

Projecting along and orthogonal to \(u^a\) yields a Maxwell-like system for \(E_{ab}\) and \(B_{ab}\). Define the spatial covariant derivative, time derivative, STF projection, and curl as
\begin{gather*}
    D_a T^{b\cdots}{}_{c\cdots} = \gamma_a{}^d \gamma^b{}_e \cdots \nabla_d T^{e\cdots}{}_{f\cdots},
    \qquad
    \dot{T}_{ab} = u^c \nabla_c T_{ab}, \\
    T_{\langle ab \rangle} = \left( \gamma_{(a}{}^c \gamma_{b)}{}^d - \tfrac{1}{3} \gamma_{ab} \gamma^{cd} \right) T_{cd}, \\
    (\curl S)_{ab} = \epsilon_{cd\langle a} D^c S_{b\rangle}{}^d \quad (S_{ab} \text{ STF}).
\end{gather*}
The Bianchi identities then take the local form~\cite{Maartens1998,1996CQGra..13.1451F}
\begin{align}
    \dot{E}_{ ab } - (\curl B)_{ab} &= \mathcal{K}^E_{ab}[E, B; u],
    \label{eq:E-evolution-local} \\
    \dot{B}_{ ab } + (\curl E)_{ab} &= \mathcal{K}^B_{ab}[E, B; u],
    \label{eq:B-evolution-local}
\end{align}
(supplemented by spatial constraints for \(D^b E_{ab}\) and \(D^b B_{ab}\).) 
Here, \(\mathcal{K}^E_{ab}\) and \(\mathcal{K}^B_{ab}\) are algebraic in \(E_{ab}\), \(B_{ab}\), with the kinematic contributions of \(u^a\) (expansion, shear, vorticity, acceleration) as shown in \ref{app:spatial-projection}. In particular, with
\(\omega_{ab}\) and \(\sigma_{ab}\) the vorticity and shear of \(u^a\), the most relevant terms in \(\mathcal{K}^B_{ab}\) are
\[
    \mathcal{K}^B_{ab} \supset B_{c\langle a} \omega_{b\rangle}{}^c + 3\sigma_{c\langle a} B_{b\rangle}{}^c - 2\epsilon_{cd\langle a}\, a^c E_{b\rangle}{}^d + \cdots,
\]
 The vorticity term is the STF rewriting of \(-\epsilon_{cd\langle a}\omega^c B_{b\rangle}{}^d\), with \(\omega^a = \tfrac{1}{2}\epsilon^{abc}\omega_{bc}\) the vorticity vector; the full source terms are given in \cref{app:three-plus-one-bianchi}. The omitted terms are additional local kinematic couplings, whose precise form depends on the congruence.
\Cref{eq:E-evolution-local,eq:B-evolution-local} are \emph{local differential equations}: boundary conditions and asymptotics enter only upon integration over a hypersurface, where flux terms arise.

\subsection{Local Evolution of Quadratic Densities}
\label{subsec:quadratic-density-evolution}

The evolution of the quadratic densities follows directly from \cref{eq:E-evolution-local,eq:B-evolution-local}. Along the congruence, with \(\tau\) the proper time of \(u^a\) and \(d/d\tau \equiv u^a \nabla_a\) (the scalar counterpart of the overdot), we have for \(q_B\)
\[
    \frac{1}{2} \frac{d}{d\tau} q_B = -B^{ab} (\curl E)_{ab} + B^{ab} \mathcal{K}^B_{ab}.
\]
Similarly (for completeness),
\begin{gather*}
    \frac{1}{2} \frac{d}{d\tau} q_E = E^{ab} (\curl B)_{ab} + E^{ab} \mathcal{K}^E_{ab}, \\
    \frac{d}{d\tau} q_H = \dot{E}_{\langle ab \rangle} B^{ab} + E^{ab} \dot{B}_{\langle ab \rangle}.
\end{gather*}
The quadratic Weyl bilinears form a closed local structure as their evolution is governed by curl exchange between \(E_{ab}\) and \(B_{ab}\), plus local kinematic couplings.

Integrating over a slice \(\Sigma\) of the foliation adapted to \(u^a\), with time function \(t\) satisfying \(t^a = \alpha u^a + \beta^a\) for lapse \(\alpha\) and shift \(\beta^a\), gives
\begin{multline*}
    \frac{1}{2} \frac{d}{dt} \mathcal{Q}_B = -\int_\Sigma \alpha\, B^{ab} (\curl E)_{ab} \sqrt{\gamma} \, d^3x \\
    + \int_\Sigma \alpha\left(B^{ab}\mathcal{K}^B_{ab}+\tfrac12\theta q_B\right)\sqrt{\gamma}\,d^3x
    +\tfrac12\oint_{\partial\Sigma}\beta^a s_a q_B\,dA.
\end{multline*}
Here the domain is fixed in the spatial coordinates, $s^a$ is its outward unit normal, and $\partial_t\sqrt\gamma=\sqrt\gamma(\alpha\theta+D_a\beta^a)$. The volume-transport term is a bulk term; further boundary terms arise if the curl is integrated by parts.

Conservation is not automatic as bulk terms are controlled by geometric mechanisms (vorticity, shear, curl-exchange cancellations), while flux terms depend on the global setting. This distinction will be important when comparing Kerr, Kerr--AdS, and confined domains.

Notice we have yet to invoke enstrophy: we have merely identified the quadratic Weyl structure. The next step (\cref{sec:spectral-weighting}) is to show that \(\mathcal{Q}_B\) has the spectral weighting of an enstrophy analogue.

\section{Spectral Weighting and the Enstrophy Analogy}
\label{sec:spectral-weighting}

The $\{q_E,q_B,q_H\}$ are observer-dependent scalars of local Weyl geometry. In the radiation zone their radiative parts carry the spectral weighting of enstrophy 
since radiative curvature contains \emph{two derivatives} of the metric perturbation. Thus, any quadratic radiative Weyl norm scales as \(\omega^4\) for a mode of frequency \(\omega\), while the perturbative gravitational-wave energy scales as \(\omega^2\). The \emph{enstrophy-like weighting} \(\mathcal{Z}_k / W_k \sim \omega_k^2\) is a generic property of radiative quadratic Weyl structures, independent of the choice of scalar representative.

A natural positive-definite quantity that can be identified within this family is the magnetic Weyl functional
\[
    \mathcal{Z}_B[\Sigma] = \int_\Sigma B_{ab}B^{ab} \sqrt{\gamma} \, d^3x,
\]
with its perturbative radiative part
\[
    \delta\mathcal{Z}_B[\Sigma] = \int_\Sigma \delta B_{ab} \delta B^{ab} \sqrt{\gamma} \, d^3x.
\]
We select \(B_{ab}B^{ab}\) for three independent reasons. First, \(B_{ab}\) is the sector of the Weyl tensor sourced by differential frame-dragging, the geometric counterpart of fluid vorticity, whose square is the fluid enstrophy density. Second, several of its kinematic couplings simplify along ZAMO congruences (\cref{sec:evolution}). Third, a fluid-weighted version of the full magnetic sector maps naturally onto the boundary fluid enstrophy (\cref{sec:holographic-fluid}). This last construction retains components excluded by the radiative projection. The choice therefore follows from the structure of the Bianchi identities and of the fluid/gravity map, rather than being put in by hand.

\begin{remark}[Phase-averaged mode sums]
\label{rem:coarse-grained-WZ}
Where a modal sum is used, the averaging interval must be long compared with the relevant periods and beat times, and short compared with damping and transfer times. We write $\omega_k=\operatorname{Re}\omega_k^{\rm QNM}$ unless the complex frequency is displayed explicitly. For a complex-frequency quadratic envelope, the radiation-zone weight is $|\omega_k^{\rm QNM}|^2$; replacing it by $\omega_k^2$ requires $|\operatorname{Im}\omega_k^{\rm QNM}|\ll\omega_k$. For degenerate or nearly degenerate modes, interference can survive this average. In particular, we retain the near-extremal ZDMs of a fixed $m$ as a quadratic block, including their interference terms (Remark~\ref{rem:duality-pointwise}).

\end{remark}

\subsection{Radiative Weyl Components}
\label{subsec:radiative-weyl-components}

Consider a Newman--Penrose tetrad~\cite{Newman:1962cia}  \((\ell^a, n^a, m^a, \bar{m}^a)\) adapted to the observer \(u^a\) and a radial direction \(\hat{r}^a\):
\begin{gather*}
    \ell^a = \frac{1}{\sqrt{2}}(u^a + \hat{r}^a), \qquad
    n^a = \frac{1}{\sqrt{2}}(u^a - \hat{r}^a), \\
    m^a = \frac{1}{\sqrt{2}}(\hat{\theta}^a + i\hat{\varphi}^a),
\end{gather*}
where \(\hat{\theta}^a\) and \(\hat{\varphi}^a\) span the transverse subspace. The transverse symmetric trace-free (STF) polarization tensors are
\[
    Q^{(+)}_{ab} = \frac{1}{2} (\hat{\theta}_a \hat{\theta}_b - \hat{\varphi}_a \hat{\varphi}_b), \quad
    Q^{(\times)}_{ab} = \hat{\theta}_{(a} \hat{\varphi}_{b)},
\]
satisfying \(Q^{(+)ab} Q^{(+)}_{ab} = Q^{(\times)ab} Q^{(\times)}_{ab} = \frac{1}{2}\) and \(Q^{(+)ab} Q^{(\times)}_{ab} = 0\).

For an outgoing radiative perturbation (reconstructed from \(\Psi_4\)), the transverse electric and magnetic Weyl perturbations are
\[
    \delta E_{ab}^{(4)} = e_+ Q^{(+)}_{ab} + e_\times Q^{(\times)}_{ab}, \quad
    \delta B_{ab}^{(4)} = e_\times Q^{(+)}_{ab} - e_+ Q^{(\times)}_{ab},
\]
where \(e_+\) and \(e_\times\) are the \emph{scalar polarization amplitudes} (curvature-level analogues of \(h_+\) and \(h_\times\)), with \(\Psi_4 = e_+ + i e_\times\) up to normalization (\cref{app:np-stf}). The magnetic perturbation is thus a 90° rotation of the electric perturbation in polarization space. Since rotations preserve norms,
\begin{equation}
    \delta B_{ab}^{(4)} \delta B_{(4)}^{ab} = \delta E_{ab}^{(4)} \delta E_{(4)}^{ab} = \frac{1}{2} |\Psi_4|^2.
    \label{eq:B4sq-Psi4}
\end{equation}
Similarly, for \emph{ingoing radiative modes} (from \(\Psi_0\)):
\begin{equation}
    \delta B_{ab}^{(0)} \delta B_{(0)}^{ab} = \delta E_{ab}^{(0)} \delta E_{(0)}^{ab} = \frac{1}{2} |\Psi_0|^2.
    \label{eq:B0sq-Psi0}
\end{equation}

Of course, a general radiative perturbation may include both outgoing and ingoing components:
\[
    \delta B_{ab}^{\rm rad} = \delta B_{ab}^{(4)} + \delta B_{ab}^{(0)}.
\]
The pointwise norm then includes an interference term:
\[
    \delta B_{ab}^{\rm rad} \delta B_{\rm rad}^{ab} = \delta B_{ab}^{(4)} \delta B_{(4)}^{ab} + \delta B_{ab}^{(0)} \delta B_{(0)}^{ab} + 2 \delta B_{ab}^{(4)} \delta B_{(0)}^{ab}.
\]
For a general superposition, the energy and curvature norm are quadratic forms on the amplitudes. Interference contributes to both, so their ratio for a standing or confined mode follows from the integrated fields (Remark~\ref{rem:duality-pointwise}).

\subsection{Relation to Metric Perturbations}
\label{subsec:Psi-h-relation}

The spectral weighting follows from the relation between radiative curvature and metric perturbations. In the radiation zone of asymptotically flat spacetimes,
\[
    \Psi_4 = -\ddot{h}_+ + i \ddot{h}_\times = -\ddot{h}, \quad h = h_+ - i h_\times,
\]
so for a mode \(\propto e^{-i \omega t}\),
\[
    \Psi_4 \sim \omega^2 h.
\]
More generally, radiative Weyl curvature contains two derivatives of the metric perturbation, together with background-curvature terms. In geometric optics the derivative terms dominate,
\begin{equation}
 \Psi_{\rm rad}^{(k)}\sim\omega_k^2h_k,
 \label{eq:Psi-rad-omega2-h}
\end{equation}
and $\mathcal Z_{B,k}\sim\omega_k^4|h_k|^2$, whereas $W_k\sim\omega_k^2|h_k|^2$. Near a black hole, neither the derivative terms nor their relation to a stationary energy are controlled by the frequency alone. We therefore establish the ratio in the radiation zone and treat its strong-field extension separately.

The mixed functional $\mathcal H=\int\delta E_{ab}\delta B^{ab}\,dV$ has the same dimensions but is sign-indefinite. It vanishes for a pure type-N wave, including a circularly polarized one, and is not a helicity charge.

\begin{proposition}[Gravitational enstrophy spectral weighting]
\label{thm:spectral-weighting}
Let $h_{ij}^{\rm TT}$ describe an outgoing, nearly monochromatic perturbation on a shell $\mathcal D$ in the asymptotically flat radiation zone. Evaluate curvature and energy in the same asymptotically inertial observer frame, and use the same phase average. To leading radiative order define
\begin{equation}
\begin{aligned}
 \delta\mathcal Z_k[\mathcal D]&=\int_{\mathcal D}\langle\delta B_{ij}\delta B^{ij}\rangle\,dV,\\
 W_k[\mathcal D]&=8\pi G\int_{\mathcal D}\rho_{{\rm GW},k}\,dV,\\
 \rho_{{\rm GW},k}&=\frac{1}{32\pi G}\langle\dot h_{ij}^{\rm TT}\dot h_{\rm TT}^{ij}\rangle.
\end{aligned}
\label{eq:prop1-defs}
\end{equation}
Then, for real frequency $\omega_k$,
\begin{equation}
 \eta_k\equiv\frac{\delta\mathcal Z_k[\mathcal D]}{\omega_k^2W_k[\mathcal D]}\simeq1.
 \label{eq:prop1-eta}
\end{equation}
For complex-frequency envelopes the weight is $|\omega_k^{\rm QNM}|^2$.
\end{proposition}

\begin{proof}
For a transverse outgoing wave, $\delta E_{ij}=-\ddot h_{ij}^{\rm TT}/2$ and $\delta B^2=\delta E^2$. Consequently,
\begin{equation}
\begin{aligned}
 \delta\mathcal Z_k&=\frac14\int_{\mathcal D}\langle\ddot h_{ij}^{\rm TT}\ddot h_{\rm TT}^{ij}\rangle\,dV,\\
 W_k&=\frac14\int_{\mathcal D}\langle\dot h_{ij}^{\rm TT}\dot h_{\rm TT}^{ij}\rangle\,dV.
\end{aligned}
\label{eq:prop1-integrals}
\end{equation}
Each additional time derivative contributes $\omega_k^2$ to the real-frequency quadratic norm, or $|\omega_k^{\rm QNM}|^2$ to the complex envelope. The common spatial profile cancels. The type-N norm identity is algebraic and does not require a type-D background.

\end{proof}

\begin{remark}[Beyond the radiation zone]
\label{rem:full-exterior-eta}
Beyond the radiation zone we can write $\delta\mathcal Z_k/W_k=c_k\omega_k^2$. Here $c_k$ is the ratio obtained by calculating the curvature and energy norms on the chosen domain. For ZDMs, background curvature can change these coefficients by order unity. The Fj{\o}rtoft argument will still apply if $c_j\omega_j^2<c_k\omega_k^2$ for receivers ordered as $\omega_j<\omega_k$. Whether this holds across the near-horizon spectrum is the question examined in \cref{app:spectral-ratio-proof}; coherent bands are treated in \cref{app:degenerate-fjortoft}.
\end{remark}

\begin{remark}[Normalization, scope, and choice of energy]
\label{rem:energy-robustness}
The energy in Proposition~\ref{thm:spectral-weighting} is calculated from the metric perturbation through the Isaacson tensor. Defining it instead as $\delta\mathcal Z_k/\omega_k^2$ would make the weighting relation an identity. This observer energy is positive, whereas a stationary Killing or canonical energy involves a different contraction. In geometric optics, the Killing energy contains $\int t^{\rm GW}_{ab}u^a\xi^b\,dV$, so lapse and shift enter the comparison and positivity can fail in a rotating exterior. The choice of a positive energy is part of (C1) below.

For Kerr--AdS, the radiation-zone calculation applies in an intermediate region $M\ll r\ll\ell_{\rm AdS}$ when one exists, or locally in geometric optics. Global AdS mode weights depend on the full profiles and boundary conditions.
\end{remark}

\Cref{fig:spectral-ratio-cascade} illustrates this radiation-zone weight and its role in the two-receiver algebra.

\begin{figure*}[t]
\centering
\includegraphics[width=\textwidth]{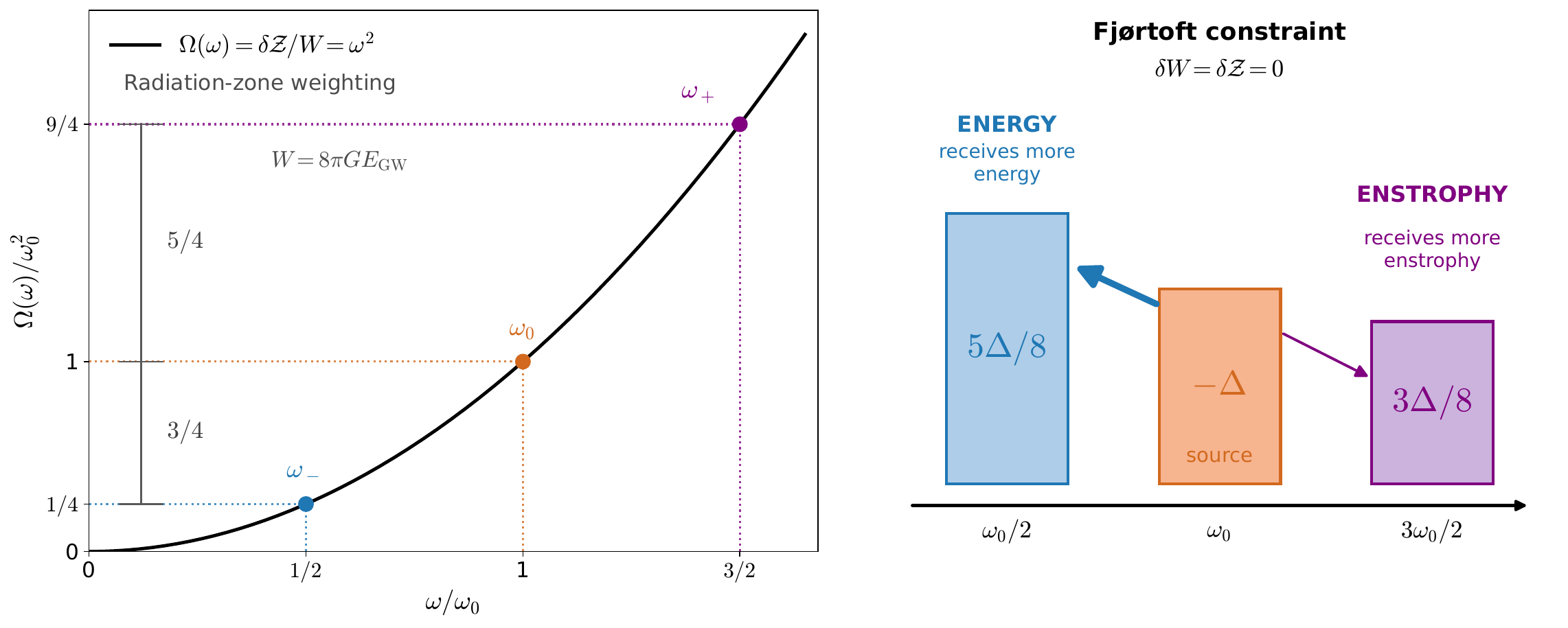}
\caption{\emph{Left}: Radiation-zone spectral weight $\Omega(\omega)=\delta\mathcal Z/W=\omega^2$ at leading order for real-frequency waves, with $W=8\pi G E_{\rm GW}$. The marked frequencies $\omega_0/2$, $\omega_0$ and $3\omega_0/2$ have weights $\omega_0^2/4$, $\omega_0^2$ and $9\omega_0^2/4$. The brackets mark the weight separations $3\omega_0^2/4$ and $5\omega_0^2/4$. \emph{Right}: A transfer conserving $W$ and $\mathcal Z$ from $\omega_0$ to these two receivers. Their unequal weight separations give $\delta W_-/\delta W_+=5/3$, so the receivers acquire $5\Delta/8$ and $3\Delta/8$, respectively. The lower-frequency side receives more energy in this symmetric-frequency example.}

\label{fig:spectral-ratio-cascade}
\end{figure*}

\section{Evolution of the Magnetic Weyl Functional}
\label{sec:evolution}

The evolution equation for the magnetic Weyl functional follows directly from the Bianchi identities (we address observer dependence in Sec. \ref{subsec:observer-dependence}.)
The divergence structure of the Bel--Robinson tensor and super-energy framework were developed in \cite{BelRobinson,Senovilla2000}, with orthogonal splittings and Poynting-type laws in \cite{GarciaParrado2008}. The 3+1 evolution of \(E_{ab}\) and \(B_{ab}\) and the super-Poynting theorem appear in \cite{Maartens1998}, and \(E/B\) dynamics has been used to visualize curvature transport in mergers via the tendex/vortex program \cite{NicholsOwenEtAl2011}. This section differs by isolating the magnetic norm \(B_{ab}B^{ab}\) alone, not the parity-even super-energy \(E^2 + B^2\). In the latter's balance law, curl-exchange terms combine into a super-Poynting flux divergence and never appear as bulk sources; for \(B^2\) alone, curl exchange becomes a potential bulk obstruction to conservation, which can be combined with the electric balance when the two norms agree (\cref{subsec:curl-exchange}). Conservation then depends on the remaining sources and fluxes.

To render the enstrophy behavior analytically tractable, we work at leading nontrivial order in perturbation theory,
\[
 E_{ab}=\bar E_{ab}+\delta E_{ab}+O(\epsilon^2),\qquad
 B_{ab}=\bar B_{ab}+\delta B_{ab}+O(\epsilon^2),
\]
where $\delta E_{ab},\delta B_{ab}=O(\epsilon)$. The relevant object is the quadratic norm of the projected radiative curvature on the background geometry. Excluding stationary background curvature from that definition does not remove its couplings from the linearized field equations.

\subsection*{The Quadratic Balance Law}
\label{subsec:quadratic-balance}

Recall the magnetic Weyl equation has the local form
\begin{equation}
    \dot{B}_{\langle ab \rangle} + (\curl E)_{ab} = \mathcal{K}^B_{ab}[E,B;u],
    \label{eq:B-evol-K}
\end{equation}
where \(\mathcal{K}^B_{ab}\) contains local kinematical couplings to \(u^a\). Linearization also acts on the derivative, projector and observer. Schematically,
\[
 \delta(\nabla C)=\bar\nabla\delta C+(\delta\nabla)\bar C.
\]
For projected radiative fields it is useful to retain these terms explicitly:
\[
 \bar{\dot{\delta B}}_{\langle ab\rangle}
 + (\overline{\curl}\,\delta E)_{ab}
 =\mathcal K^B_{\rm bg}[\delta E,\delta B]_{ab}+\mathcal J^B_{ab}.
\]
Here $\mathcal J^B$ includes variations multiplying background curvature and terms from the radiative projection. It is first order and contributes at quadratic order to the norm balance. On a fixed coordinate domain,
\begin{multline}
 \frac{d}{dt}\delta\mathcal Z
 =\int_\Sigma\alpha\left[-2\delta B^{ab}(\overline{\curl}\,\delta E)_{ab}
 +2\delta B^{ab}\mathcal K^B_{{\rm bg},ab}\right.\\
 \left.+2\delta B^{ab}\mathcal J^B_{ab}
 +\bar\theta\,\delta B^2\right]dV_{\bar\gamma}
 +\oint_{\partial\Sigma}\beta^as_a\delta B^2\,dA
 +O(\epsilon^3).
 \label{eq:deltaZ-balance-raw}
\end{multline}
This separates curl exchange, background kinematics, linearization and projection sources, and transport through the boundary. We examine the algebraic terms first.

\subsection{Algebraic Cancellation of Vorticity Coupling}
\label{subsec:vorticity-cancellation}

A key term in \(\mathcal{K}^B_{ab}\) is the vorticity coupling:
\[
    \mathcal{K}^B_{ab} \supset B_{c\langle a} \omega_{b\rangle}{}^c, \quad \omega_{ab} = D_{[a} u_{b]}.
\]
At quadratic order, this contributes
\[
    I_\omega[B] = \int_\Sigma \delta B^{ab} \delta B_{c\langle a} \omega_{b\rangle}{}^c \sqrt{\gamma} \, d^3x.
\]
Using the symmetry of \(\delta B_{ab}\), we have
\[
    \delta B^{ab} \delta B_{c(a} \omega_{b)}{}^c = \delta B^{ab} \delta B_{ca} \omega_b{}^c = T_{bc} \omega^{bc},
\]
where \(T_{bc} = \delta B^a{}_b \delta B_{ac}\) is symmetric. Since \(\omega^{bc}\) is antisymmetric, \(T_{bc} \omega^{bc} = 0\) pointwise, so \(I_\omega[B] = 0\).
The same cancellation applies to the vorticity term in the electric Weyl equation. This cancellation is local and independent of the background (Kerr, flat, AdS) or boundary conditions. It relies only on the \emph{symmetry of \(E_{ab}, B_{ab}\)} and the \emph{antisymmetry of \(\omega_{ab}\)}.

\subsection{Curl Exchange}
\label{subsec:curl-exchange}

The curl terms combine into a boundary flux in the sum of the electric and magnetic balances. For spatial STF tensors,
\begin{multline}
 \int_\Sigma S^{ab}(\curl T)_{ab}\,dV
 =\int_\Sigma T^{ab}(\curl S)_{ab}\,dV\\
 +\oint_{\partial\Sigma}s^c\epsilon_{cda}S^{ab}T_b{}^d\,dA.
\end{multline}
With $S=\alpha\delta B$, differentiation of the lapse also gives a bulk term. Keeping this, volume transport and the linearized sources, the two balances give
\begin{equation}
 \frac{d}{dt}\delta\mathcal Z
 =-\frac{d}{dt}\delta\mathcal Q_E+\mathcal S_{\rm loc}
 +\mathcal F_{\partial\Sigma}+O(\epsilon^3),
 \label{eq:Z-QE-balance}
\end{equation}
where $\delta\mathcal Q_E=\int\delta E^2dV_{\bar\gamma}$. If $\delta\mathcal Q_E=\delta\mathcal Z$, this becomes $2\dot{\delta\mathcal Z}=\mathcal S_{\rm loc}+\mathcal F_{\partial\Sigma}$; it does not set either side to zero. In particular, the magnetic norm then carries half the flux of $E^2+B^2$ (\cref{app:boundary-fluxes}).

\begin{remark}[Pointwise duality for superpositions]
\label{rem:duality-pointwise}
For a superposition of modes with frequencies \(\omega_k\) {propagating in a common direction (all outgoing, say)}, one may, for illustration, choose circularly polarized amplitudes:
\begin{gather*}
    e_+(t) = \sum_k A_k \cos(\omega_k t + \phi_k), \\
    e_\times(t) = \sum_k A_k \sin(\omega_k t + \phi_k).
\end{gather*}
The duality map \(J: (e_+, e_\times) \mapsto (e_\times, -e_+)\) acts as a norm-preserving \(90^\circ\) rotation in polarization space. Thus:
\[
    b_+^2 + b_\times^2 = e_\times^2 + e_+^2 = e_+^2 + e_\times^2
\]
\emph{at every instant}, for any superposition, with no time-averaging. Cross-mode terms oscillate at beat frequencies \(|\omega_j - \omega_k|\) but are \emph{identically equal} in the \(E\) and \(B\) sectors because \(J\) commutes with time evolution. The equality of densities gives equal integrated norms and their time derivatives on the same slice, without time averaging. Equation~\eqref{eq:Z-QE-balance} then gives the balance of either norm with the factor of two just described. {Waves with distinct propagation directions, such as ingoing and outgoing components, have different polarization rotations. Their interference can therefore give unequal electric and magnetic norms (\cref{app:np-stf}).}

For complex-amplitude norms on an axisymmetric domain, modes with different \(m\) have cross-mode terms of the form \(\overline{\delta B}_{ab}^{(k_1)}\delta B^{ab}_{(k_2)}\), which vanish exactly after azimuthal integration; for the same \(m\) and different \(\ell\) they vanish approximately, when the relevant angular inner products are approximately orthogonal. Exact scalar-harmonic orthogonality at common real \(a\omega\) does not by itself establish orthogonality of reconstructed tensor norms. For the same \((\ell, m)\) and different \(n\), radial orthogonality does not hold (the Teukolsky operator is non-normal~\cite{chandrasekhar1983}), and the cross terms vanish only on time-averaging over \(\gg |\omega_k - \omega_{k'}|^{-1}\).

For near-extremal ZDMs this averaging is not available on the cascade timescale. Since \(\omega_{\ell m n} \simeq m\Omega_H - \kappa_+[c_{\ell m} + i(n+\tfrac12)]\) (\cref{app:zdm-triad-resonance}), the real parts within an \(m\)-band differ by at most \(O(\kappa_+)\), and dephasing within \(\tau_{\rm nl} \sim (\epsilon^2\omega_R)^{-1}\) requires \(\epsilon^2 \ll \kappa_+/\omega_R \sim \sqrt{1-\chi}\) --- the opposite of the onset condition \cref{eq:epsilon-chi-condition}. We therefore retain each $m$-band as a quadratic block. Write $W_m=a_m^\dagger G_ma_m$ and $\delta\mathcal Z_m=a_m^\dagger H_ma_m$. Azimuthal integration separates different $m$. Within each band, the generalized eigenvalues of $(H_m,G_m)$ determine the curvature-to-energy weights. Comparing them with $(m\Omega_H)^2$ would test the common-weight approximation and the ordering between bands (\cref{app:degenerate-fjortoft}).

\end{remark}

\subsection{Shear, Expansion, and Non-Radiative Sectors}
\label{subsec:shear-and-coulomb}

After vorticity cancellation and curl-exchange rewriting, the local terms at quadratic order include:
\[
    \mathcal{S}_{\rm loc} = \Sigma_\sigma + \Sigma_{\rm exp} + \Sigma_{\rm acc} + \Sigma_{\rm Coul} + \cdots,
\]
where \(\Sigma_\sigma\) (shear), \(\Sigma_{\rm exp}\) (expansion), \(\Sigma_{\rm acc}\) (acceleration), and \(\Sigma_{\rm Coul}\) (Coulomb) are the remaining kinematical and non-radiative couplings. Their size depends on the perturbation and the observer. Linearizing the Bianchi equations and projecting onto the radiative sector also produces source terms, discussed below.

{\em Expansion and vorticity for ZAMOs.}
For ZAMO observers in stationary axisymmetric spacetimes, the congruence is hypersurface-orthogonal (\(u_a = -\alpha \nabla_a t\)), so vorticity vanishes by the Frobenius theorem~\cite{WaldGR}. The expansion also vanishes, by stationarity and axisymmetry together: with \(u^\mu = \alpha^{-1}(\delta^\mu_t + \Omega_{\rm fd}\, \delta^\mu_\phi)\) and \(\sqrt{-g} = \alpha\sqrt{\gamma}\),
\[
    \theta = \frac{1}{\sqrt{-g}}\,\partial_\mu\!\left(\sqrt{-g}\, u^\mu\right)
    = \frac{\partial_t \sqrt{\gamma} + \partial_\phi\!\left(\sqrt{\gamma}\,\Omega_{\rm fd}\right)}{\sqrt{-g}} = 0
\]
(\cref{app:zamo-shear}). Thus, \(\Sigma_{\rm exp} = 0\) and vorticity contributions are already zero.

{\em ZAMO shear.}
The ZAMO shear is not universally zero in Kerr or Kerr--AdS. Its nonzero components in the orthonormal frame are:
\[
    \sigma_{\hat{r}\hat{\phi}}^{\rm ZAMO} = \frac{\sqrt{\gamma_{\phi\phi}}}{2\alpha \sqrt{\gamma_{rr}}} \partial_r \Omega_{\rm fd}, \quad
    \sigma_{\hat{\theta}\hat{\phi}}^{\rm ZAMO} = \frac{\sqrt{\gamma_{\phi\phi}}}{2\alpha \sqrt{\gamma_{\theta\theta}}} \partial_\theta \Omega_{\rm fd},
\]
where \(\Omega_{\rm fd} = -g_{t\phi}/g_{\phi\phi}\) is the frame-dragging angular velocity.

{To determine which components couple to \(\delta\mathcal{Z}\), notice that a transverse radiative \(\delta B_{ab}\) is a symmetric trace-free \(2\times2\) matrix on the \((\hat\theta,\hat\phi)\) plane. By Cayley--Hamilton,
\[
    \delta B_a{}^c\,\delta B_{bc} = \tfrac12\,(\delta B_{cd}\delta B^{cd})\,(\gamma_{ab} - \hat r_a \hat r_b),
\]
so for any background shear
\begin{align*}
 \delta B^{ab}\bar\sigma^c{}_a\delta B_{bc}
 &=\tfrac12\delta B^2(\bar\sigma^a{}_a-\bar\sigma_{\hat r\hat r})\\
 &=-\tfrac12\bar\sigma_{\hat r\hat r}\delta B^2.
\end{align*}
Only the radial-radial shear couples; the transverse-transverse (\(s_+, s_\times\)) and radial-transverse components drop out identically. For ZAMOs \(\bar\sigma_{\hat r\hat r} = K_{\hat r\hat r} = 0\) (\cref{app:zamo-extrinsic}), so \(\Sigma_\sigma=0\) for a field transverse in that frame.}

The shear source in the enstrophy balance is (\cref{app:zamo-shear}):
\[
    \delta\mathcal{S}_\sigma^B \propto \int_{\mathcal{D}_t}\alpha\left[ \sigma_{\hat{r}\hat{\phi}}^{\rm ZAMO} (\delta B^2)_{\hat{r}\hat{\phi}} + \sigma_{\hat{\theta}\hat{\phi}}^{\rm ZAMO} (\delta B^2)_{\hat{\theta}\hat{\phi}} \right] \sqrt{\bar{\gamma}} \, d^3x.
\]

\begin{proposition}[Shear decoupling for transverse radiative modes]
\label{prop:shear-decoupling}
Suppose the magnetic Weyl perturbation is transverse and trace-free in the same orthonormal frame in which $\sigma_{\hat r\hat r}=0$, as for a field transverse in the ZAMO frame. Then

\[
    \delta B_{\hat{r}\hat{i}}^{\rm rad} = 0, \quad \hat{i} \in \{\hat{\theta}, \hat{\phi}\}.
\]
and as a consequence:
\[
    (\delta B^2)_{\hat{r}\hat{\phi}} = \delta B^{\hat{a}}{}_{\hat{r}} \delta B_{\hat{a}\hat{\phi}} = 0, \quad
    {(\delta B^2)_{\hat{\theta}\hat{\phi}} = \delta B^{\hat{a}}{}_{\hat{\theta}} \delta B_{\hat{a}\hat{\phi}} = 0,}
\]
{The second equality follows from the Cayley--Hamilton identity, which implies $(\delta B^2)_{ab} \propto \gamma_{ab} - \hat{r}_a \hat{r}_b$ on this subspace. As a result, the shear source vanishes pointwise: $\delta\mathcal{S}_\sigma^B\big|_{\rm rad} = 0$. This holds for any superposition of ingoing and outgoing transverse radiation, even without time averaging.}
\end{proposition}

{In the ZAMO frame, the radiative field is transverse only up to an $O(v)$ azimuthal boost (see \cref{eq:carter-zamo-boost}). This boost introduces radial-transverse components of order $v$ into $\delta B_{ab}$. The residual shear coupling is then $O(v \bar{\sigma}_{\hat{r}\hat{\phi}}) \delta\mathcal{Z}$, where $v \lesssim 0.2\chi$. We estimate this coupling in \cref{app:shear-diagnostic}. As discussed in \cref{subsec:enstrophy-balance}, for a single stationary mode it is one of the terms whose sum must reproduce the known quadratic decay.}

For the Fj{\o}rtoft argument the complete linear change must be small relative to the transfer. Individual shear and acceleration terms can be large while their sum with the other terms remains small. We measure these contributions with the dimensionless diagnostic of \cref{app:zamo-shear}.

{\em Coulomb sector.}
The scalar $\delta\Psi_2$ contains stationary mass and spin perturbations, but a general radiative solution also has a first-order $\delta\Psi_2$. Since $\bar\Psi_2\ne0$, it changes by $\mathcal L_\xi\bar\Psi_2$ under an infinitesimal coordinate-gauge transformation. So $\delta\Psi_2$ is generally first order in $\epsilon$ and gauge dependent. Under the usual regularity assumptions, its stationary Kerr completion can be absorbed into the background parameters~\cite{Wald1973}, leaving the dynamical part of $\delta\Psi_2$.

The radiative projection defines
\[
 \delta\mathcal Z^{\rm rad}=\int_\Sigma
 \delta B^{\rm rad}_{ab}\delta B_{\rm rad}^{ab}\,dV_{\bar\gamma}.
\]
Orthogonality of the STF basis separates this norm from the $\Psi_2$ component in the chosen frame. The evolution still couples them: the projected Bianchi equations contain $\mathcal J^{E,B}$ and terms describing exchange with the other components. The condition for neglecting this exchange is given in \cref{app:coulomb-balance}. An $O(\epsilon^2)$ drift of the background parameters can also shift the mode frequencies by amounts comparable to the $O(\kappa_+)$ splittings of a ZDM band, so before treating the weights as fixed, one has to check that the induced drift is small compared with the splittings.

\subsection{Perturbative Enstrophy Balance}
\label{subsec:enstrophy-balance}

The balance law therefore becomes:
\begin{equation}
    \frac{d}{dt} \delta\mathcal{Z} = \mathcal{R}_{\rm bulk} + \mathcal{F}_{\partial\Sigma} + O(\epsilon^3),
    \label{eq:Z-final-balance}
\end{equation}
where:
\begin{itemize}
    \item \(\mathcal{R}_{\rm bulk}\) contains expansion, shear, acceleration, and linearization and projection sources,
    \item \(\mathcal{F}_{\partial\Sigma}\) contains boundary fluxes (through \(\mathcal{I}^+, \mathcal{H}^+,\) or walls).
\end{itemize}
Thus, \(\delta\mathcal{Z}\) is conserved to leading perturbative order whenever:
\[
    \mathcal{R}_{\rm bulk} + \mathcal{F}_{\partial\Sigma}
    \;\;\text{negligible over } \tau_{\rm nl}.
\]
that is 
when the following conditions hold:
\begin{itemize}
 \item[(C1)] The radiative norm, observer, domain and positive energy are specified, and exchange with excluded components is accounted for or bounded.
 \item[(C2)] The complete linear balance, including bulk terms, interference and boundary fluxes, changes these quantities negligibly over $\tau_{\rm nl}$.
 \item[(C3)] The nonlinear redistribution itself preserves both moments to the accuracy needed for the transfer under consideration.
\end{itemize}
Only the first two conditions concern linear losses; (C3) is an independent dynamical premise of the Fj{\o}rtoft argument. The diagnostics $\chi_\sigma$ and $\chi_a$ check individual contributions to the balance, not its closure by themselves.

This applies directly:
\begin{itemize}
    \item In generic asymptotically flat ringdown, boundary/horizon fluxes and quasinormal damping may remove energy before nonlinear transfer matters.
    \item In near-extremal Kerr, long-lived modes allow nonlinear interactions to accumulate.
    \item In Kerr–AdS or reflecting domains, outer-boundary leakage is suppressed, improving conservation of both \(W\) and \(\delta\mathcal{Z}\).
\end{itemize}
At $O(\epsilon^2)$ the perturbation is a superposition of linear modes. The balance law quantifies the changes produced by that linear evolution; nonlinear redistribution enters the $O(\epsilon^3)$ remainder, whose role in (C3) is examined in \cref{subsec:fjortoft-interpretation}.

At quadratic order there is also a simple check. On a stationary background, with a stationary observer field and slicing, time translation maps a linear mode $\propto e^{-i\omega_k t}$ into itself, so that on any fixed region $\mathcal{D}$
\[
\delta\mathcal{Z}_k[\mathcal{D}](t) = \delta\mathcal{Z}_k[\mathcal{D}](0)\, e^{2\omega_{I,k} t}, \qquad \omega_{I,k} = {\rm Im}\,\omega_k \le 0,
\]
and similarly for $W_k$. Whatever their individual size, the bulk terms in \cref{eq:Z-final-balance} must combine with the boundary fluxes to reproduce this decay (plus beat terms for superpositions), and so they cannot act as an independent net source. For an isolated mode this is a consistency check on the decomposition, and condition (C2) reduces to $\tau_{\rm damp} \gg \tau_{\rm nl}$ together with the phase averaging of Remark~\ref{rem:coarse-grained-WZ}. For quasinormal modes this statement requires a finite norm: for example a fixed bounded exterior region, or a suitable horizon-regular and asymptotically controlled prescription. Horizon penetration alone does not remove the outgoing QNM growth at spatial infinity (\cref{app:acceleration-convergence}).

Before turning to the nonlinear redistribution, however, we must comment on the issue of observer dependence.

\subsection{Observer Dependence}
\label{subsec:observer-dependence}
As noted, the magnetic Weyl tensor \(B_{ab} = -\frac{1}{2} \epsilon_{acd} C^{cd}{}_{be} u^e\) depends on the observer \(u^a\). 
Under an infinitesimal boost \(u'^a = u^a + v^a + O(v^2)\) (where \(v^a\) is spacelike, orthogonal to \(u^a\), and \(|v| \ll 1\)), the electric and magnetic Weyl tensors mix at first order:
\begin{gather*}
    \delta_v B_{ab} = 2\epsilon_{cd(a} E_{b)}{}^c v^d + O(v^2), \\
    \delta_v E_{ab} = -2\epsilon_{cd(a} B_{b)}{}^c v^d + O(v^2).
\end{gather*}
These component formulas use projection back into the original rest space. For a radiative field \(\delta B_{ab}=O(\epsilon)\), a boost of order \(v\) changes \(\delta B_{ab}\) by \(O(v \epsilon)\), and thus:
\[
    \delta_v (\delta\mathcal{Z}) = 2 \int_\Sigma \delta B^{ab} \delta_v B_{ab} \sqrt{\gamma} \, d^3x = O(v \epsilon^2).
\]
The fractional change is therefore \(O(v)\), so \(\delta\mathcal{Z}\) is observer-dependent at order \(v\). This differs from the fluid case, where the vorticity \(\boldsymbol{\omega} = \nabla \times \mathbf{u}\) is Galilean-invariant and so is the fluid enstrophy \(\int |\boldsymbol{\omega}|^2 d^2x\); only the kinetic energy changes under a boost. For the Fj{\o}rtoft argument, the question is how the spectral ratio and frequency ordering transform.

\subsubsection*{Spectral ratio under a boost}

For a plane wave, let $k_a$ be its wave covector and $D=(-k_au'^a)/(-k_au^a)$ the Doppler factor. The measured frequency is $\omega'=D\omega$. The Isaacson energy density acquires a factor $D^2$, while the squared magnetic Weyl tensor acquires $D^4$; equivalently, $|\Psi'_4|^2=D^4|\Psi_4|^2$ in a tetrad aligned with the transformed propagation direction. Thus, using the same phase average and a common spatial measure for the two norms,
\begin{equation}
    \frac{\delta\mathcal Z'_k}{W'_k}=(\omega'_k)^2.
    \label{eq:ratio-boost}
\end{equation}
The frequency-squared relation is exact for the plane wave in either frame. For $v\ll1$, its numerical value changes by $O(v)$.

The ordering of radiative weights follows the ordering of the measured frequencies. For co-propagating modes, which share the Doppler factor, this ordering is preserved exactly; in general it is preserved whenever \(v \ll |\omega_j - \omega_k|/\omega_k\). The superradiant threshold $\omega=m\Omega_H$ is fixed by the horizon generator $\chi^a=\xi^a+\Omega_H\psi^a$. The Killing frequency $\omega-m\Omega_H$ vanishes there independently of the local observer. We discuss its possible dynamical role in \cref{subsec:cascade-attractor}.

What does change with the observer, at order \(v\), are the absolute values of \(\delta\mathcal{Z}\) and \(W\) (and hence the numerical ratio at each frequency), the shear correction \(\delta\mathcal{S}_\sigma^B\), the acceleration correction \(\Sigma_{\rm acc}\) (which vanishes for geodesic observers), and the numerical \((\epsilon,\chi)\) threshold for cascade onset, since \(\tau_{\rm nl}\) involves the Doppler-shifted \(\omega_R\). In a nearly degenerate spectrum, even small relative changes must be compared with the weight gaps.

Among stationary observers in Kerr, ZAMOs are the convenient choice for displaying the bulk terms: they are the unique stationary congruence with zero angular momentum (\(g_{ab} u^a \psi^b = 0\)), their vorticity and expansion vanish (\cref{app:zamo-shear}), and their shear couples to transverse radiation only through the \(O(v)\) misalignment of \cref{eq:carter-zamo-boost}. The residual contributes to the linear balance through the term estimated in \cref{app:shear-diagnostic}.

The ZAMO-slice norm of a horizon-reaching QNM is generally divergent for its ingoing part (\cref{app:acceleration-convergence}). A bounded exterior domain avoids that endpoint; an integral reaching the horizon needs a regular observer and slicing, as well as a prescription controlling the outgoing behavior at infinity. Hypersurface orthogonality preserves the vorticity cancellation, but the remaining terms must be recomputed.

The observer and slicing are therefore part of the definition of the functional. For co-propagating radiation the frequency ordering survives a change of frame; Remark~\ref{rem:full-exterior-eta} states the additional condition needed in the strong-field region.

\section{The gravitational Fj{\o}rtoft constraint}
\label{sec:fjortoft}

We now consider two quadratic quantities that satisfy the conservation assumptions just stated, the gravitational-wave energy
\begin{align}
    W = \sum_k W_k,
\end{align}
and the magnetic Weyl enstrophy
\begin{align}
    \mathcal Z
    =
    \sum_k \mathcal Z_k
    \simeq
    \sum_k \omega_k^2 W_k\, ;
\end{align}
Their different spectral weights constrain the direction of nonlinear energy transfer, provided both are conserved by the interactions; this is the gravitational analogue of Fj{\o}rtoft's argument in 2D turbulence. The argument is kinematic: it restricts the direction of transfer. Rates and spectral slopes depend on the mode couplings.

\subsection{Setup}
\label{subsec:fjortoft-setup}

Let $k$ label modes (e.g., $(\ell,m,n)$ in black-hole perturbation 
theory). Define the spectral weight
\begin{align}
    \Omega_k = \omega_k^2 \geq 0.
\end{align}
The total energy and enstrophy are then given by
\begin{align}
    W = \sum_k W_k, \qquad
    \mathcal Z = \sum_k \Omega_k W_k.
\end{align}

The argument does not require $\Omega_k = \omega_k^2$ exactly. Effective 
weights $\Omega_k^{\rm eff} = c_k \omega_k^2$ with $c_k > 0$ work as 
long as they preserve the frequency ordering 
(Appendix~\ref{app:degenerate-fjortoft}). Mode-dependent constants 
$c_k$ rescale quantitative ratios but not the direction of the 
constraint. We assume a positive modal energy $W_k\geq0$ on the chosen sector; 
superradiant extraction requires separate treatment 
(\cref{subsec:cascade-attractor}).

Now, if we assume that on the nonlinear transfer timescale $\tau_{\rm nl}$,
\begin{align}
    \dot W \simeq 0, \qquad
    \dot{\mathcal Z} \simeq 0\, ,
    \label{eq:approx-conservation-WZ}
\end{align}
which requires damping, boundary fluxes, and non-radiative corrections to 
be small compared with the nonlinear redistribution among modes; the redistribution itself must also conserve both moments, a condition on the interactions themselves taken up in \cref{subsec:fjortoft-interpretation}. Then 
any redistribution $\delta W_k$ must satisfy
\begin{align}
    \sum_k \delta W_k = 0, \qquad
    \sum_k \Omega_k \, \delta W_k = 0.
\end{align}

Energy conservation together with enstrophy conservation implies the existence of a second constraint.
Take energy initially at $\omega_0^2=\Omega_0$. Suppose $\Delta > 0$ leaves and 
splits between a lower band $\omega_-$ and a higher band $\omega_+$ (with frequencies equidistant to $\omega_0$) and, for simplicity,
assume $c_k=1$:
\begin{align}
    \delta W_0 = -\Delta, \qquad
    \delta W_- + \delta W_+ = \Delta.
\end{align}
Enstrophy conservation gives
\begin{align}
    \Omega_-\delta W_- + \Omega_+\delta W_+ = \Omega_0\Delta.
\end{align}
Solving,
\begin{align}
    \delta W_- = \Delta \frac{\Omega_+-\Omega_0}{\Omega_+-\Omega_-}, 
    \qquad
    \delta W_+ = \Delta \frac{\Omega_0-\Omega_-}{\Omega_+-\Omega_-}.
\end{align}
Both are positive, thus a pure transfer to higher frequencies is impossible.
Furthermore, the ratio
\begin{align}
    \frac{\delta W_+}{\delta W_-}
    =
    \frac{\Omega_0-\Omega_-}{\Omega_+-\Omega_0}
    \label{eq:W-ratio-fjortoft}
    = \frac{(\omega_0+\omega_-)}{(\omega_0+\omega_+)} < 1
\end{align}
shows that more energy goes to lower frequencies. 

\subsection{General case}
\label{subsec:general-fjortoft}

\begin{theorem}[Gravitational Fj{\o}rtoft constraint]
\label{thm:fjortoft}
Let $W = \sum_k W_k > 0$ and $\mathcal Z = \sum_k \Omega_k W_k > 0$ 
be the perturbative energy and magnetic Weyl enstrophy of a collection 
of modes or quadratic blocks, with $W_k\geq0$ and ordered weights. Suppose both are conserved to the required transfer accuracy on $\tau_{\rm nl}$ (conditions (C1)--(C3) of 
\cref{sec:evolution}). Consider a redistribution $\{\delta W_k\}$ with 
$\sum_k\delta W_k=0$, $\sum_k\Omega_k\delta W_k=0$, where $\Delta>0$ 
leaves a source block of weight $\Omega_0$ and is received by other blocks with $\delta W_k\geq0$. Assume the receivers have $\Omega_k\ne\Omega_0$; transfer within an equal-weight block is unconstrained.

Define
\begin{equation}
\begin{gathered}
    W_\mp
    =
    \!\!\sum_{\Omega_k\lessgtr\Omega_0}\!\!\delta W_k,
    \qquad
    \bar d_-
    =
    \frac{1}{W_-}\!\!\sum_{\Omega_k<\Omega_0}\!\!(\Omega_0-\Omega_k)\,\delta W_k, \\
    \bar d_+
    =
    \frac{1}{W_+}\!\!\sum_{\Omega_k>\Omega_0}\!\!(\Omega_k-\Omega_0)\,\delta W_k.
\end{gathered}
    \label{eq:fjortoft-Wpm-dpm}
\end{equation}
Then:
\begin{enumerate}
\item The receiving modes cannot lie entirely above or below $\Omega_0$.
\item $\langle\Omega\rangle_{\rm transfer} = \Omega_0$.
\item $\bar d_- W_- = \bar d_+ W_+$.
\end{enumerate}
Thus $W_->W_+$ iff $\bar d_+>\bar d_-$.
\end{theorem}

\begin{proof}
The receiving modes satisfy $\sum_{k\neq0}\delta W_k=\Delta$ and 
$\sum_{k\neq0}\Omega_k\delta W_k=\Omega_0\Delta$, giving the centroid 
identity. Two-sidedness follows by contradiction. The distance balance 
comes from subtracting $\Omega_0$ times energy conservation from 
enstrophy conservation.
\end{proof}

\begin{corollary}[Spectral locality]
\label{cor:spectral-locality}
If $\bar d_+ \geq \bar d_-$, then $W_- \geq W_+$: more energy goes 
to lower frequencies than to higher ones.
\end{corollary}

\begin{corollary}[One-sided constraint]
\label{cor:one-sided}
Let $\sum_k \delta W_k = 0$ and $\sum_k \Omega_k\,\delta W_k \le 0$. Purely direct transfer remains impossible: if $\Delta > 0$ leaves mode $\Omega_0$ and every receiver has $\Omega_k > \Omega_0$, then $\sum_k \Omega_k\,\delta W_k \ge (\min_{\rm rec}\Omega_k - \Omega_0)\,\Delta > 0$. 
Nevertheless, purely inverse transfer is allowed.
\end{corollary}

The discrete, positive-energy setting assumed here is the simplest case of the argument. Its extension to continuous spectra, together with degeneracies and normalization, open systems, and superradiant modes with indefinite $W_k$, is treated in \cref{app:fjortoft-details}.

\subsection{When does this matter?}
\label{subsec:fjortoft-interpretation}

Let us now ask when the interactions preserve $\mathcal Z$, as required by (C3). With the radiation-zone weight, the elementary three-wave and $2\leftrightarrow2$ processes generally change it. Their kinematics is enough to see why; their relative importance depends on the coupling calculation.

For a resonant three-wave process, $\omega_3=\omega_1+\omega_2$ and the Manley--Rowe relations~\cite{ManleyRowe1956} give $\delta n_1=\delta n_2=-\delta n_3$, where $W_k=\omega_kn_k$. With the radiation-zone weight,
\begin{align}
 \delta W&=0,\qquad\delta N=-\delta n_3,\nonumber\\
 \delta\mathcal Z&=(\omega_3^3-\omega_1^3-\omega_2^3)\delta n_3
 =3\omega_1\omega_2\omega_3\delta n_3.
 \label{eq:vertex-Z-change}
\end{align}
Decays decrease $\mathcal Z$ and recombinations increase it; neither conserves $N=\sum_kn_k$.

It is useful to compare this with the second-order Kerr calculations of Refs.~\cite{LoutrelRipleyGiorgiPretorius2021,RipleyLoutrelGiorgiPretorius2021}, where an $m=2$ perturbation sources an $m=4$ response, including at $\chi=0.998$. The sum-frequency part has the $2+2\to4$ structure of recombination, whose radiation-zone weight gives the positive change in \cref{eq:vertex-Z-change}. It therefore lies outside the enstrophy-conserving premise (C3). At second order the parent is still evolved linearly, so this comparison identifies the channel rather than measuring a conservative transfer with parent depletion.

When the resonant dynamics is governed by $2\leftrightarrow2$ interactions, the conserved pair is $(N,W)$. The corresponding change in $\mathcal Z=\sum_k\omega_k^3n_k$ is easy to see. For $\delta n_1=\delta n_2=-\delta n_3=-\delta n_4=\delta n$ and $\omega_1+\omega_2=\omega_3+\omega_4=s$,
\begin{equation}
\begin{aligned}
 \delta\mathcal Z
 &= (\omega_1^3+\omega_2^3-\omega_3^3-\omega_4^3)\delta n\\
 &=3s(\omega_3\omega_4-\omega_1\omega_2)\delta n.
\end{aligned}
\label{eq:four-wave-Z-change}
\end{equation}
For $2\omega_0\to(\omega_0-\delta)+(\omega_0+\delta)$, equal action $\nu$ reaches each receiver, $\delta W_\pm=(\omega_0\pm\delta)\nu$, and $\delta\mathcal Z=6\omega_0\delta^2\nu$. The higher-frequency receiver gets more energy because the conserved pair is now $(N,W)$. This is the pair underlying the direct-energy and inverse-action dual cascade of four-wave weak turbulence, including the reduced gravitational-wave model of Ref.~\cite{GaltierNazarenko2017}.\footnote{That model's periodic weak-field truncation omits part of the Einstein equations. With its assumed ordering, the averaged Hamiltonian constraint forces the positive quadratic wave energy to vanish, an instance of the classical linearization-instability obstruction on compact flat backgrounds~\cite{FischerMarsden1973,Moncrief1975}. A nontrivial periodic wave field therefore needs a different treatment of the background or of the perturbative ordering. Ref.~\cite{GayBabichevGaltierNoui2026} discusses the missing equations and numerical constraint violations, and argues for compatibility in a restricted weak-turbulence regime. Our kinematic argument does not depend on this truncation.} For general receivers, the division again depends on their spacing; the net cascade flux also depends on the dynamics. The essential point here is that interactions that conserve $(N,W)$ generally change $\mathcal Z$.

In the hydrodynamic regime, incompressible Euler evolution conserves fluid enstrophy. Through the leading-gradient map of \cref{sec:holographic-fluid}, this gives conservation of the fluid-weighted bulk norm. Viscosity introduces the enstrophy loss. In simulations, joint changes of $W$, $N$ and $\mathcal Z$ can test a proposed interaction model after external fluxes have been removed. Their signs alone do not uniquely identify mixed interaction channels.

The algebra concerns frequency transfer. To interpret it as a cascade in spatial or angular scale, one has to relate frequency to those scales.

The external losses are the second condition: they must be slow compared with the transfer. In the perturbative regime, a few cases are particularly relevant.
First, in generic Kerr ringdown, $\tau_{\rm damp}\lesssim\tau_{\rm nl}$. 
Modes decay before the cascade develops. Even if the two-moment premise holds, the putative transfer is suppressed.
Second, near-extremal Kerr is different: $\tau_{\rm damp}\sim\kappa_+^{-1}$ 
diverges. The leaked enstrophy fraction per cascade time is 
$\alpha_{\rm loss}\sim\kappa_+/(\epsilon^2\omega_R)\ll1$ 
(equation~\eqref{eq:alpha-conservation-quality}), so linear losses are small. A narrow transfer bias is more sensitive to these losses (\cref{app:damping-driving-fjortoft}). The couplings set the direction and rate of the nonlinear transfer.
Last, reflecting boundaries in AdS or cavities suppress outer leakage. With a black hole present, horizon absorption remains, though it can be weaker than in asymptotically flat ringdown for the stable, long-lived modes discussed in Sec.~\ref{subsec:kerr-ads-flux}.

Whether power accumulates near $\omega=m\Omega_H$ in Kerr (Sec.~\ref{subsec:cascade-attractor}), or in the lowest AdS normal mode, is a dynamical question.

Finally, the constraint alone does not fix inertial-range slopes; for that we turn to the mode-coupling analysis discussed in \cref{subsec:open}. In non-relativistic hydrodynamics, dimensional analysis gives the slopes within the inertial range~\cite{Kraichnan1967}. The relativistic case---where length, time, and mass share the same scale---is more delicate and demands specialized computations or limiting regimes.

\section{Flux Balances and Physical Regimes}
\label{sec:flux-regimes}

Let us examine the balance law and spectral weighting in physically relevant backgrounds. 
\rev{The key issue is whether the external losses are slow over the nonlinear transfer time, $\tau_{\rm nl}$; what the transfer itself does to $\mathcal{Z}$ was discussed in \cref{subsec:fjortoft-interpretation}.}

\subsection{General Balance Law}
\label{subsec:general-flux-balance}

\rev{Let $\mathcal{D}_t \subset \Sigma_t$ be the spatial domain on which the radiative magnetic Weyl enstrophy is measured,
\[
    \delta\mathcal{Z}_{\mathcal{D}} = \int_{\mathcal{D}_t} \delta B_{ab} \delta B^{ab} \sqrt{\gamma} \, d^3x.
\]
From \cref{sec:evolution}, the perturbative balance law reads
\[
    \frac{d}{dt} \delta\mathcal{Z}_{\mathcal{D}} = \mathcal{R}_{\rm bulk}[\delta E, \delta B; u] + \mathcal{F}_{\partial \mathcal{D}}^{\mathcal{Z}} + O(\epsilon^3),
\]
where $\mathcal{F}_{\partial\mathcal{D}}^{\mathcal{Z}}$ is the boundary flux and $\mathcal{R}_{\rm bulk}$ collects the local bulk terms that survive the vorticity cancellation and the curl-exchange rearrangement: shear, Coulomb, and observer-dependent kinematic contributions. At quadratic order the bulk terms are tied to the damping of the linear modes (\cref{subsec:enstrophy-balance}). For a stationary Killing or canonical energy with its corresponding flux,
\[
    \frac{d}{dt} W_{\mathcal{D}} = \mathcal{F}_{\partial \mathcal{D}}^{W} + \mathcal{S}_{\rm drive} - \mathcal{D}_{\rm diss} + O(\epsilon^3),
\]
with $\mathcal{S}_{\rm drive}$ an external driving and $\mathcal{D}_{\rm diss}$ dissipation or absorption; in driven systems $\mathcal{Z}$ receives the corresponding injection term (\cref{subsec:driven-domains}).

The displayed $O(\epsilon^2)$ terms describe the linear evolution; nonlinear transfer enters the $O(\epsilon^3)$ remainders. The fractional linear changes are small over a transfer time when
\[
    \left| \frac{\dot{W}_{\mathcal{D}}}{W_{\mathcal{D}}} \right| \ll \tau_{\rm nl}^{-1}, \quad \left| \frac{\dot{\mathcal{Z}}_{\mathcal{D}}}{\mathcal{Z}_{\mathcal{D}}} \right| \ll \tau_{\rm nl}^{-1},
\]
where the rates on the left are evaluated from the displayed linear terms. We now examine this timescale comparison in the different backgrounds.}

\subsection{Asymptotically Flat Kerr}
\label{subsec:kerr-flux-balance}

For Kerr, take $\mathcal{D}_t$ as the exterior between $\mathcal{H}^+$ and $\mathscr{I}^+$. The balances are:
\begin{gather*}
    \frac{d}{dt} \delta\mathcal{Z}_{\rm ext} = \mathcal{R}_{\rm bulk}^{\rm Kerr} - \mathcal{F}_{\mathcal{H}^+}^{\mathcal{Z}} - \mathcal{F}_{\mathscr{I}^+}^{\mathcal{Z}} + O(\epsilon^3), \\
    \frac{d}{dt} W_{\rm ext} = - \mathcal{F}_{\mathcal{H}^+}^{W} - \mathcal{F}_{\mathscr{I}^+}^{W} + O(\epsilon^3).
\end{gather*}
Even if $\mathcal{R}_{\rm bulk}^{\rm Kerr}$ vanishes, the integrated quantities are not conserved: radiation escapes to $\mathscr{I}^+$ and is absorbed by the horizon which can suppress non-linear interactions.
Quantitatively at linear order, perturbations decompose into quasinormal modes with $\omega = \omega_R + i \omega_I$ ($\omega_I < 0$) and quadratic-norm damping time $\tau_{\rm damp}=(2|\omega_I|)^{-1}$ (the amplitude lifetime is twice this). For the Schwarzschild fundamental $\ell = 2$ mode, $M\omega \simeq 0.3737 - 0.0890\,i$ \cite{Leaver1985}, so $\tau_{\rm damp} \simeq 5.6\,M$, while the nonlinear transfer time of \cref{eq:tau-nl-estimate} is $\tau_{\rm nl} \sim (\epsilon^2 \omega_R)^{-1} \simeq 2.7\,M/\epsilon^2$. The cascade condition $\tau_{\rm nl} \ll \tau_{\rm damp}$ then requires  $\epsilon\gtrsim0.69$ merely for comparable times --- outside the perturbative regime.\footnote{Nevertheless, nonlinear simulations in non-spinning black hole, and even flat, spacetimes~\cite{MaLehnerYangKidderPfeifferScheel2026} show such an inverse cascade. Their setup, however, realizes a steady state through a suitable driver, which maintains the source amplitude while daughter modes still compete with damping (\cref{subsec:driven-domains}).}. Moderate spin does not help much: at $\chi = 0.7$ the fundamental $\ell = m = 2$ mode has $M\omega \simeq 0.533 - 0.081\,i$, and the threshold remains nonperturbative ($\epsilon\gtrsim0.55$ for comparable times). Within the reach of the perturbative analysis presented here, perturbations decay before nonlinear transfer can organize a cascade: the enstrophy structure is locally present, but its dynamical effect is suppressed by damping and leakage, as in decaying turbulence in a strongly dissipative fluid (e.g.~\cite{ComteBellot1971}). From this vantage point, generic Kerr ringdown therefore would not exhibit a clean inverse cascade; here $\mathcal{Z}$ serves as a diagnostic, identifying the spectral constraint that becomes dynamical in the less dissipative settings that follow.

\subsection{Near-Extremal Kerr}
\label{subsec:near-extremal-kerr-flux}

Long-lived ZDMs make near-extremal Kerr a natural setting for this timescale comparison. Indeed,
for a perturbation of amplitude $\epsilon$ (dimensionless strain) and frequency $\omega_R \sim m \Omega_H$, the nonlinear coupling rate is $\sim \epsilon^2 \omega_R$, giving the nonlinear timescale
\begin{equation}
    \tau_{\rm nl} \sim \frac{1}{\epsilon^2 \omega_R} \sim \frac{1}{\epsilon^2 m \Omega_H},
    \label{eq:tau-nl-estimate}
\end{equation}
parametrically longer than the oscillation period by $\epsilon^{-2}$. This is the reference estimate for coherent transfer through a quartic interaction, taking its dimensionless coefficient to be of order unity. An active resonant cubic interaction instead gives $\tau_{\rm nl}^{-1}\sim\epsilon\omega_R$ under the same assumption. The ZDM frequency matches in \cref{app:zdm-triad-resonance} allow this possibility, but the normalized coupling sets its actual rate. For now, we give both amplitude thresholds below. These coherent transfer estimates should be distinguished from the longer statistical transfer times of random-phase wave turbulence~\cite{Nazarenko2011}.

Against this, the damping timescale of zero-damping modes lengthens without bound as $\chi \to 1$~\cite{Yang:2012pj}:
\begin{equation}
    |\omega_I| \simeq \left(n + \tfrac{1}{2}\right)\kappa_+, \qquad \kappa_+ = \frac{\sqrt{1 - \chi^2}}{2r_+} \simeq \frac{\sqrt{1 - \chi}}{M \sqrt{2}},
    \label{eq:zdm-kappa-scaling}
\end{equation}
so for the least-damped overtone,
\begin{equation}
    \tau_{\rm damp} \sim \kappa_+^{-1} \simeq \frac{M \sqrt{2}}{\sqrt{1 - \chi}}.
    \label{eq:zdm-damping-time}
\end{equation}
Combining \cref{eq:tau-nl-estimate,eq:zdm-damping-time} with $\Omega_H \simeq 1/(2M)$, the timescale separation $\tau_{\rm nl} \ll \tau_{\rm damp}$ becomes
\begin{equation}
    \epsilon^2 \gg \frac{\sqrt{2(1 - \chi)}}{m} \,;
    \label{eq:epsilon-chi-condition}
\end{equation}
for $\chi=0.99$ and $m=2$, equality of these times occurs near $\epsilon=0.27$. With the cubic estimate, equality instead occurs at $\epsilon\simeq\kappa_+/\omega_R\simeq0.071$. Both numbers assume a unit dimensionless coupling; a separation $\tau_{\rm damp}\gg\tau_{\rm nl}$ requires amplitudes above the corresponding crossover. Figure~\ref{fig:timescale-separation} uses the quartic reference rate of \cref{eq:tau-nl-estimate}.

For a fundamental ZDM, the complete linear quadratic norm has fractional decay rate $2|\omega_I|\simeq\kappa_+$. Thus
\begin{equation}
 \frac{|\dot{\delta\mathcal Z}_k|_{\rm lin}}{\delta\mathcal Z_k/\tau_{\rm nl}}
 \simeq\frac{\kappa_+}{\epsilon^2\omega_R}.
 \label{eq:horizon-Z-fraction}
\end{equation}
This estimates the complete linear change of the mode; the individual horizon term is discussed in \cref{app:horizon-flux}. Near-extremal Kerr thus supplies long-lived modes on which nonlinear transfer can act. Figure~\ref{fig:timescale-separation} displays the competition between transfer and damping for the reference rate.

\begin{figure}[ht]
\centering
\includegraphics[width=\columnwidth]{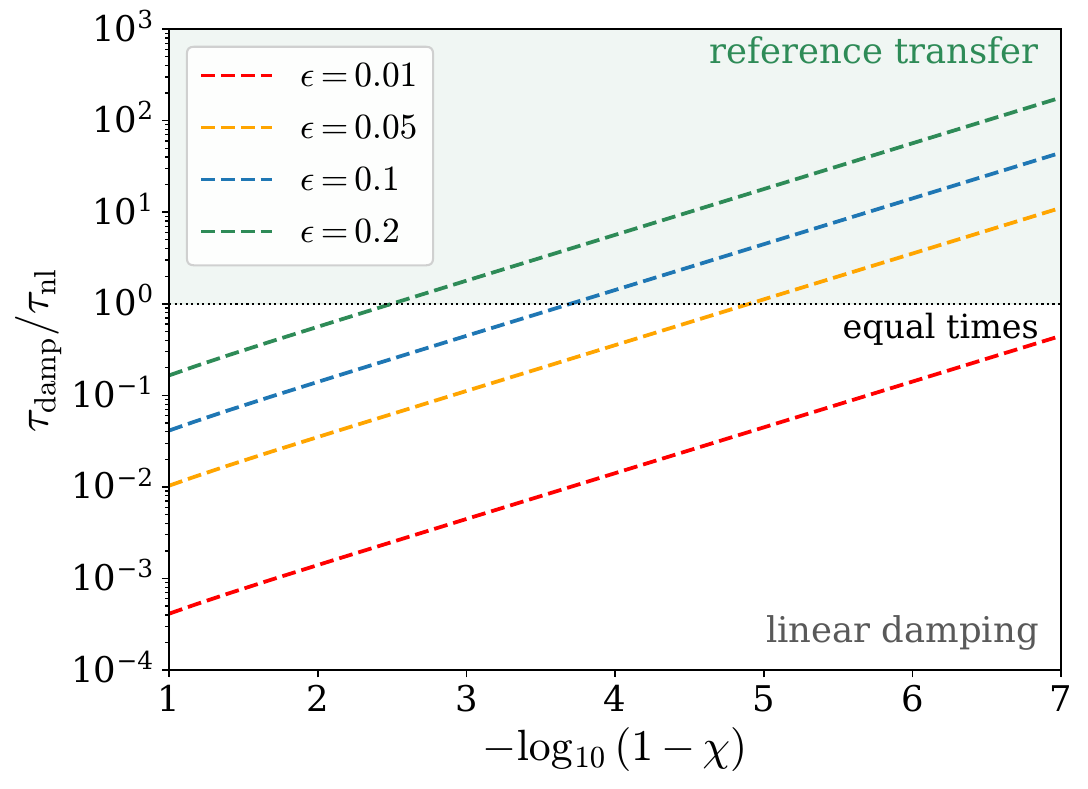}
\caption{Timescale comparison for fundamental $m=2$ ZDMs, using $\tau_{\rm damp}=\kappa_+^{-1}$, $\tau_{\rm nl}^{-1}=\epsilon^2\omega_R$, and $\omega_R\simeq m\Omega_H$. The horizontal line marks equality of times; controlled small linear losses require a ratio much greater than one. The shaded region indicates that the assumed transfer rate exceeds the decay rate; nonlinear conservation is addressed in Sec.~\ref{subsec:fjortoft-interpretation}.}

\label{fig:timescale-separation}
\end{figure}
\begin{equation}
 \alpha_{\rm loss}\equiv\tau_{\rm nl}
 \frac{|\dot{\mathcal Z}|_{\rm lin}}{\mathcal Z}
 \simeq\frac{\kappa_+}{\epsilon^2\omega_R}
 =\frac{\tau_{\rm nl}}{\tau_{\rm damp}}.
 \label{eq:alpha-conservation-quality}
\end{equation}
The subscript distinguishes this loss parameter from the lapse. For a single mode the fraction lost through linear evolution over one transfer time is $1-e^{-\alpha_{\rm loss}}$, approximated by $\alpha_{\rm loss}$ when small.

\subsection{A Cascade Attractor at $\omega = m \Omega_H$?}
\label{subsec:cascade-attractor}

The Fj{\o}rtoft argument fixes the \emph{direction} of energy transfer but not where --if anywhere-- energy accumulates. 
The horizon geometry naturally singles out the frequency $\omega=m \Omega_H$, suggesting it as a candidate spectral attractor.
Near-extremality can help satisfy the timescale condition $\tau_{\rm damp} \gg \tau_{\rm nl}$ so that non-linear transfer uncovers such an attractor.

Kerr admits two Killing vectors, $\xi^a = (\partial_t)^a$ and $\psi^a = (\partial_\phi)^a$, and the horizon $\mathcal{H}^+$ is generated by their helical combination~\cite{WaldGR}
\[
    \chi^a = \xi^a + \Omega_H \psi^a,
\]
which is null precisely on $\mathcal{H}^+$. For a mode $e^{-i \omega t + i m \phi}$,
\[
    \mathcal{L}_\chi e^{-i \omega t + i m \phi} = -i (\omega - m \Omega_H) e^{-i \omega t + i m \phi},
\]
so a mode with $\omega = m \Omega_H$ is stationary with respect to the horizon generator: it co-rotates with $\mathcal{H}^+$.

The Killing charges obey $W^{(\chi)}=W^{(\xi)}-\Omega_HJ$. For real-frequency separated waves their horizon fluxes satisfy
\begin{equation}
 \mathcal F_{\mathcal H}^{(\chi)}
 =\frac{\omega-m\Omega_H}{\omega}\mathcal F_{\mathcal H}^{(\xi)}
 \geq0.
 \label{eq:HH-flux-chi}
\end{equation}
The stationarity-energy flux changes sign at the superradiant threshold. Neither its sign nor positivity of the horizon flux of $W^{(\chi)}$ proves positivity of the corresponding bulk charge. In a normalization based on a regular transmitted metric amplitude, the latter flux contains two powers of $\omega-m\Omega_H$; the conversion from horizon curvature to metric amplitude has additional frequency and surface-gravity factors (\cref{app:horizon-flux}).

Notice that trading $W^{(\xi)}$ for $W^{(\chi)}$ gains us no new conservation law: on a fixed background, the two pairs are linear combinations of each other. Also, a transfer to smaller $m$ along $\omega\simeq m\Omega_H$ can leave the distance from that curve unchanged. To decide whether the threshold attracts spectral power, we need the transfer across it.

\begin{conjecture}[Near-extremal spectral accumulation]
\label{conj:attractor}
Long-lived ZDMs may support nonlinear accumulation near $\omega=m\Omega_H$. Whether such accumulation occurs, and from which side of the threshold, depends on the couplings and the horizon fluxes.
\end{conjecture}
This can be tested by resolving transfer across $\omega-m\Omega_H$ and comparing with linear QNM evolution at the same spin. Excess power near $\omega\simeq m\Omega_H$ alone is insufficient, since the linear ZDM spectrum already clusters there.

\subsection{Kerr–Anti–de Sitter}
\label{subsec:kerr-ads-flux}

In Kerr–AdS, the local Weyl/Bianchi structure is unchanged; the cosmological constant does not enter the local Weyl--Bianchi system, what changes is the global domain and boundary conditions. For reflecting boundaries at AdS conformal infinity:
\[
    \mathcal{F}_{\mathscr{I}_{\rm AdS}}^{W} = 0, \quad \mathcal{F}_{\mathscr{I}_{\rm AdS}}^{\mathcal{Z}} = 0.
\]
The exterior balances reduce to:
\begin{gather*}
    \frac{d}{dt} \delta\mathcal{Z}_{\rm ext}^{\rm KAdS} = \mathcal{R}_{\rm bulk}^{\rm KAdS} - \mathcal{F}_{\mathcal{H}^+}^{\mathcal{Z}} + O(\epsilon^3), \\
    \frac{d}{dt} W_{\rm ext}^{\rm KAdS} = - \mathcal{F}_{\mathcal{H}^+}^{W} + O(\epsilon^3).
\end{gather*}
With the curvature boundary condition of \cref{app:ads-boundary-flux} also imposed, outer leakage is absent; horizon exchange remains.
The boundary conditions also enter the global mode profiles and hence their curvature-to-energy weights, as discussed in Remark~\ref{rem:energy-robustness}. For small black holes, barrier penetration can suppress absorption. Denoting the actual stable-mode quadratic decay rate by $\Gamma_k=2|\omega_{I,k}|$, the relevant hierarchy is
\begin{equation}
 \Gamma_k\ll\tau_{\rm nl}^{-1}.
 \label{eq:ads-hierarchy}
\end{equation}
Small-black-hole calculations illustrate such long lifetimes~\cite{CardosoLemos2001,Konoplya2002}, but their scaling depends on the perturbation sector and rotation. Near-extremality is not required for weak absorption; stability is. Superradiantly growing Kerr--AdS modes must instead be treated with their source terms included. The hierarchy \eqref{eq:ads-hierarchy} is condition (C2) of \cref{subsec:enstrophy-balance}.

\subsection{Horizonless Anti–de Sitter Backgrounds}
\label{subsec:horizonless-ads}

For horizonless AdS (no black hole), both horizon and outer boundary fluxes vanish:
\[
    \mathcal{F}_{\mathscr{I}_{\rm AdS}}^{W} = \mathcal{F}_{\mathscr{I}_{\rm AdS}}^{\mathcal{Z}} = 0, \quad \mathcal{F}_{\mathcal{H}^+}^{W} = \mathcal{F}_{\mathcal{H}^+}^{\mathcal{Z}} = 0.
\]
The balances reduce to:
\[
    \frac{d}{dt} \delta\mathcal{Z}_{\rm AdS} = \mathcal{R}_{\rm bulk}^{\rm AdS} + O(\epsilon^3), \quad \frac{d}{dt} W_{\rm AdS} = O(\epsilon^3).
\]
At this order, the bulk term $\mathcal{R}_{\rm bulk}^{\rm AdS}$ accounts for the change in $\mathcal Z$, while $W_{\rm AdS}$ is conserved.

In global AdS, the background Weyl tensor vanishes (conformally flat). The local linearized vacuum equations admit the exact duality $(\delta E_{ab},\delta B_{ab})\to(\delta B_{ab},-\delta E_{ab})$, without requiring Petrov type D. Boundary conditions can restrict this symmetry. The invariance maps solutions into solutions; it does not make $\delta E^2 = \delta B^2$ for a given solution (a standing normal mode has its electric and magnetic nodes displaced from each other). A weaker property is enough, however. Normal modes of global AdS have real frequencies, so by the argument of \cref{subsec:enstrophy-balance} the phase-averaged $\delta\mathcal{Z}_k$ of each mode is exactly constant at linear order, and $d\langle\delta\mathcal{Z}\rangle/dt = O(\epsilon^3)$ in the transverse radiative sector.
With neither dissipation nor a background Weyl tensor, global AdS is a particularly simple arena in which to test the local bulk conservation law.

\subsection{Weak Turbulence and the AdS instability}
\label{subsec:weak-ads-turbulence}

It is interesting to connect with another case, the AdS instability, where nonlinear interactions induce a rich dynamics that has been the subject of considerable scrutiny. In the small-amplitude, two-time framework (TTF) for a self-gravitating massless scalar, the conserved pair is $E=\sum_k\omega_kn_k$ and $N=\sum_kn_k$~\cite{BuchelGreenLehnerLiebling2015,Craps_2015}. A direct energy cascade can then be accompanied by inverse particle-number transfer, as discussed in Ref.~\cite{BuchelGreenLehnerLiebling2015}. Thus the transfer to higher frequencies seen in the AdS instability~\cite{BizonRostworowski2011} is consistent with the two-moment argument: it is $(E,N)$, rather than $(W,\mathcal Z)$, that constrains this scalar resonant system.

The original instability was studied with smooth, multimode initial data. The scalar field supplies a matter source, so its energy is governed by the coupled Einstein--scalar equations rather than the vacuum Weyl balance used here.

A pure ground-state source is a special case. At $E=\omega_{\min}N$,
\[
 \sum_k(\omega_k-\omega_{\min})n_k=0
\]
forbids occupation of higher modes while both quantities are conserved and $n_k\geq0$. The spectrum $\omega_{n\ell}=d+\ell+2n$ also makes clear that adding angular sectors does not create frequencies below the global ground state. Thus a direct cascade from a pure ground state is not an example of the $(E,N)$ constraint.

The TTF Hamiltonian supplies a further invariant and helps organize special recurrent solutions. Conserved quantities by themselves neither imply recurrence nor rule out collapse; an argument that works for a finite truncation need not carry over to the full evolution.

In nonspherical vacuum AdS, weakly nonlinear calculations exhibit secular higher-frequency corrections, while some single-mode seeds admit nonlinear extensions to geons~\cite{DiasHorowitzSantos2012,DiasHorowitzMarolfSantos2012,DiasSantos2016}. Third-order secular terms motivate a four-wave description but do not, by themselves, establish an exclusively $2\leftrightarrow2$ normal form or conservation of wave action. The corresponding resonant invariants and the strong-field weights must be determined before applying either two-moment argument. This dynamical question is separate from the local magnetic Weyl balance.

\subsection{Driven and Finite-Domain Systems}
\label{subsec:driven-domains}

Our analysis naturally extends to driven and finite-domain systems. For radiation injected through a timelike boundary or external perturbation at frequency $\omega_{\rm inj}$, the balance laws are:
\[
    \dot{W} = \mathcal{I}_W - \mathcal{D}_W, \quad \dot{\mathcal{Z}} = \mathcal{I}_{\mathcal{Z}} - \mathcal{D}_{\mathcal{Z}} + \mathcal{R}_{\rm bulk},
\]
where $\mathcal{I}$ = injection, $\mathcal{D}$ = dissipation/leakage. In analogy with studies in fluid-dynamics, a statistically steady state is defined by the conditions:
\[
    \langle \dot{W} \rangle = 0, \quad \langle \dot{\mathcal{Z}} \rangle = 0,
\]
even if neither quantity is conserved instantaneously. The simulations presented in~\cite{MaLehnerYangKidderPfeifferScheel2026} realized such scenario through boundary injections, being able
to set steady state regimes and study energy cascade in both  black hole and flat perturbed spacetimes.
Driving maintains the source amplitude, but transfer into daughter modes still competes with their damping. A steady dual-flux interpretation requires an interval of scales over which injection and losses are negligible and the nonlinear dynamics preserves the proposed pair. Power below a monochromatic injection frequency is evidence of redistribution; it does not by itself fix the fraction of energy transferred downward or select an interaction channel.

To convert the horizon $\Psi_0$ spectrum into Killing-energy flux, we use the horizon response factors of \cref{app:horizon-flux}. One cannot use $F^{\mathcal Z}_\omega=\omega^2F^W_\omega$ there. Jointly measuring the curvature, energy and action spectra, together with their boundary budgets, would test the proposed conservation structure. The mode couplings determine the slopes and the dominant interaction.

A radiated action diagnostic can be defined without reference to perturbation theory, at least for the radiated part. At $\mathscr{I}^+$ let $H(u,\theta,\phi)=\lim r\,(h_+-ih_\times)$ be the asymptotic strain, fixed by the Bondi shear of the full spacetime. Here the physical radiated energy is written in $G=1$ units, without the $8\pi G$ rescaling of $W$. The radiated energy, given by the Bondi mass-loss formula~\cite{Bondi1962}, and the radiated action are
\begin{equation}
\begin{aligned}
 E_{\rm rad}&=\frac{1}{16\pi}\oint\!\int_{-\infty}^{\infty}
 \frac{d\omega}{2\pi}\,\omega^2|\tilde H(\omega)|^2\,d\Omega,\\
 N_{\rm rad}&=\frac{1}{16\pi}\oint\!\int_{-\infty}^{\infty}
 \frac{d\omega}{2\pi}\,|\omega||\tilde H(\omega)|^2\,d\Omega.
\end{aligned}
\label{eq:N-rad}
\end{equation}
so that $dE_{\rm rad}/d\omega=|\omega|\,dN_{\rm rad}/d\omega$ at every frequency. $N_{\rm rad}$ is the norm of the positive-frequency part of the radiative data in the symplectic structure of the radiative phase space of full General Relativity~\cite{AshtekarStreubel1981}; divided by $\hbar$, it is the expected number of gravitons in the coherent state describing the classical radiation~\cite{Ashtekar1981}, and for linear radiation it reduces to the corresponding physical-energy mode action. The action associated with the rescaled $W$ differs by the same common $8\pi G$ factor. The factor $|\omega|$ makes this expression nonlocal in retarded time. Memory also introduces an infrared divergence: for $\Delta H\neq0$, $\tilde H\sim i\Delta H/\omega$ at small $\omega$ and the integrand behaves as $|\Delta H|^2/|\omega|$. The diagnostic therefore uses the spectral density $dN_{\rm rad}/d\omega$ at finite frequency, or the fluxes between frequency bands, rather than $N_{\rm rad}$ itself. At the horizon the corresponding flux follows from the Teukolsky--Press expression~\cite{TeukolskyPress1974} only perturbatively, and relative to $\partial_t$ it carries the sign of $\omega-m\Omega_H$. In the bulk, conservation of $N$ belongs to the $2\leftrightarrow2$ resonant system discussed in \cref{subsec:fjortoft-interpretation}; the three-wave processes considered there change it.

\subsection*{Summary}
\label{subsec:regime-summary}

Admittedly, the presentation so far has been long and with multiple fronts and applications discussed; to help the reader 
\cref{tab:regime-summary} summarizes the behavior of the enstrophy balance in the regimes discussed above.
\begin{table*}[t]
\centering
\caption{Role of gravitational enstrophy across the physical regimes considered.}
\label{tab:regime-summary}
\begin{tabular}{|p{0.14\textwidth}|p{0.23\textwidth}|p{0.25\textwidth}|p{0.25\textwidth}|}
\hline
Regime & External losses & Conservation & Role of $\mathcal{Z}$ \\
\hline
Generic Kerr ringdown & radiation through $\mathscr{I}^+$ and absorption at $\mathcal{H}^+$ & suppressed: $\tau_{\rm damp} \lesssim \tau_{\rm nl}$ & diagnostic only; the cascade is suppressed by damping \\
\hline
Near-extremal Kerr & quadratic ZDM decay rate $\sim\kappa_+$ & linear losses small if $\kappa_+\tau_{\rm nl}\ll1$ & strong-field weights and nonlinear conservation remain to be established \\
\hline
Kerr--AdS & no outer energy leakage; horizon exchange remains & stable long-lived modes may satisfy \cref{eq:ads-hierarchy} & confined setting; curvature weights and fluxes require a separate check \\
\hline
Global AdS, vacuum & no horizon or outer energy leakage & finite modal envelopes constant at linear order & nonlinear invariants and mode weights determine applicability \\
\hline
Global AdS with a scalar field & matter sources in the Weyl balance & $(E,N)$ conserved in the scalar TTF & two-moment redistribution; special recurrent data \\
\hline
Driven finite domain & injection and dissipation set externally & statistical steady state may occur & test spectral transfer with both injection and losses included \\
\hline

\end{tabular}
\end{table*}
As argued, a common feature is that a local magnetic Weyl enstrophy structure is universal. The main differences lie in the flux balance, energy source location, and additional conservation laws. 

\section{Holographic Interpretation and Fluid/Gravity Diagnostics}
\label{sec:holographic-fluid}

The fluid/gravity correspondence~\cite{Bhattacharyya:2007vjd}
provides a related magnetic Weyl realization of fluid enstrophy.  Further, turbulent regimes in gravity dual to those on the fluid side of the correspondence have been  demonstrated~\cite{CarrascoLehnerMyersReulaSingh2012,Adams:2013vsa}.
Here we work out this relation for asymptotically $\text{AdS}_4$ spacetimes, mapping the bulk magnetic Weyl curvature dynamics onto the hydrodynamic observables of the dual conformal field theory (CFT).

\subsection{Holographic fluid enstrophy}
\label{subsec:boundary-fluid-enstrophy}
Let the boundary coordinates be $x^\mu=(v,x,y)$ on a flat metric $\eta_{\mu\nu} = \mathrm{diag}(-1, 1, 1)$. We consider a CFT state whose late-time, long-wavelength dynamics are governed by a relativistic conformal fluid. The boundary stress-energy tensor admits the standard hydrodynamic expansion:
\begin{align}
    T_{\mu\nu}^{\mathrm{CFT}}
    =
    (\varepsilon+p) u_\mu u_\nu
    +
    p \, \eta_{\mu\nu}
    +
    \Pi_{\mu\nu} \,,
    \label{eq:cft-stress-fluid}
\end{align}
where $u^\mu$ is the fluid local three-velocity, $\varepsilon$ is the energy density, $p$ is the pressure, and $\Pi_{\mu\nu}$ represents dissipative gradient corrections. 

In the non-relativistic, incompressible limit, with $\beta \equiv |v| \ll 1$ and divergence-free flow, the fluid velocity reduces to:
\begin{align}
    u^\mu = (1, v^i) \,, \qquad \partial_i v^i = 0 \,.
\end{align}
In this limit the flow is described by the scalar vorticity:
\begin{align}
    \omega_{\mathrm{fl}} = \epsilon^{ij} \partial_i v_j \,,
\end{align}
and the corresponding two-dimensional enstrophy is:
\begin{align}
    \Omega_{\mathrm{CFT}} = \int_{\partial\Sigma_t} \omega_{\mathrm{fl}}^2 \sqrt{h} \, d^2x \,.
\end{align}
(Here $\sqrt{h}=1$ since the boundary metric is flat; we keep it for completeness.)
To map this boundary observable to bulk variables, we invert the hydrodynamic constitutive relation. To leading order in the non-relativistic expansion, the momentum density satisfies $T^{0i}_{\mathrm{CFT}}\simeq(\varepsilon+p)v^i$. Substituting this into the definition of vorticity yields:
\begin{align}
    \omega_{\mathrm{fl}}
    =
    \epsilon^{ij} \partial_i \left( \frac{T^{0j}_{\mathrm{CFT}}}{\varepsilon+p} \right)
    +
    \mathcal{O}(\beta^2, \partial^2) \,.
\end{align}
For fluctuations around a homogeneous thermal equilibrium where the global enthalpy density $\varepsilon+p$ remains constant to the order of interest, the boundary enstrophy simplifies to:
\begin{align}
    \Omega_{\mathrm{CFT}}
    =
    \frac{1}{(\varepsilon+p)^2} \int_{\partial\Sigma_t} \left( \epsilon^{ij} \partial_i T^{0j}_{\mathrm{CFT}} \right)^2 \sqrt{h} \, d^2x
    +
    \mathcal{O}(\beta^3, \partial^3) \,.
    \label{eq:Omega-stress-constant-enthalpy}
\end{align}
The boundary enstrophy is thus written entirely in terms of the stress tensor, which is the form we relate to the bulk magnetic Weyl enstrophy below.

\subsection{Fluid/Gravity Reconstruction and Slicing Map}
\label{subsec:fluid-gravity-reconstruction}

The fluid/gravity correspondence~\cite{Bhattacharyya:2007vjd} reconstructs an asymptotically $\text{AdS}_4$ bulk metric, order by order in boundary gradients, from a solution of the boundary fluid equations. It is most simply written in ingoing Eddington--Finkelstein (EF) coordinates, which are regular at the future event horizon; to first order in gradients the metric reads:
\begin{align}
    ds^2
    &=
    -2 u_\mu dx^\mu dr
    -
    r^2 f(br) u_\mu u_\nu dx^\mu dx^\nu \nonumber \\
    &\quad +
    r^2 P_{\mu\nu} dx^\mu dx^\nu
    +
    ds^2_{\mathrm{grad}} \,,
    \label{eq:fluid-gravity-metric}
\end{align}
where the spatial projector and black-brane factor are:
\begin{align}
    P_{\mu\nu} = \eta_{\mu\nu} + u_\mu u_\nu \,, \qquad f(br) = 1 - \frac{1}{(br)^3} \,.
\end{align}
Here, $b = 3/(4\pi T)$ is the inverse temperature parameter, and $ds^2_{\mathrm{grad}}$ parameterizes first-order derivative corrections built from fluid shear, vorticity, acceleration, and temperature gradients (details can be found in~\cite{Bhattacharyya:2007vjd}). 

To read off the boundary stress tensor we pass to Fefferman--Graham (FG) coordinates, in which the metric takes the form:
\begin{align}
    ds^2 = \frac{\ell^2}{z^2} \left[ dz^2 + \left( \eta_{\mu\nu} + z^3 g_{\mu\nu}^{(3)} + \dots \right) dx^\mu dx^\nu \right] \,,
\end{align}
which yields the boundary stress tensor through the normalizable mode:
\begin{align}
    T_{\mu\nu}^{\mathrm{CFT}} = \frac{3\ell^2}{16\pi G_4} g_{\mu\nu}^{(3)} \,.
\end{align}
A subtlety arises in specifying the integration surface. Constant EF time is null for the static brane. For a boosted flow, already at zeroth derivative order,
\[
 g^{vv}=\frac{P^{vv}}{r^2}=\frac{(u^v)^2-1}{r^2}
 =\frac{\gamma^2\beta^2}{r^2}>0\qquad(\beta\neq0).
\]
Here $\gamma=(1-\beta^2)^{-1/2}$. The normal is therefore spacelike, and the constant-$v$ surfaces are timelike, with their departure from nullness entering at $O(\beta^2)$. The radial coordinate curves remain null and affinely parametrized, but for vortical flow they are not normal generators of these surfaces. For the leading quadratic-gradient diagnostic we use the background EF measure $r^2dr\,d^2x$; corrections of order $\beta^2$ multiply an already quadratic perturbative density and lie beyond the amplitude order considered here. A spacelike alternative is discussed in \cref{app:fg-spacelike}.

Define the \emph{fluid-weighted} magnetic Weyl tensor
\begin{align}
    \tilde B_{ab} \equiv {}^{\star}C_{acbd}\, u^c u^d \,,
    \label{eq:fluid-weighted-B}
\end{align}
where $u^a$ is the bulk extension $u^\mu\partial_\mu$ of the boundary-normalized fluid velocity ($\eta_{\mu\nu}u^\mu u^\nu = -1$), used \emph{without} unit normalization in the bulk. Symmetry, tracelessness, and transversality ($\tilde B_{ab}u^b = 0$) are algebraic properties of ${}^{\star}C_{abcd}$ and hold for any timelike vector; the normalization enters only as an overall weight,
\begin{align}
    \tilde B_{ab}=|u|^2B_{ab}[U],\qquad U^a=\frac{u^a}{|u|},\quad |u|^2=r^2f+O(\partial) \, .
\end{align}
At equilibrium $u^a$ is timelike outside the brane and becomes its null generator at the horizon. Its norm receives derivative corrections away from equilibrium. The weighted tensor can be evaluated directly in regular EF coordinates without dividing by $r^2f$. We define
\begin{align}
 \delta\tilde{\mathcal Z}(v)
 =\int_{1/b}^{\infty}\!dr\int d^2x\,
 r^2\delta\tilde B_{ab}\delta\tilde B^{ab}.
 \label{eq:weighted-null-enstrophy}
\end{align}
The EF coordinates, radial measure and boundary normalization of $u^\mu$ are part of this definition. Bulk and boundary data are paired at the same boundary time.

With this prescription fixed, the coefficient is unambiguous. Recall that the correspondence holds in the long-wavelength regime, where gradients are subleading. The perturbative expansion then proceeds as before, with the gradient expansion as additional control: $\tilde B_{ab}$ vanishes for the uniformly boosted brane, and the leading contribution is first order in boundary gradients. Two sectors are available by symmetry: the pseudo-scalar vorticity $\omega_{\rm fl}$ and the pseudo-tensor $\epsilon^{k}{}_{(i}\sigma_{j)k}$ built from the boundary shear. (Expansion vanishes by incompressibility, and the acceleration is higher order on shell: for linearized incompressible hydrodynamics, $\partial_t v_i = O(\partial^2 v)$.) In the non-relativistic incompressible limit the two sectors contribute with no cross term,
\begin{align}
    \delta\tilde B_{ab}\, \delta\tilde B^{ab} = \tilde K_\omega(r; T) \, \omega_{\mathrm{fl}}^2 + \tilde K_\sigma(r;T)\,\sigma_{ij}\sigma^{ij} + \mathcal{O}(\beta^3, \partial^3) \,.
    \label{eq:BB-fluid-kernel}
\end{align}
Here $\sigma_{ij}=\partial_{(i}v_{j)}-\tfrac12\delta_{ij}\partial_kv^k$. Both radial profiles follow in closed form, each from a first-order solution. A rigidly rotating boundary flow, $v_i = \Omega\,(-y, x)$, has vanishing shear and isolates $\tilde K_\omega$; since its expansion and shear vanish and its acceleration is $O(\Omega^2)$, the metric \cref{eq:fluid-gravity-metric} with the corresponding position-dependent $u_\mu$ solves the Einstein equations exactly at $O(\Omega)$ with no correction $ds^2_{\rm grad}$ sourced. A pure-shear flow, $v_i = s\,(y, x)$, isolates $\tilde K_\sigma$; there a first-order response $\delta g_{xy} \propto \sigma_{xy}$ is sourced, but carrying its radial profile as an unknown function shows that it cancels identically from $\delta\tilde B_{ab}$ --- at this order the weighted magnetic sector is algebraic in the boundary gradients and insensitive to $ds^2_{\rm grad}$, which feeds only the electric sector --- so both profiles are fixed at this order. At the symmetry point of each flow (where $\omega_{\rm fl} = 2\Omega$ and $\sigma_{xy} = s$, respectively), the nonvanishing components are, in $(v,x,y,r)$ coordinates,
\begin{align}
    \delta\tilde B_{xx} = \delta\tilde B_{yy} = -\frac{3 f}{2 b^3}\,\Omega \,, \qquad
    \delta\tilde B_{rr} = \frac{3}{b^3 r^4}\,\Omega
    \label{eq:weighted-B-components}
\end{align}
for the rotation, and $\delta\tilde B_{xx} = -\delta\tilde B_{yy} = -3s/(2b^3)$ for the shear (both trace-free and orthogonal to the fluid vector, and superposing at first order for mixed flows), hence
\begin{align}
    \tilde K_\omega(r; T) = \frac{27\, f^2(br)}{8\, b^6 r^4} \,, \qquad
    \tilde K_\sigma(r; T) = \frac{9}{4\, b^6 r^4} \,.
    \label{eq:weighted-kernel-closed}
\end{align}
Both profiles are regular at $\mathcal{H}^+$ --- $\tilde K_\omega$ vanishes there as $f^2$, as follows directly from the rotation solution, while $\tilde K_\sigma$ remains finite --- and both integrands fall as $r^{-2}$ at conformal infinity, so the radial integrals converge at both ends. For incompressible flows with decaying or periodic data, the integrated shear and vorticity are not independent, $\int \sigma_{ij}\sigma^{ij}\, d^2x = \tfrac12 \int \omega_{\rm fl}^2\, d^2x$, so the holographic map retains its one-coefficient form:
\begin{align}
    \delta\tilde{\mathcal{Z}} &= \tilde{\mathcal{C}}(T)\,\Omega_{\rm CFT} + \mathcal{O}(\beta^3, \partial^3) \,, \qquad
    \tilde{\mathcal{C}} = \tilde{\mathcal{C}}_\omega + \tfrac12\,\tilde{\mathcal{C}}_\sigma \,, \nonumber\\
    \tilde{\mathcal{C}}_\omega &= \frac{243}{112\, b^5} = \frac{64\pi^5}{7}\, T^5 \,, \qquad
    \tilde{\mathcal{C}}_\sigma = \frac{9}{4\, b^5} = \frac{256\pi^5}{27}\, T^5 \,, \nonumber\\
    \tilde{\mathcal{C}}(T) &= \frac{369}{112\, b^5} = \frac{2624\,\pi^5}{189}\, T^5 \,,
    \label{eq:weighted-kernel-integrated}
\end{align}
in units $\ell_{\rm AdS}=1$. The coefficient belongs to the stated EF weighting and measure; another observer or spacelike volume element requires a new calculation.

This map uses the full first-gradient magnetic field, not the transverse radiative projection of the earlier sections. In the unit fluid frame the rotational contribution is proportional to $P_{ab}$, with orthonormal components $[-3\Omega/(2b^3r^4),-3\Omega/(2b^3r^4),3\Omega/(b^3r^4)]$. It would be excluded by a $\Psi_0,\Psi_4$ projection. The map therefore gives a related realization of enstrophy, rather than a proof of conservation of the radiative functional. Its spectral weighting is the fluid spatial-wavenumber weight $k^2$, not a vacuum oscillation-frequency weight.

For constant $T$ and an incompressible unforced fluid with periodic or decaying data, viscosity gives
\begin{align}
 \frac{d\Omega_{\rm CFT}}{dv}
 &=-2\nu\int|\nabla\omega_{\rm fl}|^2\,d^2x,\nonumber\\
 \frac{d\delta\tilde{\mathcal Z}}{dv}
 &=-2\nu\tilde{\mathcal C}(T)\int|\nabla\omega_{\rm fl}|^2\,d^2x
 +\text{higher-order terms}.
\end{align}
Thus conservation is exact for the ideal fluid and approximate when viscous losses are small over the transfer time. A changing temperature also contributes $\dot{\tilde{\mathcal C}}\Omega_{\rm CFT}$.

\subsection{Relation to Holographic Turbulence}
\label{subsec:holo-turbulence-relation}

The bulk-to-boundary map above is meant as a diagnostic for numerical simulations of holographic turbulence in $(2+1)$-dimensional conformal fluids (e.g.~\cite{CarrascoLehnerMyersReulaSingh2012,GreenCarrascoLehner2014}). There, the bulk geometry realizes gravitationally the dual cascade of the boundary fluid, inverse in energy and direct in enstrophy, and the ratio $\mathcal{R}_{\mathrm{holo}}(v)$ defined below provides a quantitative test of the correspondence between the gravitational and fluid enstrophies.

We define a holographic diagnostic ratio at boundary time $t=v$ (the bulk functional uses the EF prescription at that boundary time, with the fluid-weighted construction of \cref{subsec:fluid-gravity-reconstruction}) as:
\begin{align}
    \mathcal{R}_{\mathrm{holo}}(v)
    =
    \frac{\delta\tilde{\mathcal{Z}}(v)}{\tilde{\mathcal{C}}(T) \, \Omega_{\mathrm{CFT}}(v)} \,,
    \label{eq:Rholo-diagnostic}
\end{align}

In the hydrodynamic regime, the duality dictates that $\mathcal{R}_{\mathrm{holo}}(v) = 1$ at leading order in the gradient expansion. This is to be expected: at first order in gradients the two quadratic scalars retained here are $\omega_{\rm fl}^2$ and $\sigma_{ij}\sigma^{ij}$, and incompressibility ties their integrals together, so any functional of the type considered is proportional to $\Omega_{\rm CFT}$ at this order. \Cref{eq:weighted-kernel-integrated} fixes the proportionality coefficient and gives a weighted density that is regular at the horizon. 
To interpret departures from unity, one would first estimate the hydrodynamic corrections from higher gradients, compressibility and temperature variation, using the same integration prescription. A residual discrepancy could then test for non-hydrodynamic gravitational excitations. Hence, if $\delta\tilde{\mathcal{Z}}$ tracks $\Omega_{\mathrm{CFT}}$ as the boundary flow develops its dual cascade, this provides direct holographic evidence that the bulk magnetic Weyl enstrophy follows the boundary enstrophy through the cascade.

\subsection{\texorpdfstring{$\text{AdS}_4/\text{CFT}_3$}{AdS4/CFT3} versus \texorpdfstring{$\text{AdS}_5/\text{CFT}_4$}{AdS5/CFT4}}
\label{subsec:ads4-vs-ads5}

The dimension of the boundary matters here. In $\text{AdS}_4/\text{CFT}_3$ the boundary fluid lives in two spatial dimensions, the vorticity is a pseudo-scalar, and the boundary enstrophy,
\begin{align}
    \Omega_{\mathrm{CFT}} = \int \omega_{\mathrm{fl}}^2 \sqrt{h} \, d^2x \,,
\end{align}
is conserved by the ideal (inviscid) 2D Euler equations. This is the case where the hydrodynamic analogy is tightest.

In $\text{AdS}_5/\text{CFT}_4$, instead, the boundary fluid is three-dimensional and the vorticity is a vector, $\boldsymbol{\omega} = \nabla \times \mathbf{v}$. The enstrophy is then not conserved even in the inviscid limit, because of vortex stretching:
\begin{align}
    \frac{d}{dt} \int |\boldsymbol{\omega}|^2 \, d^3x
    =
    2 \int \boldsymbol{\omega} \cdot \left[ (\boldsymbol{\omega} \cdot \nabla) \mathbf{v} \right] d^3x \,.
\end{align}
Thus, while a bulk magnetic Weyl functional can still be defined in $\text{AdS}_5$, its boundary counterpart is not a conserved quantity. This is the same dimensional distinction that separates the inverse cascade of 2D turbulence from the direct cascade in 3D.

\section{Discussion and Outlook}
\label{sec:conclusion}

In this work we have introduced a magnetic Weyl measure of gravitational enstrophy and examined the conditions under which it can constrain spectral transfer. Its radiative part has the frequency-squared weighting relative to energy in the radiation zone. Its local balance follows from the Bianchi identities, with algebraic simplifications that can be stated independently of a mode expansion. In the hydrodynamic regime of AdS$_4$, a related weighted functional maps to boundary enstrophy with a coefficient computed in closed form.

The remaining question is dynamical: which interactions realize the transfers allowed by the conserved quantities? Sec.~\ref{subsec:fjortoft-interpretation} gives simple weak-wave examples, while the fluid/gravity map supplies the hydrodynamic case.

\subsection{Summary of Established Results}
\label{subsec:established}
The $(W,\mathcal Z)$ argument has a hydrodynamic counterpart through the fluid-weighted map; for resonant $2\leftrightarrow2$ waves the conserved pair is instead $(N,W)$, and the interactions generally change $\mathcal Z$ (Sec.~\ref{subsec:fjortoft-interpretation}).

The core kinematic results, the spectral weighting $\delta\mathcal{Z}_k/W_k \simeq \omega_k^2$ and the Fj{\o}rtoft constraint that follows from it when $\mathcal{Z}$ is conserved, are established in Proposition~\ref{thm:spectral-weighting} and Theorem~\ref{thm:fjortoft}. The perturbative balance law,
\begin{align}
    \frac{d}{dt}\delta\mathcal{Z} = \mathcal{F}_{\mathcal{Z}} + \mathcal{S}_{\mathrm{source}} \,,
    \label{eq:Z-final-balance-discussion}
\end{align}
follows from the Einstein equations, with the bulk and boundary terms derived in Sec.~\ref{sec:evolution}. The vorticity contraction vanishes identically. The shear contraction vanishes for a transverse field with $\sigma_{\hat r\hat r}=0$ in the same frame, and the ZAMO expansion vanishes on stationary circular backgrounds. Equality of electric and magnetic norms allows their balances to be combined, leaving bulk sources and a boundary flux. A stationary linear mode supplies an independent consistency check on their sum.

Table~\ref{tab:logical-status} collects the main results and the calculations that remain.

\begin{table*}[t]
\centering
\caption{Main results and remaining calculations.}
\label{tab:logical-status}
\smallskip
\begin{tabular}{p{0.34\textwidth}p{0.60\textwidth}}
\toprule
Quantity or relation & Result or remaining calculation \\
\midrule
Electric--magnetic Weyl decomposition & Local geometric identity \\
$\eta_k=1$  & {Radiation zone (Proposition~\ref{thm:spectral-weighting}); $O(1)$, mode-dependent correction for ZDMs, to be computed (\cref{app:spectral-ratio-proof})} \\
Vorticity coupling to $\dot{\mathcal{Z}}$ & Vanishes identically (local algebraic proof) \\
Weyl curl-exchange coupling & Combines into a super-Poynting flux in the sum of the two norms \\
Fj{\o}rtoft transfer constraint & Algebraic consequence of positive ordered weights and conservation through transfer (Theorem~\ref{thm:fjortoft}) \\
Shear and Coulomb residuals & Enter the bulk balance together with linearization and projection sources \\
Conservation through nonlinear interactions & Generally violated by the active weak-wave channels of Sec.~\ref{subsec:fjortoft-interpretation}; inherited from ideal incompressible flow for the weighted fluid/gravity map \\
Direction of energy transfer & Receiver fractions follow from the two conserved moments; net flux depends on the couplings \\
Dynamical relevance & Comparison of transfer and linear-loss times \\
Strong-field spectral ordering & Mode or block norms remain to be computed \\
Asymptotic spectral index & To be determined from the mode couplings \\

\bottomrule
\end{tabular}
\end{table*}

For near-horizon ZDMs, the main outstanding calculation is the comparison of curvature and energy norms across the mode family (\cref{app:spectral-ratio-proof}). This would determine the ordering needed by the transfer argument, including the treatment of coherent bands in \cref{app:degenerate-fjortoft}.

In the fluid/gravity regime, the explicit coefficient $\tilde{\mathcal C}(T)=(2624\pi^5/189)T^5$ gives a quantitative relation to boundary enstrophy. The definition and EF measure used for this weighted functional are given in \cref{sec:holographic-fluid}.

The construction is also special to four dimensions, mirroring the fact that enstrophy conservation is itself special to two spatial dimensions on the fluid side. In $D=5$, the magnetic part of the Weyl tensor is no longer a symmetric rank-2 spatial tensor but a rank-3 object $H_{abc}$, and an additional irreducible, fully spatial Weyl component appears that in $D=4$ is algebraically determined by $E_{ab}$. The two-field curl structure used in \cref{sec:evolution} is then structurally absent: the Bianchi projections couple three blocks rather than two, and no analogue of the duality argument of \cref{rem:duality-pointwise} is available.

\subsection{Open Questions and Future Directions}
\label{subsec:open}

\paragraph{Inertial range and spectral slopes.}
Even where the Fj{\o}rtoft constraint applies, it restricts the direction of spectral transfer but does not determine the power-law indices of the resulting inertial range (e.g., the gravitational analogues of the Kolmogorov $k^{-5/3}$ or Kraichnan $k^{-3}$ spectra). Deriving the corresponding slopes requires establishing the scale-invariance and locality of the nonlinear gravitational mode-coupling coefficients. 

{The way forward is a direct calculation of the second-order Teukolsky coupling coefficients $C_{k_1 k_2 k_3}$. That analysis must also address modal locality across the ZDM band. The second-order perturbation framework of~\cite{LoutrelRipleyGiorgiPretorius2021, RipleyLoutrelGiorgiPretorius2021} is well suited to this, particularly in NHEK, where the radial functions are hypergeometric.}

\paragraph{Subleading observer dependence.}
The radiation-zone ratio retains its frequency-squared form under consistent frame changes. Observer-dependent corrections can nevertheless introduce numerical shifts in the cascade onset threshold $\epsilon^2 \gg m^{-1}\sqrt{2(1-\chi)}$. These corrections should be quantified to fix the domain of validity of the enstrophy balance in strongly dynamical, non-axisymmetric configurations.

\paragraph{Nonlinear validation of the enstrophy balance.}
Joint measurements of $\dot{\mathcal Z}$, $\dot W$ and $\dot N$, after accounting for external fluxes, can test specific interaction models (\cref{subsec:fjortoft-interpretation}). A future numerical test of the balance law would check that
\begin{align}
    \Delta\mathcal{Z} + \int \mathcal{F}_{\mathcal{Z}} \, dt - \int \left( \mathcal{S}_{\sigma} + \mathcal{S}_a + \mathcal{S}_{\mathrm{mix}} \right) dt = \mathcal{O}(\epsilon^3)
\end{align}
\rev{using a horizon-regular observer in a second-order Teukolsky or fully nonlinear Einstein evolution. A horizon-penetrating slice with a suitable outer prescription, for example hyperboloidal compactification, would control both ends of the domain; the ZAMO integral diverges at the horizon for the ingoing field (\cref{app:acceleration-convergence}). The bulk balance would include expansion and the linearization and projection terms, as well as the shear and acceleration diagnostics. The displayed remainder is at fixed perturbative time; following its accumulation over $\tau_{\rm nl}$ is part of this proposed test.}

\paragraph{Third conservation laws and phase-space recurrence.}
In the scalar TTF~\cite{BuchelGreenLehnerLiebling2015}, a resonant Hamiltonian accompanies the energy and particle number and helps characterize recurrent solutions. Whether near-extremal Kerr or Kerr--AdS admits further resonant invariants is an open question. The role of such an invariant in recurrence is discussed in \cref{subsec:weak-ads-turbulence}.

\paragraph{Universality and critical phenomena.}
A more speculative question is whether the dual-cascade constraints share structure with the universality of critical collapse, and in particular whether the critical exponent is related to the spectral slope.

\subsection{Observational Signatures and Phenomenological Implications}
\label{subsec:observational-signatures}

What would an inverse transfer toward the ZDM threshold look like in observations? The first two examples below explore the conjecture of Sec.~\ref{subsec:cascade-attractor}; the third concerns a numerical test in the fluid/gravity regime.

\paragraph{Gravitational-wave spectral content.}
{If} the {cascade} shifts spectral energy toward the ZDM threshold, near-extremal Kerr ringdowns {would} display a pronounced spectral excess near the co-rotation frequency:
\begin{align}
    \omega_{\mathrm{attractor}} \approx m \Omega_H = \frac{m}{2M} \left( \frac{\chi}{1+\sqrt{1-\chi^2}} \right) \,.
\end{align}
For the dominant $m=2$ mode in the extreme limit $\chi \to 1$, this would concentrate power near $\omega_{\mathrm{peak}} \approx 1/M$, with a corresponding suppression of higher-frequency harmonics relative to linear perturbation theory. The relevant comparison is with linear QNM evolution at the same spin, since the linear ZDM spectrum already clusters near $m\Omega_H$ (Conjecture~\ref{conj:attractor}). {Detecting this nonlinear excess would require high-SNR ringdown spectroscopy of rapidly spinning merger remnants, of the kind expected with the Einstein Telescope, Cosmic Explorer, or LISA.}

\paragraph{Nonlinear wave memory enhancement.}
\rev{A redistribution of radiative energy to lower frequencies would change the time profile and angular distribution of the radiated energy flux, and with them the nonlinear Christodoulou memory, which that flux sources. The memory itself is a zero-frequency effect and carries no weight in $\mathcal{Z}$; what such a cascade would modify is the build-up of the memory signal on timescales of order $M$. Turning this into a quantitative prediction requires evaluating the memory integral on the cascaded flux.}

\paragraph{Holographic fluid diagnostics.}
In the context of the AdS/CFT correspondence, the ratio $\mathcal{R}_{\mathrm{holo}}(v)$ defined in \cref{eq:Rholo-diagnostic} can be evaluated in numerical simulations of holographic 2D turbulence (such as those of~\cite{CarrascoLehnerMyersReulaSingh2012,GreenCarrascoLehner2014}). Showing that the bulk magnetic Weyl enstrophy tracks the boundary fluid enstrophy through a turbulent cascade would give a direct numerical test of the bulk enstrophy interpretation.

\subsection{Closing Perspective}

In fluid dynamics, much of what is understood about strongly nonlinear, multiscale evolution comes from conserved spectral quantities and the constraints they impose on transfer. Here we have identified a geometric counterpart in General Relativity: a curvature-based second moment whose balance law follows from the Bianchi identities rather than from the properties of a particular background. Its spectral weighting gives a gravitational Fj{\o}rtoft constraint when the interactions conserve the curvature norm and the weights are ordered. The underlying ingredients---the quadratic Weyl structure, the algebraic cancellation of the vorticity coupling, and the self-adjoint curl exchange---follow directly from Einstein's equations. In near-extremal Kerr, Kerr--AdS and related confined geometries, linear losses can be slow enough to examine nonlinear redistribution. The relevant conservation laws are then set by the interactions.

{Fully nonlinear simulations exhibit inverse spectral transfer in black-hole spacetimes~\cite{YangZimmermanLehner2015,MaLehnerYangKidderPfeifferScheel2026,Zhu:2026mhn,Cardoso:2026llh,Ianniccari:2025nkf}; related behavior has also appeared in non-vacuum settings~\cite{Siemonsen:2025fne,Redondo-Yuste:2025hlv}, recent mathematical analyses~\cite{Kehle:2026vxm}, and perturbative studies of the matter sector~\cite{Kehagias:2025zws}. Whether these results are manifestations of a common geometric mechanism remains an open question. The distinction introduced in Sec.~\ref{sec:introduction} suggests a concrete comparison. Where $(N,W)$ is conserved, the action spectrum can reveal inverse action transfer; in a hydrodynamic regime, the weighted bulk norm can be compared with fluid enstrophy. Joint measurements of these spectra and of $\dot{\mathcal Z}$, $\dot W$ and $\dot N$ would help determine which description accounts for the observed evolution.}

\section*{Acknowledgements}
The author thanks  S. Aretakis, M. Dafermos, W. East, S. Ma, E. Poisson, F. Pretorius, N. Siemonsen,
H. Yang and N. Yunes 
for discussions and comments on different aspects of this work.
This work was supported in part by the Natural Sciences and Engineering
Research Council (NSERC) of Canada and the Simons Foundation through
Award SFI-MPSBH-00012593-12. LL also thanks financial support via
the Carlo Fidani Rainer Weiss Chair at Perimeter Institute and CIFAR.
This research was supported in part by Perimeter Institute for Theoretical
Physics. Research at Perimeter Institute is supported in part by the
Government of Canada through the Department of Innovation, Science and
Economic Development and by the Province of Ontario through the Ministry
of Colleges and Universities.

\appendix

\section{Conventions and the $3+1$ Bianchi Equations}
\label{app:bianchi-conventions}
 We use  
  metric signature $(-, +, +, +)$ and define the Riemann curvature tensor via the Ricci identity for an arbitrary one-form $\omega_c$:
\begin{align}
    \left( \nabla_a \nabla_b - \nabla_b \nabla_a \right) \omega_c = R_{abc}{}^{d} \omega_d \,,
\end{align}
with the Ricci tensor given by the contraction $R_{ab} = R_{acb}{}^{c}$. The spacetime volume form is denoted by $\epsilon_{abcd}$, with the orientation set by $\epsilon_{0123} = \sqrt{-g}$. In vacuum, we take 
the background spacetime satisfying Einstein field equations with a cosmological constant $\Lambda$:
\begin{align}
    G_{ab} + \Lambda g_{ab} = 0 \,,
\end{align}
which implies the Ricci tensor is Einstein, $R_{ab} = \Lambda g_{ab}$ and one can derive the Bianchi identity:
\begin{align}
    \nabla^d C_{abcd} = 0 \,.
    \label{eq:app-sourcefree-weyl}
\end{align}
So, while the cosmological constant dynamically alters the background metric and the associated radiative spectrum, it does not introduce local source terms into the spatial projection of the Weyl system.

\subsection{Spatial Projection and Kinematical Decomposition}
\label{app:spatial-projection}

Let $u^a$ be a future-directed, unit timelike vector field ($u^a u_a = -1$). The symmetric tensor field,
\begin{align}
    \gamma_{ab} = g_{ab} + u_a u_b \,,
\end{align}
acts as the spatial projector orthogonal to $u^a$. 
Angular brackets denote the projected, symmetric, and trace-free (STF) part of a spatial tensor:
\begin{align}
    T_{\langle ab \rangle} = \left( \gamma_{(a}{}^c \gamma_{b)}{}^d - \frac{1}{3} \gamma_{ab} \gamma^{cd} \right) T_{cd} \,.
\end{align}

The kinematical properties of the observer congruence $u^a$ are defined by the standard covariant decomposition~\cite{WaldGR}:
\begin{align}
    \nabla_a u_b = -u_a a_b + \frac{1}{3} \theta \gamma_{ab} + \sigma_{ab} + \omega_{ab} \,,
\end{align}
where $a_a = u^b \nabla_b u_a$ is the acceleration vector, $\theta = \nabla_a u^a$ is the expansion scalar, $\sigma_{ab} = D_{\langle a} u_{b \rangle}$ is the symmetric shear tensor, and $\omega_{ab} = D_{[a} u_{b]}$ is the antisymmetric vorticity tensor. The spatial Levi-Civita tensor is defined via the contraction:
\begin{align}
    \epsilon_{abc} \equiv u^d \epsilon_{dabc} \,,
    \label{eq:spatial-epsilon}
\end{align}
allowing the vorticity tensor to be mapped to the vorticity vector $\omega^a \equiv \frac{1}{2} \epsilon^{abc} \omega_{bc}$, such that $\omega_{ab} = \epsilon_{abc} \omega^c$.

\subsection{$3+1$ Bianchi Evolution Equations}
\label{app:three-plus-one-bianchi}

With the definitions above, both $E_{ab}$ and $B_{ab}$ are spatial, symmetric, and trace-free. Projecting the source-free Weyl equation $\nabla^d C_{abcd} = 0$ along and orthogonal to $u^a$ yields the $3+1$ system. In a vacuum Einstein spacetime, the dynamical evolution equations for the Weyl curvature are~\cite{Maartens1998,1996CQGra..13.1451F}:
\begin{align}
    \dot E_{ ab} - (\operatorname{curl} B)_{ab}
    &=
    -\theta E_{ab} + 3\sigma_{c\langle a}E_{b\rangle}{}^c \nonumber \\
    &\quad - \epsilon_{cd\langle a} \omega^c E_{b\rangle}{}^d + 2\epsilon_{cd\langle a} a^c B_{b\rangle}{}^d \,,
    \label{eq:app-E-evolution} \\
    \dot B_{ ab} + (\operatorname{curl} E)_{ab}
    &=
    -\theta B_{ab} + 3\sigma_{c\langle a}B_{b\rangle}{}^c \nonumber \\
    &\quad - \epsilon_{cd\langle a} \omega^c B_{b\rangle}{}^d - 2\epsilon_{cd\langle a} a^c E_{b\rangle}{}^d \,,
    \label{eq:app-B-evolution}
\end{align}
where the projected temporal derivatives are defined as:
\begin{align}
    \dot E_{ ab} \equiv \gamma_{\langle a}{}^c \gamma_{b\rangle}{}^d u^e \nabla_e E_{cd} \,, \qquad 
    \dot B_{ ab} \equiv \gamma_{\langle a}{}^c \gamma_{b\rangle}{}^d u^e \nabla_e B_{cd} \,.
\end{align}
The spatial curl of a symmetric trace-free spatial tensor $T_{ab}$ is defined as:
\begin{align}
    (\operatorname{curl} T)_{ab} \equiv \epsilon_{cd\langle a} D^c T_{b\rangle}{}^d \,.
    \label{eq:curl-def-app}
\end{align}

Using the vorticity tensor $\omega_{ab} = \epsilon_{abc} \omega^c$, the vorticity terms in \cref{eq:app-E-evolution,eq:app-B-evolution} can be written in STF form as:
\begin{align}
    \epsilon_{cd\langle a} \omega^c E_{b\rangle}{}^d = -E_{c\langle a} \omega_{b\rangle}{}^c \,, \qquad 
    \epsilon_{cd\langle a} \omega^c B_{b\rangle}{}^d = -B_{c\langle a} \omega_{b\rangle}{}^c \,.
    \label{eq:app-vorticity-schematic}
\end{align}
This is the term that cancels identically in the enstrophy balance (\cref{subsec:vorticity-cancellation}).

The spatial projection also generates a pair of constraint equations:
\begin{align}
    D^b E_{ab} = \mathcal{C}^E_a[E, B; \sigma, \omega] \,, \qquad 
    D^b B_{ab} = \mathcal{C}^B_a[E, B; \sigma, \omega] \,,
\end{align}
where $\mathcal{C}^E_a$ and $\mathcal{C}^B_a$ are local algebraic source terms coupling the Weyl curvature to the kinematical variables of the congruence. Because the derivation of the quadratic enstrophy balance relies solely on the evolution equations \cref{eq:app-E-evolution,eq:app-B-evolution}, the balance is an identity for fields satisfying the full system, including these constraints.

\subsection{Curl Integration by Parts}
\label{app:curl-identity}

Let $S_{ab}$ and $T_{ab}$ be spatial, symmetric, trace-free tensors. Applying the curl definition in \cref{eq:curl-def-app} and noting that the contraction of any spatial symmetric tensor with the antisymmetric indices of the Levi-Civita tensor allows the projection brackets to be omitted, we have:
\begin{align}
    S^{ab} (\operatorname{curl} T)_{ab} = S^{ab} \epsilon_{cda} D^c T_b{}^d \,.
\end{align}
Integrating this product over a spatial domain $\mathcal{D} \subset \Sigma$ with metric determinant $\gamma$, and using the metric compatibility of the spatial volume form ($D_a \epsilon_{bcd} = 0$), we obtain:
\begin{align}
    \int_{\mathcal{D}} S^{ab} (\operatorname{curl} T)_{ab} \sqrt{\gamma} \, d^3x 
    &= 
    \int_{\mathcal{D}} T^{ab} (\operatorname{curl} S)_{ab} \sqrt{\gamma} \, d^3x \nonumber \\
    &\quad + 
    \int_{\partial\mathcal{D}} n^c \epsilon_{cda} S^{ab} T_b{}^d \sqrt{q} \, d^2x \,,
    \label{eq:app-curl-ibp}
\end{align}
where $n^a$ is the outward-pointing unit normal to the boundary $\partial\mathcal{D}$ within the slice $\Sigma$, and $q$ is the determinant of the induced metric on $\partial\mathcal{D}$. 

Thus, up to a boundary flux term, the spatial curl operator is self-adjoint on the space of symmetric, trace-free spatial tensors:
\begin{align}
    \int_{\mathcal{D}} S^{ab} (\operatorname{curl} T)_{ab} \sqrt{\gamma} \, d^3x 
    &= 
    \int_{\mathcal{D}} T^{ab} (\operatorname{curl} S)_{ab} \sqrt{\gamma} \, d^3x \nonumber \\
    &\qquad\qquad \pmod{\partial\mathcal{D}} \,.
\end{align}

\subsection{Quadratic Balance and the Exact Vorticity Identity}
\label{app:B2-balance}

Contracting the magnetic Weyl evolution equation \cref{eq:app-B-evolution} with $B^{ab}$ yields, with $\tau$ the proper time along $u^a$:
\begin{multline}
    \frac{1}{2} \frac{d}{d\tau} \left( B_{ab} B^{ab} \right)
    =
    -B^{ab} (\operatorname{curl} E)_{ab} - \theta B_{ab} B^{ab} \\
    + 3 B^{ab} \sigma_{c\langle a} B_{b\rangle}{}^c + B^{ab} \left( B_{c\langle a} \omega_{b\rangle}{}^c \right)
    - 2 B^{ab} \epsilon_{cd\langle a} a^c E_{b\rangle}{}^d \,.
\end{multline}
As shown in \cref{subsec:vorticity-cancellation}, the fourth term on the right-hand side vanishes. This local cancellation requires no assumption on the metric, its symmetries, or the boundary conditions.

Removing this vanishing term, the quadratic local balance equation simplifies to:
\begin{multline}
    \frac{1}{2} \frac{d}{d\tau} \left( B_{ab} B^{ab} \right)
    =
    -B^{ab} (\operatorname{curl} E)_{ab} - \theta B_{ab} B^{ab} \\
    + 3 B^{ab} \sigma_{c\langle a} B_{b\rangle}{}^c - 2 B^{ab} \epsilon_{cd\langle a} a^c E_{b\rangle}{}^d \,.
    \label{eq:app-local-B2-balance-reduced}
\end{multline}

For perturbations about a stationary background, we define the radiative magnetic Weyl enstrophy on a spatial subregion $\mathcal{D}$ as:
\begin{align}
    \delta\mathcal{Z}_{\mathcal{D}}(t) \equiv \int_{\mathcal{D}} \delta B_{ab} \delta B^{ab} \sqrt{\gamma} \, d^3x \,.
\end{align}
For the full fields, transport of the spatial measure gives, with the positive-$K_{ij}$ convention of \cref{app:zamo-extrinsic},
\begin{multline}
 \frac{d\mathcal Q_B}{dt}
 =\int_{\mathcal D}\alpha\left[-2B^{ab}(\curl E)_{ab}-\theta B^2
 +6B^{ab}\sigma_{ca}B_b{}^c\right.\\
 \left.-4B^{ab}\epsilon_{cda}a^cE_b{}^d\right]dV
 +\oint_{\partial\mathcal D}\beta^as_aB^2\,dA.
 \label{eq:app-integrated-B2}
\end{multline}
The coefficient of the integrated expansion term is $-1$, since the volume element contributes $+\theta B^2$. For a projected perturbative norm, one must linearize the operators and the observer as well as the fields. Their contractions with background curvature give the $\mathcal J^{E,B}$ terms of \cref{subsec:quadratic-balance}, at the same quadratic order as the displayed kinematic terms.

\section{Newman--Penrose Scalars and the STF Weyl Basis}
\label{app:np-stf}
The spatial, symmetric, trace-free (STF) electric--magnetic Weyl decomposition can be mapped to the Newman--Penrose (NP) formalism~\cite{Newman:1962cia}. 
We assume $\tilde n^a$ is the future-directed, unit timelike vector field defining a $3+1$ foliation, and let $(\hat r^a, \hat\theta^a, \hat\varphi^a)$ form an orthonormal spatial triad. We construct a null tetrad $\{\ell^a, n^a, m^a, \bar m^a\}$ adapted to the radial direction $\hat r^a$ as:
\begin{gather}
    \ell^a = \frac{1}{\sqrt{2}}\left(\tilde n^a + \hat r^a\right), \qquad
    n^a = \frac{1}{\sqrt{2}}\left(\tilde n^a - \hat r^a\right), \\
    m^a = \frac{1}{\sqrt{2}}\left(\hat\theta^a + i\hat\varphi^a\right). \nonumber
\end{gather}
(with the only non-zero inner products $\ell^a n_a = -1$ and $m^a \bar m_a = 1$), yielding the spacetime metric reconstruction:
\begin{align}
    g_{ab} = -2\ell_{(a}n_{b)} + 2m_{(a}\bar m_{b)} \,.
\end{align}
With these conventions, the five complex NP Weyl scalars are defined by~\cite{Newman:1962cia}:
\begin{align}
    \Psi_0 &= -C_{abcd}\ell^a m^b\ell^c m^d \,, \\
    \Psi_1 &= -C_{abcd}\ell^a n^b\ell^c m^d \,, \\
    \Psi_2 &= -C_{abcd}\ell^a m^b\bar m^c n^d \,, \\
    \Psi_3 &= -C_{abcd}\ell^a n^b\bar m^c n^d \,, \\
    \Psi_4 &= -C_{abcd}n^a\bar m^b n^c\bar m^d \,.
\end{align}
For notational convenience, linearized (first-order) perturbations of any quantity $X$ are denoted interchangeably as $\delta X$ or $X^{(1)}$; in particular, we write $\delta\Psi_i \equiv \Psi_i^{(1)}$ for the perturbed Weyl scalars, and $\delta E_{ab} \equiv E_{ab}^{(1)}$, $\delta B_{ab} \equiv B_{ab}^{(1)}$.

{\em From the Radiative Tetrad to the ZAMO Frame.}
By ``Kinnersley gauge'' we mean a tetrad based on the principal null directions of the type-D background~\cite{Kinnersley:1969zz}, with first-order null rotations chosen so that $\delta\Psi_1=\delta\Psi_3=0$~\cite{Newman:1962cia}. The extreme scalars $\delta\Psi_0$ and $\delta\Psi_4$ are invariant under infinitesimal coordinate-gauge and tetrad rotations about this background, but carry nonzero spin--boost weights. We fix the residual spin--boost normalization when reconstructing the transverse curvature. In Kerr the standard Kinnersley tetrad is related to the symmetric Carter tetrad by
\[
 \ell_C=A_{\rm b}\ell_K,\qquad n_C=A_{\rm b}^{-1}n_K,
 \qquad A_{\rm b}=\sqrt{\frac{\Delta}{2\rho^2}}.
\]
Thus $\Psi_{4,C}=A_{\rm b}^{-2}\Psi_{4,K}$ and $\Psi_{0,C}=A_{\rm b}^{2}\Psi_{0,K}$, with a transverse spin phase fixed consistently. This radial boost is singular at the horizon. The following bounded azimuthal boost is a separate transformation.

The rest frame of the symmetric tetrad is the Carter observer, $u_C^a \propto \xi^a + \omega_C \psi^a$, with $\omega_C = a/(r^2+a^2)$ and $\xi^a$, $\psi^a$ the stationary and axial Killing vectors, whereas the ZAMO angular velocity is $\omega_Z = 2Mar/A$; here $\Delta = r^2 - 2Mr + a^2$, $\rho^2 = r^2 + a^2\cos^2\theta$, and $A = (r^2+a^2)^2 - a^2 \Delta \sin^2\theta$ are the Boyer--Lindquist metric functions. Their difference simplifies exactly,
\begin{align}
    \omega_Z - \omega_C = -\frac{a\,\Delta\,\rho^2}{A\,(r^2+a^2)},
\end{align}
corresponding to a relative azimuthal velocity
\begin{align}
    v = \frac{|\omega_Z - \omega_C|\sqrt{g_{\phi\phi}}}{\alpha}
      = \frac{a\sqrt{\Delta}\,\sin\theta}{r^2+a^2},
    \label{eq:carter-zamo-boost}
\end{align}
which vanishes on the horizon (where both congruences corotate at $\Omega_H$), on the axis, and at spatial infinity, and attains its global maximum $v = (\sqrt{2}-1)/2 \simeq 0.21$ at $r = (1+\sqrt{2})M$ for extremal Kerr; across spins, $v_{\max} \simeq 0.2\,\chi$ to within a few percent. Reconstructing $\delta E_{ab}$ and $\delta B_{ab}$ in the ZAMO orthonormal frame therefore requires an azimuthal boost of magnitude \cref{eq:carter-zamo-boost}. Because this boost tilts the frame out of the $\ell$--$n$ plane, it is not a spin--boost transformation: in the ZAMO-adapted tetrad the background scalars $\Psi_1$ and $\Psi_3$ are $O(v)\,\Psi_2$ rather than zero, and the projected field acquires radial-transverse components. A comparison with $\tfrac12|\Psi_{4,C}|^2$ has bounded azimuthal-boost corrections; a comparison with the unrescaled Kinnersley scalar must also include $A_{\rm b}^{-4}$. The vanishing azimuthal velocity at the horizon does not make either singular stationary frame horizon-regular.

For a type-N field, $E^2=B^2$ and $E_{ab}B^{ab}=0$ hold for every timelike observer. The expression $B^2=\tfrac12|\Psi_4|^2$ additionally requires a tetrad adapted to that observer and propagation direction. Superposing distinct propagation directions generally destroys the type-N identity; equality after integration or averaging is a separate property.

\subsection{Complex STF Weyl Tensor}
\label{app:complex-stf-weyl}

We combine the spatial electric and magnetic Weyl tensors into a single complex spatial STF tensor:
\begin{align}
    \mathcal Q_{ab} = E_{ab} + iB_{ab} \,.
    \label{eq:app-Q-def}
\end{align}
In the orthonormal spatial basis $(\hat\theta^a, \hat\varphi^a, \hat r^a)$, its components are:
\begin{align}
    \mathcal Q_{\theta\theta} &= -\Psi_2 + \frac{1}{2}\left(\Psi_0 + \Psi_4\right), \\
    \mathcal Q_{\varphi\varphi} &= -\Psi_2 - \frac{1}{2}\left(\Psi_0 + \Psi_4\right), \\
    \mathcal Q_{rr} &= 2\Psi_2 \,, \\
    \mathcal Q_{\theta\varphi} &= \frac{i}{2}\left(\Psi_0 - \Psi_4\right), \\
    \mathcal Q_{\theta r} &= \Psi_1 - \Psi_3 \,, \\
    \mathcal Q_{\varphi r} &= i\left(\Psi_1 + \Psi_3\right).
\end{align}
The trace vanishes explicitly, $\mathcal Q_{\theta\theta} + \mathcal Q_{\varphi\varphi} + \mathcal Q_{rr} = 0$.

It is useful to introduce the following orthonormal spatial STF basis tensors:
\begin{align}
    P_{ab} &= \hat r_a \hat r_b - \frac{1}{2} \left( \hat\theta_a\hat\theta_b + \hat\varphi_a\hat\varphi_b \right), \\
    Q^{(+)}_{ab} &= \frac{1}{2} \left( \hat\theta_a\hat\theta_b - \hat\varphi_a\hat\varphi_b \right), \qquad Q^{(\times)}_{ab} = \hat\theta_{(a}\hat\varphi_{b)} \,, \\
    L^{(\theta)}_{ab} &= \hat r_{(a}\hat\theta_{b)} \,, \qquad\qquad\qquad\quad\ \, L^{(\varphi)}_{ab} = \hat r_{(a}\hat\varphi_{b)} \,.
\end{align}
The non-zero quadratic contractions of this basis are:
\begin{align}
    P_{ab}P^{ab} = \frac{3}{2} \,, \qquad 
    Q^{(+)ab}Q^{(+)}_{ab} = Q^{(\times)ab}Q^{(\times)}_{ab} \\
    = L^{(\theta)ab}L^{(\theta)}_{ab} = L^{(\varphi)ab}L^{(\varphi)}_{ab} = \frac{1}{2} \,,
    \label{eq:app-STF-basis-norms}
\end{align}
with all cross-contractions between distinct basis elements vanishing. In this basis, the complex Weyl tensor $\mathcal Q_{ab}$ decomposes as:
\begin{multline}
    \mathcal Q_{ab} = 2\Psi_2 P_{ab} + (\Psi_0 + \Psi_4)Q^{(+)}_{ab} + i(\Psi_0 - \Psi_4)Q^{(\times)}_{ab} \\
    + 2(\Psi_1 - \Psi_3)L^{(\theta)}_{ab} + 2i(\Psi_1 + \Psi_3)L^{(\varphi)}_{ab} \,.
    \label{eq:app-Q-STF-decomp}
\end{multline}
This decomposition makes the physical content of the Weyl tensor clear: the scalars $\Psi_0$ and $\Psi_4$ carry the transverse projections; $\Psi_2$ carries the Coulombic background parameters and can also have dynamical perturbations; and $\Psi_1$ and $\Psi_3$ isolate longitudinal or mixed projections whose role depends on the coordinate gauge and physical regime~\cite{Newman:1962cia,NewmanTod1980}.

\subsection{Quadratic Weyl Invariants}
\label{app:quadratic-np}

Using the orthogonality of the STF basis and the norms in \cref{eq:app-STF-basis-norms}, the positive-definite Bel--Robinson super-energy density associated with the observer $u^a$ is:
\begin{multline}
    E_{ab}E^{ab} + B_{ab}B^{ab} = \mathcal Q_{ab}\overline{\mathcal Q}^{ab} \\
    = |\Psi_0|^2 + |\Psi_4|^2 + 4|\Psi_1|^2 + 4|\Psi_3|^2 + 6|\Psi_2|^2 \,.
\end{multline}
This positive-definite norm receiving contributions from all five Weyl scalars is a super-energy density, not an ordinary localized gravitational energy density.

The complex quadratic contraction is given by:
\begin{align}
    \mathcal Q_{ab}\mathcal Q^{ab} = \left( E_{ab}E^{ab} - B_{ab}B^{ab} \right) + 2iE_{ab}B^{ab} = 2 I \,,
\end{align}
where the complex scalar invariant $I$ is defined as:
\begin{align}
    I \equiv \Psi_0\Psi_4 - 4\Psi_1\Psi_3 + 3\Psi_2^2 \,.
\end{align}
Equivalently, extracting the real and imaginary parts of $I$ yields the individual quadratic invariants:
\begin{align}
    E_{ab}E^{ab} - B_{ab}B^{ab} &= 2\operatorname{Re} I \,, \\
    E_{ab}B^{ab} &= \operatorname{Im} I \,.
\end{align}

\subsection{Pure Radiative Sectors}
\label{app:pure-radiative-sectors}

Let parenthesized numerical subscripts denote the restriction to a subsector where only the indicated Weyl scalar is non-zero; thus $\mathcal Q_{ab(s)}$, $E_{ab(s)}$, and $B_{ab(s)}$ represent the complex Weyl tensor and its real/imaginary parts evaluated with $\Psi_s$ alone.

For a pure outgoing transverse radiative mode ($\Psi_4 \neq 0$), the complex STF tensor reduces to:
\begin{align}
    \mathcal Q_{ab(4)} = \Psi_4 \left( Q^{(+)}_{ab} - iQ^{(\times)}_{ab} \right).
\end{align}
Writing the complex scalar as $\Psi_4 = e_+ + i e_\times$, we find:
\begin{align}
    E_{ab(4)} &= e_+ Q^{(+)}_{ab} + e_\times Q^{(\times)}_{ab} \,, \\
    B_{ab(4)} &= e_\times Q^{(+)}_{ab} - e_+ Q^{(\times)}_{ab} \,.
\end{align}
Comparing these expressions shows that $B_{ab(4)}$ is obtained from $E_{ab(4)}$ by a $90^\circ$ rotation in polarization space. Utilizing the basis norms in \cref{eq:app-STF-basis-norms}, we obtain:
\begin{align}
    E_{ab(4)}E_{(4)}^{ab} = B_{ab(4)}B_{(4)}^{ab} = \frac{1}{2} |\Psi_4|^2 \,, \qquad E_{ab(4)}B_{(4)}^{ab} = 0 \,,
    \label{eq:app-pure-Psi4-norms}
\end{align}
which implies:
\begin{align}
    \mathcal Q_{ab(4)}\overline{\mathcal Q}_{(4)}^{ab} = |\Psi_4|^2 \,, \qquad \mathcal Q_{ab(4)}\mathcal Q_{(4)}^{ab} = 0 \,.
\end{align}
The vanishing of the second contraction corresponds to the standard statement that a pure type-N radiative field has vanishing scalar Weyl invariants ($I = 0$).

For a pure ingoing transverse radiative mode ($\Psi_0 \neq 0$), the complex STF tensor is:
\begin{align}
    \mathcal Q_{ab(0)} = \Psi_0 \left( Q^{(+)}_{ab} + i Q^{(\times)}_{ab} \right),
\end{align}
which similarly yields:
\begin{align}
    E_{ab(0)}E_{(0)}^{ab} = B_{ab(0)}B_{(0)}^{ab} = \frac{1}{2} |\Psi_0|^2 \,, \qquad E_{ab(0)}B_{(0)}^{ab} = 0 \,.
\end{align}

These formulas give equal electric and magnetic norms for a pure type-N field. The mixed contraction vanishes for that field; interference between waves propagating in different directions can make it nonzero.

\subsection{Coulomb Sector}
\label{app:coulomb-sector}

For a purely Coulombic field, only the $\Psi_2$ component is non-zero:
\begin{align}
    \mathcal Q_{ab(2)} = 2\Psi_2 P_{ab} \,.
\end{align}
Using the basis norm $P_{ab}P^{ab} = 3/2$, the quadratic invariants evaluate to:
\begin{align}
    E_{ab(2)}E_{(2)}^{ab} + B_{ab(2)}B_{(2)}^{ab} = 6|\Psi_2|^2 \,,
\end{align}
and
\begin{align}
    E_{ab(2)}E_{(2)}^{ab} - B_{ab(2)}B_{(2)}^{ab} &= 6\operatorname{Re}\left(\Psi_2^2\right), \\
    E_{ab(2)}B_{(2)}^{ab} &= 3\operatorname{Im}\left(\Psi_2^2\right). \nonumber
\end{align}
\rev{A stationary parameter perturbation in this sector changes the background mass or angular momentum. A dynamical $\delta\Psi_2$ can also accompany radiation. The projection used to separate these components from the transverse radiative norm is described in \cref{app:coulomb-nonradiative}.}

\subsection{Frequency Weighting}
\label{app:NP-frequency-weighting}

\rev{The enstrophy spectral weighting of the quadratic Weyl structure, namely that radiative modes scale as $\Psi_{\mathrm{rad}}^{(k)} \sim \omega_k^2 h_k$, so that $\int \delta B^2 \sqrt{\gamma} \, d^3x / W_k \sim \omega_k^2$, is analyzed in \cref{sec:spectral-weighting}; \cref{app:spectral-ratio-proof} shows where it holds exactly and where it receives near-horizon corrections.}
We note that the pseudoscalar functional
constructed with $E_{ab}^{(k)} B_{(k)}^{ab}$
carries the exact same frequency scaling ($\omega_k^2$) as the enstrophy. However, because $\mathcal H$ is not positive-definite, it cannot serve as the constraining enstrophy required for the Fj{\o}rtoft dual-cascade constraint.

\section{Non-radiative sectors and their control}
\label{app:coulomb-nonradiative}

\rev{The Fj{\o}rtoft argument involves only radiative modes, but a generic perturbation of a black hole or AdS background also excites Coulombic, longitudinal, and gauge degrees of freedom. These sectors do not obey the spectral scaling $\mathcal{Z}_k \sim \omega_k^2 W_k$. Here we show how they can be separated from the radiative enstrophy, and under which conditions their contamination of the magnetic norm and of the balance law is under control.}

\subsection{The radiative/Coulomb split}
\label{app:rad-vs-coulomb}

\rev{In the NP tetrad of \cref{app:np-stf}, the complex STF Weyl tensor $\mathcal{Q}_{ab}=E_{ab}+iB_{ab}$ decomposes as in \cref{eq:app-Q-STF-decomp}, with $\Psi_0$ and $\Psi_4$ carrying the transverse radiative data, $\Psi_2$ the Coulombic data, and $\Psi_1$, $\Psi_3$ the longitudinal or mixed projections.}

Schematically,
\begin{align}
    \delta\mathcal{Q}_{ab} 
    &= 
    \delta\mathcal{Q}_{ab}^{\mathrm{rad}} 
    + \delta\mathcal{Q}_{ab}^{\mathrm{C}} 
    + \delta\mathcal{Q}_{ab}^{\mathrm{L}} 
    + \delta\mathcal{Q}_{ab}^{\mathrm{gauge}},
\end{align}
with
\begin{align}
    \delta\mathcal{Q}_{ab}^{\mathrm{rad}} &\leftrightarrow \delta\Psi_0,\delta\Psi_4, \\
    \delta\mathcal{Q}_{ab}^{\mathrm{C}} &\leftrightarrow \delta\Psi_2, \\
    \delta\mathcal{Q}_{ab}^{\mathrm{L}} &\leftrightarrow \delta\Psi_1,\delta\Psi_3.
\end{align}
\rev{This is not an invariant decomposition; it is the operational one used in perturbation theory once the background, the observer and the mode basis are fixed. On a type D background, $\delta\Psi_0$ and $\delta\Psi_4$ are the standard gauge-invariant radiative data.}

We define the radiative enstrophy on the background geometry as
\begin{align}
    \delta\mathcal{Z}_{\mathrm{rad}} 
    = 
    \int_\Sigma 
    \delta B_{ab}^{\mathrm{rad}} \delta B_{\mathrm{rad}}^{ab} 
    \sqrt{\bar\gamma}\,d^3x.
    \label{eq:app-Zrad-def}
\end{align}
This is the analogue of the quadratic wave energy. It is \emph{not} the full variation of the nonlinear functional
$\mathcal{Z}[g]=\int_\Sigma B_{ab}B^{ab}\sqrt{\gamma}\,d^3x$. 
Schematically, with $\delta B=O(\epsilon)$,
\begin{multline}
 \mathcal Z[g,u]=\bar{\mathcal Z}+\mathcal Z^{(1)}
 +\int_\Sigma\delta B_{ab}\delta B^{ab}\,dV_{\bar\gamma}\\
 +\mathcal Z^{(2)}_{\rm background,\,metric,\,observer,\,measure}
 +O(\epsilon^3).
 \label{eq:app-Z-expansion}
\end{multline}
The last term at quadratic order includes $2\bar B^{ab}B^{(2)}_{ab}$ and variations of the contractions and measure. These are not generically removed by phase averaging. The Fj{\o}rtoft argument uses the separately defined positive radiative norm, not this full variation.

\subsection{The Coulomb sector}
\label{app:pure-coulomb}

A purely Coulombic field has only $\Psi_2$:
$\mathcal{Q}_{ab}^{\mathrm{C}}=2\Psi_2 P_{ab}$. With 
$\Psi_2=\Psi_2^{\mathrm{R}}+i\Psi_2^{\mathrm{I}}$,
\begin{align}
    E_{ab}^{\mathrm{C}} = 2\Psi_2^{\mathrm{R}} P_{ab}, \qquad
    B_{ab}^{\mathrm{C}} = 2\Psi_2^{\mathrm{I}} P_{ab}.
\end{align}
For a stationary parameter perturbation this contributes to the magnetic norm without a radiative frequency, so the radiation-zone weighting does not apply. A dynamical $\delta\Psi_2$ is instead part of the coupled perturbation and is accounted for through the projected sources.

\rev{For example, in Schwarzschild with static observers $\Psi_2$ is real, so $\bar B_{ab}=0$. In the principal Kerr tetrad $\Psi_2$ is generally complex and $\bar B_{ab}\neq0$. Changing the observer mixes the electric and magnetic fields and can also introduce other NP components. This is stationary curvature, not radiative enstrophy.}

In the Kinnersley tetrad~\cite{Kinnersley:1969zz},
\begin{align}
    \Psi_2 = -\frac{M}{(r - ia\cos\theta)^3}.
\end{align}
A stationary perturbation of mass and spin gives
\begin{align}
    \delta\Psi_2^{\mathrm{charge}} 
    = 
    -\frac{\delta M}{(r - ia\cos\theta)^3} 
    - \frac{3iM\cos\theta\,\delta a}{(r - ia\cos\theta)^4}.
\end{align}
This deforms the background Coulomb curvature. It does not propagate, 
so it must be projected out of the radiative mode sum.

\subsection{Contamination of the magnetic norm}
\label{app:quadratic-contamination}

Decompose the full perturbation as
\begin{align}
    \delta B_{ab} 
    = 
    \delta B_{ab}^{\mathrm{rad}} 
    + \delta B_{ab}^{\mathrm{C}} 
    + \delta B_{ab}^{\mathrm{L}} 
    + \delta B_{ab}^{\mathrm{gauge}}.
\end{align}
The total magnetic norm expands as
\begin{multline}
    \int_\Sigma \delta B_{ab}\delta B^{ab}\sqrt{\bar\gamma}\,d^3x 
    = 
    \delta\mathcal{Z}_{\mathrm{rad}} 
    + \delta\mathcal{Z}_{\mathrm{C}} 
    + \delta\mathcal{Z}_{\mathrm{L}} \\
    + 2\mathcal{X}_{\mathrm{rad,C}} 
    + 2\mathcal{X}_{\mathrm{rad,L}} 
    + 2\mathcal{X}_{\mathrm{C,L}} 
    + \cdots,
    \label{eq:app-Z-sector-expansion}
\end{multline}
with cross-terms like
$\mathcal{X}_{\mathrm{rad,C}}
= \int_\Sigma \delta B_{ab}^{\mathrm{rad}} \delta B_{\mathrm{C}}^{ab}
\sqrt{\bar\gamma}\,d^3x$.

\rev{The enstrophy paired with the radiative energy is $\delta\mathcal{Z}_{\mathrm{rad}}$. In the observer-adapted STF basis, the $P$, $Q$ and longitudinal tensors are orthogonal. Interference between radiative modes is a separate issue, treated in \cref{app:degenerate-fjortoft}.}

\rev{The unprojected magnetic norm approximates $\delta\mathcal{Z}_{\mathrm{rad}}$ when the other contributions are small:}
\begin{align}
    \left| 
    \frac{
        \delta\mathcal{Z}_{\mathrm{C}} 
        + \delta\mathcal{Z}_{\mathrm{L}} 
        + 2\mathcal{X}_{\mathrm{rad,C}} 
        + 2\mathcal{X}_{\mathrm{rad,L}} 
        + \cdots
    }{
        \delta\mathcal{Z}_{\mathrm{rad}}
    } 
    \right| \ll 1.
    \label{eq:app-coulomb-smallness}
\end{align}
\rev{When these contributions are appreciable, the full $B^2$ norm includes sectors with different spectral weights.}

\subsection{Time averages vs. instantaneous norms}
\label{app:instantaneous-averaged}

A real radiative perturbation has both $+\omega$ and $-\omega$ 
components:
\begin{align}
    \delta B_{ab}^{\mathrm{rad}} 
    = 
    b_{ab} e^{-i\omega t} + \bar b_{ab} e^{i\omega t}.
\end{align}
The quadratic integrand is
\begin{align}
    \delta B_{ab}^{\mathrm{rad}} \delta B_{\mathrm{rad}}^{ab} 
    = 
    2 b_{ab}\bar b^{ab} 
    + b_{ab}b^{ab} e^{-2i\omega t} 
    + \bar b_{ab}\bar b^{ab} e^{2i\omega t}.
\end{align}
\rev{Averaging over several periods removes the oscillatory terms and gives the modal enstrophy $2\int_\Sigma b_{ab}\bar b^{ab}\sqrt{\bar\gamma}\,d^3x$.}

The linear term in \cref{eq:app-Z-expansion} behaves similarly: 
for a harmonic mode on a stationary background, it oscillates and 
averages to zero (unless resonant with stationary deformations). 
Together with the cross-mode beat averaging of 
Remark~\ref{rem:coarse-grained-WZ}, this is why the radiative 
enstrophy must be defined as a coarse-grained modal norm, not the 
raw instantaneous variation.

\subsection{Coulomb terms in the balance law}
\label{app:coulomb-balance}

The radiative enstrophy balance takes the form
\begin{align}
    \frac{d}{dt}\delta\mathcal{Z}_{\mathrm{rad}} 
    = 
    \mathcal{R}_{\mathrm{rad}} 
    + \mathcal{R}_{\mathrm{mix}} 
    + \mathcal{F}_{\partial\Sigma}^{\mathcal{Z}} 
    + O(\epsilon^3),
\end{align}
where $\mathcal{R}_{\mathrm{rad}}$ is built from radiative fields alone, 
while $\mathcal{R}_{\mathrm{mix}}$ couples to Coulombic, longitudinal, 
and constraint sectors.

\rev{The projection onto $\Psi_0$ and $\Psi_4$ separates components algebraically, but their evolution can couple them to the other sectors even at linear order. The term $\mathcal{R}_{\mathrm{mix}}$ accounts for this exchange in the chosen background tetrad. We can neglect it over the transfer time if}
\begin{align}
    \left| \int_{t_1}^{t_2} \mathcal{R}_{\mathrm{mix}}\,dt \right| 
    \ll 
    \delta\mathcal{Z}_{\mathrm{rad}}.
\end{align}
Integrating $\mathcal{R}_{\mathrm{mix}}$ gives a direct numerical test of this condition. Before using fixed weights $\Omega_k$, one also checks that their drift is small compared with the weight differences resolved by the transfer.

\rev{Per unit time, the condition reads}
\begin{align}
    |\mathcal{R}_{\mathrm{mix}}| 
    \ll 
    \frac{\delta\mathcal{Z}_{\mathrm{rad}}}{\tau_{\mathrm{nl}}},
\end{align}
with $\tau_{\mathrm{nl}}$ the nonlinear transfer time. A stationary background then sets the mode spectrum, while the explicitly retained perturbative sources determine exchange with the projected sector.

\rev{When this exchange is appreciable, the radiative norm evolves together with the excluded sector.}
\subsection{Implications for the main text regimes}
\label{app:coulomb-regimes}

The role of the Coulomb sector varies by setting. In Kerr and 
Kerr--AdS, the background curvature is Coulombic ($\bar\Psi_2$), 
while radiative perturbations live in $\delta\Psi_0,\delta\Psi_4$. 
Mass and spin changes, as well as dynamical perturbative curvature, enter through $\delta\Psi_2$, but the 
relevant enstrophy is the projected radiative norm, not the raw 
variation of the stationary curvature.

\rev{Global AdS is simpler: $\bar C_{abcd}=0$, so there is no Coulombic background at all, and any Coulombic contamination can only come from the perturbation itself.}

In the fluid/gravity regime, the black brane has non-vanishing Weyl 
curvature. The bulk magnetic norm is expanded in spatial gradients; 
the leading vortical component matches the boundary enstrophy. That component lies in the $P_{ab}$ sector and is retained in the fluid-weighted functional, unlike the earlier radiative projection.

\section{Boundary Fluxes of the Magnetic Weyl Functional}
\label{app:boundary-fluxes}

The Bianchi identities fix the local evolution of the enstrophy density; the integrated quantity, however, changes through boundary fluxes. Whether the balance laws have physical consequences therefore depends on the boundaries: horizons, spatial or null infinity, conformal AdS boundaries, or finite reflecting walls.

Let $\mathcal{D}_t \subset \Sigma_t$ be a spatial domain with outward unit normal $s^a$ tangent to the spatial slice $\Sigma_t$. The magnetic Weyl enstrophy contained within the volume $\mathcal{D}_t$ is defined as:
\begin{align}
    \mathcal{Z}_B[\mathcal{D}_t] = \int_{\mathcal{D}_t} B_{ab} B^{ab} \sqrt{\gamma} \, d^3x \,.
\end{align}
For perturbative applications, the first-order fields and the linearization and projection sources of \cref{subsec:quadratic-balance} must be used together.

\subsection{The Super-Poynting Vector}
\label{app:super-poynting}

Applying the integration-by-parts identity for the spatial curl operator (derived in Appendix~\ref{app:curl-identity}) yields the boundary term:
\begin{align}
    \int_{\mathcal{D}} B^{ab} (\operatorname{curl} E)_{ab} \sqrt{\gamma} \, d^3x 
    &= 
    \int_{\mathcal{D}} E^{ab} (\operatorname{curl} B)_{ab} \sqrt{\gamma} \, d^3x \nonumber \\
    &\quad + 
    \int_{\partial\mathcal{D}} s^c \epsilon_{cda} B^{ab} E_b{}^d \sqrt{q} \, d^2x \,.
    \label{eq:app-BcurlE-ibp}
\end{align}
To clarify the physical significance of the boundary integrand, we define the spatial Bel--Robinson (or super-Poynting) vector as~\cite{BelRobinson}:
\begin{align}
    \mathcal{P}^a \equiv \epsilon^{abc} E_b{}^d B_{cd} \,.
\end{align}
Under the spatial orientation conventions of \cref{eq:spatial-epsilon}, the boundary integrand in \cref{eq:app-BcurlE-ibp} simplifies to:
\begin{align}
    s^c \epsilon_{cda} B^{ab} E_b{}^d = s_a \mathcal{P}^a \,.
\end{align}
Note that a change in the sign convention of the spatial volume form reverses the sign of both the spatial curl and $\mathcal{P}^a$, leaving the overall balance equation invariant.

The boundary flux associated with the curl term is therefore governed by the normal projection of the super-Poynting vector. We define the outward curvature flux through the boundary $\partial\mathcal{D}$ as:
\begin{align}
    \mathcal{F}_{\partial\mathcal{D}}^{\mathrm{BR}} \equiv 2 \int_{\partial\mathcal{D}} s_a \mathcal{P}^a \sqrt{q} \, d^2x \,,
\end{align}
subject to an overall sign convention. Under our convention, a positive value of $\mathcal{F}_{\partial\mathcal{D}}^{\mathrm{BR}}$ signifies a net loss of Bel--Robinson curvature norm from the spatial domain.

\subsection{Magnetic Enstrophy Balance}
\label{app:magnetic-flux-balance}

Recall the magnetic Weyl evolution equation has the schematic form:
\begin{align}
    \dot B_{\langle ab\rangle} + (\operatorname{curl} E)_{ab} = \mathcal{K}^B_{ab} \,,
\end{align}
where $\mathcal{K}^B_{ab}$ represents the local kinematical terms involving spatial expansion, shear, vorticity, and acceleration. For a fixed coordinate domain, \cref{eq:app-integrated-B2} gives the balance including lapse and volume transport. Integrating the curl by parts yields
\begin{align}
 \dot{\mathcal Z}_B
 &=-2\int_{\mathcal D}\alpha E^{ab}(\curl B)_{ab}\,dV
 -2\oint_{\partial\mathcal D}\alpha s_a\mathcal P^a\,dA
 +\mathcal R_B,\label{eq:app-ZB-flux-balance}\\
 \dot{\mathcal Z}_E
 &=2\int_{\mathcal D}\alpha E^{ab}(\curl B)_{ab}\,dV
 +\mathcal R_E.\label{eq:app-ZE-balance}
\end{align}
Here $\mathcal Z_E=\int E^2dV$. The residuals contain kinematic and transport terms; $\mathcal R_B$ also includes the lapse-gradient term from the integration by parts. Adding gives
\[
 \frac{d}{dt}(\mathcal Z_E+\mathcal Z_B)
 =-2\oint\alpha s_a\mathcal P^a\,dA+\mathcal R_E+\mathcal R_B.
\]
For projected perturbations the same organization is possible only with the $\mathcal J^{E,B}$ sources included. If the chosen solution and averaging give $\delta\mathcal Z_E=\delta\mathcal Z_B$, then
\[
 \dot{\delta\mathcal Z}_B
 =-\oint\alpha s_a\delta\mathcal P^a\,dA
 +\tfrac12(\delta\mathcal R_E+\delta\mathcal R_B),
\]
where $\delta\mathcal P^a=\epsilon^{abc}\delta E_b{}^d\delta B_{cd}$. Thus the magnetic norm carries half the super-energy flux. Without equality of the two norms one must keep the electric time derivative explicitly.

\subsection{Radiative Flux Through a Large Sphere}
\label{app:null-infinity-flux}
Consider a large sphere $S_r$ of radius $r$ with outward normal $s^a$. For an outgoing transverse mode, $E_{ab}$ and $B_{ab}$ are related by the duality rotation, and the radial super-Poynting flux therefore satisfies:
\begin{align}
    |s_a \mathcal{P}^a| = E_{ab}^{\mathrm{rad}} E_{\mathrm{rad}}^{ab} = B_{ab}^{\mathrm{rad}} B_{\mathrm{rad}}^{ab} \,,
\end{align}
which, in terms of the outgoing Newman--Penrose scalar, scales as:
\begin{align}
    E_{ab}^{\mathrm{rad}} E_{\mathrm{rad}}^{ab} = B_{ab}^{\mathrm{rad}} B_{\mathrm{rad}}^{ab} \sim |\Psi_4|^2 \,.
\end{align}

At future null infinity $\mathscr{I}^+$, the radiation field peels as:
\begin{align}
    \Psi_4 = \frac{\psi_4(u,\theta,\varphi)}{r} + \mathcal{O}(r^{-2}) \,,
\end{align}
ensuring that the curvature flux through $S_r$ has a well-defined, finite limit:
\begin{align}
    \mathcal{F}_{\mathscr{I}^+}^{\mathcal{Z}}(u)\sim\oint |\psi_4(u,\theta,\varphi)|^2\,d\Omega \,.
\end{align}
The exact numerical coefficient is determined by the normalization of the Newman--Penrose tetrad and of $\mathcal{P}^a$, but the overall frequency weighting remains unaffected.

Using the relation to the gravitational-wave strain,
\begin{align}
    \Psi_4 = -\ddot h_+ + i\ddot h_\times \,,
\end{align}
a mode with frequency $\omega$ yields an enstrophy flux scaling as:
\begin{align}
    \mathcal{F}_{\mathscr{I}^+}^{\mathcal{Z}}(\omega) \sim \omega^4 |h_\omega|^2 \,.
\end{align}
By contrast, the canonical gravitational-wave energy flux scales as:
\begin{align}
    \mathcal{F}_{\mathscr{I}^+}^{W}(\omega) \sim \omega^2 |h_\omega|^2 \,.
\end{align}
Thus, we obtain the spectral ratio:
\begin{align}
    \frac{\mathcal{F}_{\mathscr{I}^+}^{\mathcal{Z}}(\omega)}{\mathcal{F}_{\mathscr{I}^+}^{W}(\omega)} \sim \omega^2 \,,
\end{align}
which represents the boundary flux analog of the bulk spectral frequency weighting used throughout the main text.

\subsection{Horizon Flux}
\label{app:horizon-flux}

The conversion between horizon curvature and energy flux follows from the shear equation. Let $\tilde\omega=\omega-m\Omega_H$ and let $\sigma_{\mathcal H}$ be the shear of a generator normalized as $\chi^a=\xi^a+\Omega_H\psi^a$. The linearized horizon shear equation has the form
\[
 (\mathcal L_\chi-\kappa_+)\sigma_{\mathcal H}
 =-\Psi_0^{\rm reg},
\]
up to the orientation convention. For a real-frequency separated wave,
\[
 |\Psi_0^{\rm reg}|^2
 =(\tilde\omega^2+\kappa_+^2)|\sigma_{\mathcal H}|^2.
\]
The horizon-generator energy flux is proportional to $|\sigma_{\mathcal H}|^2$, whereas
\[
 \mathcal F_{\mathcal H}^{(\xi)}
 =\frac{\omega}{\tilde\omega}\mathcal F_{\mathcal H}^{(\chi)}
\]
away from threshold, with the threshold limit taken using regular transmitted data~\cite{TeukolskyPress1974}. Thus converting curvature into energy flux brings in both $\tilde\omega$ and $\kappa_+$; the normalization of $\Psi_0^{\rm reg}$ is fixed by the chosen tetrad. The stationarity-energy flux changes sign in the superradiant regime, while the squared shear remains nonnegative.

To use this conversion in the spacelike magnetic-norm balance, one also includes the lapse and the transport of the integration boundary, as in \cref{subsec:quadratic-balance}. The quantity $\mathcal F_{\mathcal H}^{\mathcal Z}$ in the main text denotes that boundary term. The QNM decay rate, in turn, gives the sum of all terms in the linear balance.

\subsection{AdS Boundary Conditions}
\label{app:ads-boundary-flux}

In asymptotically Anti--de Sitter spacetimes, the conformal boundary is timelike rather than null. The local flux density remains:
\begin{align}
    s_a \mathcal{P}^a = s_a \epsilon^{abc} E_b{}^d B_{cd} \,,
\end{align}
but the physical behavior is determined by the chosen boundary conditions at conformal infinity.

Reflecting gravitational boundary conditions set the symplectic-energy flux through the conformal boundary to zero. A vanishing curvature flux must also be checked for the chosen falloff, observer and normalization; it does not follow solely from vanishing symplectic flux. In the confined balances used here we assume boundary conditions for which
\[
 \mathcal F_{\mathscr I_{\rm AdS}}^{\mathcal Z}=0
\]
at the retained order. A standing mode may have local oscillatory flux while its integrated, phase-averaged flux vanishes. Its norm includes the ingoing--outgoing interference.

If transparent, dissipative, or mixed boundary conditions are imposed, then $\mathcal{F}_{\mathscr{I}_{\mathrm{AdS}}}^{\mathcal{Z}}$ is non-vanishing. In this case, the bulk enstrophy is not conserved even in the absence of a black hole horizon, because curvature is exchanged with the boundary sector. This represents a physical choice of boundary physics rather than a failure of the local enstrophy structure.

For a Kerr--AdS exterior with reflecting boundary conditions, the outer flux vanishes while the horizon absorption remains:
\begin{align}
    \frac{d}{dt} \delta\mathcal{Z}_{\mathrm{ext}}^{\mathrm{KAdS}} = \delta\mathcal{R}_{\mathrm{bulk}}^{\mathrm{KAdS}} - \mathcal{F}_{\mathcal{H}^+}^{\mathcal{Z}} + \mathcal{O}(\epsilon^3) \,.
\end{align}
For horizonless global AdS with reflecting boundary conditions, the balance reduces strictly to:
\begin{align}
    \frac{d}{dt} \delta\mathcal{Z}_{\mathrm{AdS}} = \delta\mathcal{R}_{\mathrm{bulk}}^{\mathrm{AdS}} + \mathcal{O}(\epsilon^3) \,.
\end{align}
\rev{In this case the balance is set by the bulk terms alone; whether $\mathcal{Z}$ then constrains the transfer depends on the interactions (\cref{subsec:fjortoft-interpretation}).}

\subsection{Finite Reflecting Domains}
\label{app:finite-reflecting-domain}

The same logic applies to finite spatial domains bounded by artificial reflecting boundaries. Let the boundary $\partial\mathcal{D}$ be a timelike wall. A perfect curvature-reflecting boundary condition imposes:
\begin{align}
    s_a \delta\mathcal{P}^a = 0 \qquad \text{on } \partial\mathcal{D} \,,
\end{align}
which yields:
\begin{align}
    \mathcal{F}_{\partial\mathcal{D}}^{\mathcal{Z}} = 0 \,,
\end{align}
meaning the enstrophy balance is governed entirely by the remaining bulk terms. If instead the wall is absorbing or partially transmitting, we have:
\begin{align}
    \mathcal{F}_{\partial\mathcal{D}}^{\mathcal{Z}} \neq 0 \,,
\end{align}
and the system behaves as an open system.

This distinction is useful for numerical experiments. A finite reflecting box can be implemented to test whether the local Weyl enstrophy is approximately conserved and whether the Fj{\o}rtoft constraint organizes the nonlinear transfer. An absorbing boundary, by contrast, provides a direct test of the leakage terms in the global balance law.

\subsection{Numerical Evaluation}
\label{app:numerical-flux-evaluation}

The enstrophy flux can be evaluated directly from $E_{ij}$ and $B_{ij}$ on each spatial slice. Letting $s^i$ be the outward unit normal to the extraction surface, the spatial super-Poynting vector is:
\begin{align}
    \delta\mathcal{P}^i = \epsilon^{ijk} \delta E_j{}^\ell \delta B_{k\ell} \,,
\end{align}
and the integrated flux is:
\begin{align}
    \mathcal{F}_{\partial\mathcal{D}}^{\mathcal{Z}} = \int_{\partial\mathcal{D}}\alpha s_i\delta\mathcal{P}^i\sqrt q\,d^2x \,,
\end{align}
for the equal-electric-and-magnetic-norm case, with the corresponding half-sum of residuals and boundary-advection terms retained. Positive flux decreases that norm. Otherwise use \cref{eq:app-integrated-B2} directly. A numerical balance check can be performed via:
\begin{align}
 \Delta\delta\mathcal Z_{\mathcal D}
 &+\int_{t_1}^{t_2}\left(\mathcal F_{\partial\mathcal D}^{\mathcal Z}
 -\delta\mathcal R_{\rm bulk}\right)dt\nonumber\\
 &=O(\epsilon^3)+O(\Delta_{\rm num}).
    \label{eq:app-numerical-balance-check}
\end{align}
where $\Delta_{\mathrm{num}}$ accounts for numerical truncation and extraction errors. In regimes where the bulk term is expected to be small, this simplifies to the direct conservation test:
\begin{align}
    \Delta\delta\mathcal{Z}_{\mathcal{D}} + \int_{t_1}^{t_2} \mathcal{F}_{\partial\mathcal{D}}^{\mathcal{Z}} \, dt \simeq 0 \,.
\end{align}

For an extraction sphere in an asymptotically flat simulation, one can directly compare the curvature flux with the standard gravitational-wave energy flux. Mode by mode, this yields:
\begin{align}
    \frac{\mathcal{F}_{\ell m \omega}^{\mathcal{Z}}}{\mathcal{F}_{\ell m \omega}^{W}} \sim \omega^2 \,,
\end{align}
which provides a direct numerical diagnostic of the enstrophy frequency weighting.

\section{ZAMO Kinematics and the Shear Contribution}
\label{app:zamo-shear}
\rev{The ZAMO frame has vanishing vorticity and expansion in stationary axisymmetric spacetimes, since the associated congruence is hypersurface orthogonal. In a rotating geometry, however, it generally has non-vanishing shear, sourced by spatial gradients of the frame-dragging angular velocity, so the shear term in the enstrophy balance does not vanish identically in Kerr or Kerr--AdS. At quadratic order this term is tied, together with the fluxes, to the linear decay of each mode (\cref{subsec:enstrophy-balance}); here we compute it explicitly, and \cref{app:shear-smallness} describes how it enters a numerical check of the balance.}

\subsection{Stationary Axisymmetric ADM Form}
\label{app:zamo-adm-form}

Consider a stationary, axisymmetric metric written in the standard ADM \(3+1\) form:
\begin{align}
    ds^2 = -\alpha^2 dt^2 + \gamma_{ij} \left( dx^i + \beta^i dt \right) \left( dx^j + \beta^j dt \right) \,,
\end{align}
where indices $i,j \in \{r, \theta, \phi\}$. For a circular, stationary, and axisymmetric spacetime, we can choose coordinate charts such that the spatial metric is diagonal and the shift vector is purely azimuthal:
\begin{equation}
\begin{gathered}
    \gamma_{ij} dx^i dx^j = \gamma_{rr} dr^2 + \gamma_{\theta\theta} d\theta^2 + \gamma_{\phi\phi} d\phi^2 \,, \\
    \beta^i \partial_i = \beta^\phi(r,\theta) \partial_\phi \,.
    \label{eq:app-circular-adm}
\end{gathered}
\end{equation}
It is conventional to define:
\begin{align}
    \beta^\phi = -\Omega_{\mathrm{fd}}(r,\theta) \,,
\end{align}
where $\Omega_{\mathrm{fd}} \equiv -g_{t\phi}/g_{\phi\phi}$ represents the frame-dragging angular velocity.

The ZAMO four-velocity is defined as the future-directed unit normal to the $t = \mathrm{constant}$ hypersurfaces. While denoted as $u^a$ in the main text, we write $n^a$ in this appendix to emphasize its geometric role as the hypersurface normal:
\begin{align}
    n_a = -\alpha \nabla_a t \,, \\
    n^a = \frac{1}{\alpha} \left( \partial_t - \beta^\phi \partial_\phi \right)^a = \frac{1}{\alpha} \left( \partial_t + \Omega_{\mathrm{fd}} \partial_\phi \right)^a \,. \nonumber
\end{align}
Because the one-form $n_a$ is proportional to a gradient, the ZAMO congruence is hypersurface orthogonal by Frobenius' theorem. Its vorticity tensor therefore vanishes identically:
\begin{align}
    \omega_{ij}^{\mathrm{ZAMO}} = 0 \,.
\end{align}

The acceleration of the congruence is generally non-zero and is given by:
\begin{align}
    a_i = n^b \nabla_b n_i = D_i \ln\alpha \,,
    \label{eq:app-zamo-acceleration}
\end{align}
where $D_i$ is the covariant derivative compatible with the spatial metric $\gamma_{ij}$. Thus, in Kerr or Kerr--AdS spacetimes, ZAMO observers are non-geodesic and experience a physical acceleration.

\subsection{Extrinsic Curvature, Expansion, and Shear}
\label{app:zamo-extrinsic}

We define the extrinsic curvature of the spatial slices as:
\begin{align}
    K_{ij} = \gamma_i{}^a \gamma_j{}^b \nabla_a n_b \,.
\end{align}
In terms of the ADM variables, this is expressed as:
\begin{align}
    K_{ij} = \frac{1}{2\alpha} \left( \partial_t \gamma_{ij} - D_i \beta_j - D_j \beta_i \right) \,.
\end{align}
This definition differs by an overall sign from conventions often adopted in numerical relativity, but the relative signs of the resulting balance equations are self-consistent.

For a stationary background spacetime, the metric satisfies $\partial_t \gamma_{ij} = 0$, reducing the extrinsic curvature to:
\begin{align}
    K_{ij} = -\frac{1}{2\alpha} \left( D_i \beta_j + D_j \beta_i \right) \,.
\end{align}
Using the circular ADM form in \cref{eq:app-circular-adm}, the only non-vanishing components are:
\begin{align}
    K_{A\phi} = K_{\phi A} = -\frac{1}{2\alpha} \gamma_{\phi\phi} \partial_A \beta^\phi = \frac{1}{2\alpha} \gamma_{\phi\phi} \partial_A \Omega_{\mathrm{fd}} \,, \\
    A \in \{r, \theta\} \,. \nonumber
\end{align}
All diagonal components vanish identically:
\begin{align}
    K_{rr} = K_{\theta\theta} = K_{\phi\phi} = 0 \,.
\end{align}
Because the spatial metric is diagonal, the trace of the extrinsic curvature vanishes:
\begin{align}
    K = \gamma^{ij} K_{ij} = 0 \,,
\end{align}
which directly implies that the ZAMO expansion vanishes:
\begin{align}
    \theta_{\mathrm{ZAMO}} = 0 \,.
\end{align}

The shear tensor $\sigma_{ij}$ is the symmetric, trace-free part of the spatial velocity gradient $\nabla_i n_j$. Since the trace $K$ vanishes, the shear tensor is identical to the extrinsic curvature:
\begin{align}
    \sigma_{ij}^{\mathrm{ZAMO}} = K_{ij} \,.
\end{align}
Thus, the ZAMO shear is non-zero if the frame-dragging angular velocity possesses non-vanishing spatial gradients:
\begin{align}
    \sigma_{A\phi}^{\mathrm{ZAMO}} = \frac{1}{2\alpha} \gamma_{\phi\phi} \partial_A \Omega_{\mathrm{fd}} \,.
\end{align}

In a local orthonormal spatial frame given by:
\begin{align}
    e_{\hat r} = \gamma_{rr}^{-1/2} \partial_r \,, \qquad e_{\hat\theta} = \gamma_{\theta\theta}^{-1/2} \partial_\theta \,, \qquad e_{\hat\phi} = \gamma_{\phi\phi}^{-1/2} \partial_\phi \,,
\end{align}
the non-zero orthonormal shear components are:
\begin{align}
    \sigma_{\hat r \hat\phi}^{\mathrm{ZAMO}} &= \frac{\sqrt{\gamma_{\phi\phi}}}{2\alpha \sqrt{\gamma_{rr}}} \partial_r \Omega_{\mathrm{fd}} \,, \label{eq:app-sigma-rphi-orth} \\
    \sigma_{\hat\theta \hat\phi}^{\mathrm{ZAMO}} &= \frac{\sqrt{\gamma_{\phi\phi}}}{2\alpha \sqrt{\gamma_{\theta\theta}}} \partial_\theta \Omega_{\mathrm{fd}} \,. \label{eq:app-sigma-thphi-orth}
\end{align}
This is the kinematic relation we use below: the ZAMO shear is produced by differential frame dragging, and it vanishes in the static limit ($\Omega_{\mathrm{fd}} = 0$).

\subsection{Kerr Example}
\label{app:kerr-zamo-shear}

For a Kerr black hole in Boyer--Lindquist coordinates, we define the standard functions:
\begin{align}
    \rho^2 = r^2 + a^2 \cos^2\theta \,, \qquad \Delta = r^2 - 2Mr + a^2 \,,
\end{align}
and
\begin{align}
    A = (r^2 + a^2)^2 - a^2 \Delta \sin^2\theta \,.
\end{align}
The ADM metric functions are given by:
\begin{align}
 \alpha&=\left(\frac{\Delta\rho^2}{A}\right)^{1/2},
 \qquad \gamma_{rr}=\frac{\rho^2}{\Delta},\nonumber\\
 \gamma_{\theta\theta}&=\rho^2,\qquad
 \gamma_{\phi\phi}=\frac{A\sin^2\theta}{\rho^2}.
\end{align}
and the frame-dragging profile is:
\begin{align}
    \Omega_{\mathrm{fd}} = \frac{2aMr}{A} \,.
\end{align}
Evaluating the shear components yields:
\begin{align}
    \sigma_{\hat r \hat\phi}^{\mathrm{ZAMO}} = \frac{A \sin\theta}{2\rho^3} \partial_r \Omega_{\mathrm{fd}} \,,
    \label{eq:app-kerr-sigma-rphi}
\end{align}
and
\begin{align}
    \sigma_{\hat\theta \hat\phi}^{\mathrm{ZAMO}} = \frac{A \sin\theta}{2\rho^3 \sqrt{\Delta}} \partial_\theta \Omega_{\mathrm{fd}} \,.
    \label{eq:app-kerr-sigma-thphi}
\end{align}

These expressions use the stationary Boyer--Lindquist ZAMO frame, which degenerates at the horizon. The shear components themselves have finite limits: $\sigma_{\hat r\hat\phi}$ stays finite, while $\partial_\theta\Omega_{\rm fd}=O(\Delta)$ makes $\sigma_{\hat\theta\hat\phi}\to0$. For regular calculations at the horizon, a horizon-penetrating coordinate system must be employed. However, for exterior bulk estimates away from the horizon, \cref{eq:app-kerr-sigma-rphi,eq:app-kerr-sigma-thphi} provide a precise measure of the local shear generated by differential frame dragging.

In the Schwarzschild limit ($a = 0$), we recover:
\begin{align}
    \Omega_{\mathrm{fd}} = 0 \,, \qquad \sigma_{ij}^{\mathrm{ZAMO}} = 0 \,.
\end{align}
Static Schwarzschild observers have zero vorticity, zero expansion, and zero shear, though they retain a non-vanishing spatial acceleration $a_i = D_i \ln\alpha$.

\subsection{Kerr--AdS}
\label{app:kerr-ads-zamo-shear}

The same kinematic analysis applies to Kerr--AdS spacetimes. In standard Boyer--Lindquist-type coordinates, the metric is:
\begin{multline}
    ds^2 = -\frac{\Delta_r}{\rho^2} \left( dt - \frac{a \sin^2\theta}{\Xi} d\phi \right)^2 + \frac{\rho^2}{\Delta_r} dr^2 \\
    + \frac{\rho^2}{\Delta_\theta} d\theta^2 + \frac{\Delta_\theta \sin^2\theta}{\rho^2} \left( a \, dt - \frac{r^2+a^2}{\Xi} d\phi \right)^2 \,,
\end{multline}
where
\begin{align}
    \rho^2 = r^2 + a^2 \cos^2\theta \,, \quad \Delta_\theta = 1 - \frac{a^2}{\ell^2} \cos^2\theta \,, \quad \Xi = 1 - \frac{a^2}{\ell^2} \,,
\end{align}
and
\begin{align}
    \Delta_r = (r^2 + a^2) \left( 1 + \frac{r^2}{\ell^2} \right) - 2Mr \,.
\end{align}
From this metric, we extract the frame-dragging profile $\Omega_{\mathrm{fd}} = -g_{t\phi}/g_{\phi\phi}$. The ZAMO shear components retain their algebraic form:
\begin{align}
    \sigma_{\hat r \hat\phi}^{\mathrm{ZAMO}} = \frac{\sqrt{\gamma_{\phi\phi}}}{2\alpha \sqrt{\gamma_{rr}}} \partial_r \Omega_{\mathrm{fd}} \,, \qquad \sigma_{\hat\theta \hat\phi}^{\mathrm{ZAMO}} = \frac{\sqrt{\gamma_{\phi\phi}}}{2\alpha \sqrt{\gamma_{\theta\theta}}} \partial_\theta \Omega_{\mathrm{fd}} \,.
    \label{eq:app-kerrads-zamo-shear}
\end{align}
While the AdS length scale $\ell$ modifies the lapse, the spatial metric, and the frame-dragging profile, it does not alter the underlying local kinematic structure of the shear terms.

\subsection{Shear Source in the $B^2$ Balance}
\label{app:shear-source-B2}

The shear contribution to the coordinate-time integrated magnetic Weyl balance is:
\begin{align}
    \mathcal{S}_\sigma^B = 6 \int_{\mathcal{D}_t}\alpha B^{ij} \sigma_{ki} B_j{}^k \sqrt{\gamma} \, d^3x \,,
\end{align}
where the coefficient of $6$ arises from contracting the symmetric, trace-free combination $3\sigma_{c\langle i} B_{j\rangle}{}^c$ with $2B^{ij}$.

In the ZAMO orthonormal frame, the only non-vanishing shear components are $\sigma_{\hat r \hat\phi}$ and $\sigma_{\hat\theta \hat\phi}$ and the tensor contraction simplifies to:
\begin{align}
    B^{ij} \sigma_{ki} B_j{}^k = 2\sigma_{\hat r \hat\phi} (B^2)_{\hat r \hat\phi} + 2\sigma_{\hat\theta \hat\phi} (B^2)_{\hat\theta \hat\phi} \,,
\end{align}
where we define the squared tensor components as:
\begin{align}
    (B^2)_{\hat i \hat j} \equiv B_{\hat i}{}^{\hat k} B_{\hat k \hat j} \,.
\end{align}
The integrated shear source term is thus:
\begin{align}
    \mathcal{S}_\sigma^B = 12 \int_{\mathcal{D}_t}\alpha\left[ \sigma_{\hat r \hat\phi} (B^2)_{\hat r \hat\phi} + \sigma_{\hat\theta \hat\phi} (B^2)_{\hat\theta \hat\phi} \right] \sqrt{\gamma} \, d^3x \,.
    \label{eq:app-shear-source-zamo}
\end{align}

For the first-order perturbative radiative enstrophy, we substitute the radiative perturbation field $\delta B_{ij}^{\mathrm{rad}}$:
\begin{align}
    \delta\mathcal{S}_\sigma^B = 12 \int_{\mathcal{D}_t}\alpha\left[ \sigma_{\hat r \hat\phi} (\delta B^2)_{\hat r \hat\phi} + \sigma_{\hat\theta \hat\phi} (\delta B^2)_{\hat\theta \hat\phi} \right] \sqrt{\bar\gamma} \, d^3x \,.
    \label{eq:app-pert-shear-source-zamo}
\end{align}
\rev{Thus the shear contribution vanishes identically in spacetimes without differential frame dragging (such as Schwarzschild, or global AdS in a static slicing), and it may also vanish under angular, time or phase averages for particular modes. For a generic, non-axisymmetric perturbation on a Kerr or Kerr--AdS background, however, \cref{eq:app-pert-shear-source-zamo} is a non-trivial bulk term; at quadratic order it is part of the complete linear balance (\cref{subsec:enstrophy-balance}).}

\subsection{Acceleration Source}
\label{app:zamo-acceleration-source}

The acceleration contribution to the magnetic Weyl balance is also non-zero for ZAMO observers. From \cref{eq:app-zamo-acceleration}, the orthonormal components of the acceleration are:
\begin{align}
    a_{\hat r} = \frac{1}{\sqrt{\gamma_{rr}}} \partial_r \ln\alpha \,, \qquad a_{\hat\theta} = \frac{1}{\sqrt{\gamma_{\theta\theta}}} \partial_\theta \ln\alpha \,, \qquad a_{\hat\phi} = 0 \,.
\end{align}
The acceleration term in the $B^2$ balance is given by:
\begin{align}
    \mathcal{S}_a^B = -4 \int_{\mathcal{D}_t}\alpha B^{ij} \epsilon_{kli} a^k E_j{}^l \sqrt{\gamma} \, d^3x \,,
    \label{eq:app-accel-source-general}
\end{align}
which dynamically couples the electric and magnetic Weyl sectors. While this term is absent for geodesic observers, it must be accounted for when using ZAMO observers in a strong gravitational potential.

Unlike the shear, the proper acceleration of stationary ZAMOs diverges at the event horizon: the exact Kerr kinematics dictate that the redshifted acceleration approaches the surface gravity, $\alpha \, a_{\hat r} \to \kappa_+$ as $\alpha \to 0$, implying $a_{\hat r} \sim \kappa_+/\alpha$. \rev{Whether the integrated source nevertheless remains finite is the subject of \cref{app:acceleration-convergence}.}

Thus, for ZAMO observers in rotating black-hole exteriors, the local bulk remainder in the magnetic enstrophy balance contains both shear and acceleration terms:
\begin{align}
    \mathcal{R}_{\mathrm{bulk}}^{\mathrm{ZAMO}} = \delta\mathcal{S}_\sigma^B + \delta\mathcal{S}_a^B + (\text{non-rad.\ and higher order}) \,.
    \label{eq:app-zamo-Rbulk}
\end{align}
While the vorticity and expansion terms vanish for the ZAMO congruence, these remaining terms generally do not.

\subsection{Convergence of the Acceleration Source}
\label{app:acceleration-convergence}

\rev{The convergence of the acceleration source in \cref{eq:app-accel-source-general} needs care because the ZAMO acceleration diverges at the horizon. The Kerr ZAMO kinematics give $\alpha \, a_{\hat r} \to \kappa_+$ as $r \to r_+$, with $\kappa_+$ the surface gravity, so that $a_{\hat r} \sim \kappa_+/\alpha$ (\cref{fig:acceleration-source}, left). Whether the integrated source remains finite then depends on how the perturbation behaves near the horizon in the ZAMO frame, and the answer is different for its outgoing and ingoing pieces. We consider the outgoing piece first.}

Let the near-horizon curvature satisfy $|\delta B|^2\sim\alpha^{2p}$ and assume $|\delta E|$ has the same leading order. For fixed subextremal spin, $\alpha^2\propto\Delta$ and $\sqrt{\gamma_{rr}}\propto\alpha^{-1}$. Including the lapse in the coordinate-time source gives
\begin{align}
 \alpha a_{\hat r}|\delta B|^2\sqrt{\gamma_{rr}}
 &\sim\kappa_+\alpha^{2p-1}
 \sim\Delta^{p-1/2}.
\end{align}
Since $\Delta\propto r-r_+$ at a nonextremal horizon, the radial integral converges if
\begin{equation}
 p>-\tfrac12.
 \label{eq:accel-convergence}
\end{equation}
This power count holds at fixed subextremal spin; taking the extremal limit changes the near-horizon expansion. Once convergence is established, the value of the source is obtained by integrating the chosen perturbation profile over the domain.

The ingoing part of a horizon-regular field is infinitely boosted in the stationary ZAMO frame. A generic nonzero regular ingoing curvature gives $|\delta B|^2\sim\alpha^{-4}$, or $p=-2$, so both this source and the ZAMO norm diverge. A treatment of the full perturbation with a horizon-regular observer, together with control of the outer QNM behavior, remains future work (Sec.~\ref{subsec:open}, ``Nonlinear validation of the enstrophy balance''). It will require evaluating the bulk balance in that frame.

\begin{figure*}[t]
\centering
\includegraphics[width=\textwidth]{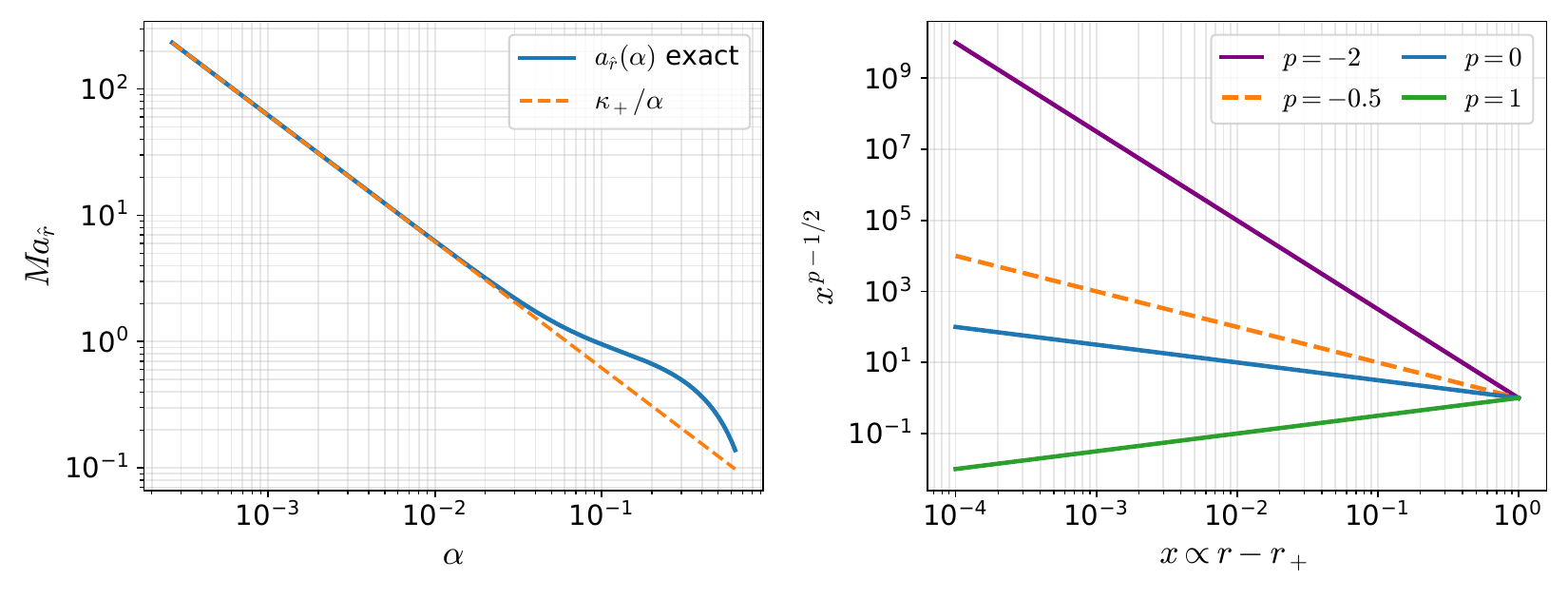}
\caption{\emph{Left}: Exact equatorial ZAMO acceleration for $\chi=0.99$ and its near-horizon form $\kappa_+/\alpha$. \emph{Right}: Local power-counting illustration for the coordinate-time acceleration integrand, including its lapse: for $x\propto r-r_+$ it scales as $x^{p-1/2}$ at fixed subextremal spin. The plotted powers have arbitrary common normalization and are not QNM source estimates. The radial integral converges for $p>-1/2$, including $p=0$ despite its pointwise divergent integrand; $p=-1/2$ gives a logarithmic divergence. The generic ingoing ZAMO field with $p=-2$ fails this condition.}
\label{fig:acceleration-source}
\end{figure*}

\subsection{Diagnostics}
\label{app:shear-smallness}

\rev{The shear and acceleration sources are not small term by term in near-extremal Kerr (\cref{app:shear-diagnostic}), and at quadratic order they need not be: together with the boundary fluxes they must reproduce the linear decay of each mode (\cref{subsec:enstrophy-balance}). Their use is rather to check that the decomposition of the balance law closes. For that purpose we measure them against the enstrophy change expected over a nonlinear transfer time,
\begin{align}
    \chi_\sigma(t) \equiv \frac{\left| \delta\mathcal{S}_\sigma^B(t) \right|}{\delta\mathcal{Z}_B(t) / \tau_{\mathrm{nl}}} \,, \qquad \chi_a(t) \equiv \frac{\left| \delta\mathcal{S}_a^B(t) \right|}{\delta\mathcal{Z}_B(t) / \tau_{\mathrm{nl}}} \,.
\end{align}
For ZAMOs near extremality these bulk terms contribute appreciably to the balance, so we include them in the numerical check below.}

\subsection{Numerical Implementation}
\label{app:zamo-numerical-implementation}

\rev{In a standard $3+1$ numerical simulation the Eulerian observer is the unit normal $n^a$ (in Boyer--Lindquist slicing these are the ZAMOs), and all required variables are readily available on each spatial slice: the lapse $\alpha$, the shift $\beta^i$, the spatial metric $\gamma_{ij}$, the extrinsic curvature $K_{ij}$, and the electric and magnetic Weyl tensors $E_{ij}$ and $B_{ij}$. The shear and expansion are computed directly as $\sigma_{ij} = K_{ij} - \frac{1}{3} K \gamma_{ij}$ and $\theta = K$ (where $K = 0$ for stationary circular backgrounds, but may be non-zero in dynamical spacetimes).

The shear and acceleration contributions to the magnetic Weyl balance are evaluated on each slice as
\begin{align}
    \delta\mathcal{S}_\sigma^B &= 6 \int_{\mathcal{D}_t}\alpha\delta B^{ij} \sigma_{ki} \delta B_j{}^k \sqrt{\gamma} \, d^3x \,, \nonumber\\
    \delta\mathcal{S}_a^B &= -4 \int_{\mathcal{D}_t}\alpha\delta B^{ij} \epsilon_{kli} a^k \delta E_j{}^l \sqrt{\gamma} \, d^3x \,.
\end{align}
The full-field balance also contains $-\int\alpha\theta B^2dV$ and boundary advection. For radiative perturbations we add the linearization and projection sources. Together with the fluxes in \cref{app:boundary-fluxes}, these terms enter the numerical check \cref{eq:app-numerical-balance-check}.}

\section{Fj{\o}rtoft algebra}
\label{app:fjortoft-details}

For the radiative sector:
\begin{align}
    W = \sum_k W_k, \qquad 
    \mathcal{Z} = \sum_k \mathcal{Z}_k \simeq \sum_k \Omega_k W_k, 
    \qquad \Omega_k = \omega_k^2.
\end{align}
\rev{Here $W_k \ge 0$ is the modal energy and $\mathcal{Z}_k$ the corresponding enstrophy. The $\simeq$ allows for mode-normalization constants, as long as they do not break the monotonic dependence on $\Omega_k$; if needed, they can be absorbed into a redefined $\Omega_k$ (\cref{app:degenerate-fjortoft}).}

Suppose that on the nonlinear transfer timescale $\tau_{\mathrm{nl}}$
\begin{align}
    \delta W = 0, \qquad \delta \mathcal{Z} = 0.
\end{align}
Then any redistribution $\delta W_k$ must satisfy
\begin{align}
    \sum_k \delta W_k = 0, \qquad 
    \sum_k \Omega_k \delta W_k = 0.
    \label{eq:app-fj-constraints}
\end{align}
\rev{These two relations give the argument in \cref{sec:fjortoft}. We now extend it to continuous spectra, coherent bands and open systems.}
\subsection{Continuous spectra}
\label{app:continuous-fjortoft}

For a continuous distribution, let $W(\Omega)\ge0$ be the spectral 
density. Then
\begin{align}
    W = \int_0^\infty W(\Omega)\,d\Omega, \qquad
    \mathcal{Z} = \int_0^\infty \Omega W(\Omega)\,d\Omega.
\end{align}
A redistribution $\delta W(\Omega)$ preserving both invariants satisfies
\begin{align}
    \int_0^\infty \delta W(\Omega)\,d\Omega = 0, \qquad
    \int_0^\infty \Omega\,\delta W(\Omega)\,d\Omega = 0.
\end{align}
The continuous form is suited to scattering states and the discrete one to confined geometries; the structure is the same.

\subsection{Degeneracies and normalization}
\label{app:degenerate-fjortoft}

Modes can share a frequency without sharing the same ratio of curvature norm to energy. In a coherent block write $W=a^\dagger Ga$ and $\mathcal Z=a^\dagger Ha$, with $G>0$. Simultaneous diagonalization gives the eigenvalues of $G^{-1/2}HG^{-1/2}$ as effective weights. A common band weight $\Omega_m$ is justified only if, for all band amplitudes,
\[
 (\Omega_m-\delta_m)G_m\leq H_m\leq(\Omega_m+\delta_m)G_m
\]
with a sufficiently small $\delta_m$. Disjoint ordered intervals give ordering between bands.

More generally, if
\begin{align}
    \mathcal{Z} = \sum_k c_k \Omega_k W_k, \qquad c_k>0,
\end{align}
define $\widetilde{\Omega}_k = c_k \Omega_k$. The argument goes through 
identically as long as the ordering by frequency is preserved. The 
coefficients $c_k$ can change quantitative ratios but not the direction 
of the constraint.

\subsection{The mixed Weyl functional}
\label{app:helicity-fjortoft}

The mixed functional
\begin{align}
    \mathcal{H} = \int_\Sigma E_{ab} B^{ab} \sqrt{\gamma}\,d^3x
\end{align}
has the same frequency scaling as $E^2$ and $B^2$, but it is sign-indefinite. Schematically, $\mathcal{H} = \sum_k \chi_k \Omega_k W_k$, where $\chi_k$ is a polarization or phase factor that can change sign. It cannot serve as the positive-definite enstrophy of the Fj{\o}rtoft pair. If conserved, it would supply a different, sign-indefinite constraint. The vanishing contraction for a pure type-N wave is displayed in \cref{app:np-stf}.

\subsection{When the system is open}
\label{app:damping-driving-fjortoft}

In GR the radiative sector is open. Energy and enstrophy leak through 
horizons, escape to infinity, or exchange with matter. The balance 
laws are
\begin{align}
    \dot{W} = \mathcal{I}_W - \mathcal{D}_W - \mathcal{F}_W, \qquad
    \dot{\mathcal{Z}} = \mathcal{I}_{\mathcal{Z}} - \mathcal{D}_{\mathcal{Z}} 
    - \mathcal{F}_{\mathcal{Z}} + \mathcal{R}_{\mathrm{bulk}}.
\end{align}
\rev{These open-system terms are negligible only if they are slow compared to the nonlinear transfer,}
\begin{align}
    \left|\frac{\dot{W}}{W}\right| \ll \tau_{\mathrm{nl}}^{-1}, \qquad
    \left|\frac{\dot{\mathcal{Z}}}{\mathcal{Z}}\right| \ll \tau_{\mathrm{nl}}^{-1}.
    \label{eq:app-open-smallness}
\end{align}
\rev{where the rates include only the open-system terms, as in \cref{subsec:general-flux-balance}. These inequalities express (C2); conservation by the interactions is (C3).}

\rev{This is why generic ringdown and near-extremal Kerr differ. In a generic ringdown $\tau_{\mathrm{damp}} \lesssim \tau_{\mathrm{nl}}$, and the modes damp before nonlinear transfer can act. Near extremality the ZDMs give $\tau_{\mathrm{damp}} \gg \tau_{\mathrm{nl}}$, and the transfer has time to develop; \cref{app:timescale-estimate} gives the timescale comparison.}

For a two-receiver transfer, errors must be compared with the transferred moments, not only with the total norm. If $\delta W=0$ and $\delta\mathcal Z=e$, then
\begin{align}
 \delta W_-&=\frac{(\Omega_+-\Omega_0)\Delta-e}{\Omega_+-\Omega_-},\\
 \delta W_+&=\frac{(\Omega_0-\Omega_-)\Delta+e}{\Omega_+-\Omega_-}.
\end{align}
For $\omega_\pm=\omega_0\pm\delta$ and $\Omega=\omega^2$, the inverse energy bias survives only if $e<\delta^2\Delta$. Narrow transfers can therefore require much greater accuracy than $|e|\ll\mathcal Z$.

\subsection{Superradiance and energy positivity}
\label{app:superradiance-fjortoft}

The Fj{\o}rtoft argument assumes a positive modal energy. In the superradiant range $0<\omega<m\Omega_H$, the stationarity-energy horizon flux is negative: the black hole supplies energy to the exterior. This is a statement about the flux, not by itself about the sign of the bulk mode norm.

The co-rotating charge $W^{(\chi)}=W^{(\xi)}-\Omega_HJ$ has nonnegative horizon flux, but need not be positive throughout the exterior. Changing the Killing vector therefore does not automatically restore the theorem's hypotheses. One must establish positivity on the chosen perturbation sector and include the horizon work in its balance. Transfer across the threshold remains a dynamical question.

\section{Fluid/Gravity Map for the Bulk Enstrophy}
\label{app:fluid-gravity-kernel}

The relation $\delta\tilde{\mathcal Z}\simeq\tilde{\mathcal C}(T)\,\Omega_{\rm CFT}$ in the AdS$_4$/CFT$_3$ fluid/gravity regime, \cref{eq:weighted-kernel-integrated}, can be derived as follows.

\subsection{Weighted Magnetic Weyl Tensor}
\label{app:fg-bulk}
On the ingoing EF slices $v = \text{const}$ of the metric \cref{eq:fluid-gravity-metric}, with the fluid-weighted tensor of \cref{eq:fluid-weighted-B},
\begin{align}
    \tilde B_{ab} = {}^{\star}C_{acbd}\, u^c u^d
    = \frac{1}{2}\, \epsilon_{ac}{}^{ef}\, C_{efbd}\, u^c u^d \,,
\end{align}
where $u^a=u^\mu\partial_\mu$ is the bulk extension of the boundary-normalized fluid velocity. We use the background EF measure $r^2dr\,d^2x$ at the order retained. Constant-$v$ surfaces are null for the static background, not for a general vortical flow. The weighted tensor remains finite in the calculation below without normalizing $u^a$ at the horizon.

We linearize about the static black brane background ($v_i = 0$, $T = \text{const}$). At first order in the gradient expansion, the weighted magnetic Weyl perturbation is a linear combination of the two available pseudo-structures of the boundary flow,
\begin{align}
\delta \tilde B_{ab}(r,\mathbf x) =& \tilde{\mathscr B}_{ab}(r;T)\,\omega_{\rm fl}(\mathbf x) \nonumber \\
&+ \tilde{\mathscr S}_{ab}{}^{ij}(r;T)\,\epsilon^{k}{}_{(i}\sigma_{j)k}(\mathbf x) + O(\partial^2 v, v^2) \,,
\label{eq:app-deltaB-vort}
\end{align}
Here $\tilde{\mathscr B}_{ab}(r;T)$ and $\tilde{\mathscr S}_{ab}{}^{ij}(r;T)$ are radial profile tensors determined by the linearized Einstein equations about the brane; their explicit form is not needed, since the quadratic invariant contracts them into the two scalars $\tilde K_\omega$ and $\tilde K_\sigma$ of \cref{eq:BB-fluid-kernel}. These are fixed by evaluating $\delta\tilde B_{ab}$ on the two exact first-order flows of \cref{sec:holographic-fluid} --- rigid rotation, for which only the first term survives, and pure shear, for which only the second does.
(Note $\tilde B_{ab} = 0$ for the uniformly boosted brane, so the expansion starts at first order; expansion and on-shell acceleration do not contribute at this order.)

\subsection{Vorticity and Shear contributions}
\label{app:fg-kernel}

The bulk enstrophy perturbation is
\begin{align}
\delta\tilde{\mathcal Z} = \int_{r_h}^{\infty} dr \int \delta \tilde B_{ab}\,\delta \tilde{B}^{ab}\,\sqrt{\sigma}\,d^2x \,,
\end{align}
and substituting the two-sector decomposition \cref{eq:BB-fluid-kernel} together with the incompressibility identity $\int \sigma_{ij}\sigma^{ij}\,d^2x = \tfrac12\int\omega_{\rm fl}^2\,d^2x$ (decaying or periodic data) factorizes it into \cref{eq:weighted-kernel-integrated}, with
\begin{align}
\tilde{\mathcal C}_{\omega,\sigma}(T) = \int_{r_h}^{\infty} \tilde K_{\omega,\sigma}(r;T)\, r^2\, dr \,.
\label{eq:app-CB-def}
\end{align}

The radial profiles follow from two exact first-order solutions, and when carrying the sourced first-order metric response as an unknown radial profile, it cancels identically from $\delta\tilde B_{ab}$. At first order in gradients, the weighted magnetic sector is algebraic in $\partial_\mu u_\nu$ and blind to $ds^2_{\rm grad}$, which sources only the electric (tidal) sector --- a property of these first-gradient incompressible solutions. Both vorticity and shear contribute to the weighted magnetic field.

\emph{Vorticity sector.} For a rigidly rotating boundary flow, $v_i = \Omega\,(-y,x)$, the shear and expansion vanish identically and the acceleration is $O(\Omega^2)$, so $ds^2_{\rm grad}=0$. Evaluating the linearized Weyl tensor, dualizing, and contracting with $u^a$ gives the components quoted in \cref{eq:weighted-B-components}.

\emph{Shear sector.} For a pure-shear flow, $v_i = s\,(y,x)$, with $\sigma_{xy} = s$ and $\omega_{\rm fl} = 0$, the linearized Einstein equations at $O(s)$ reduce to a single radial ODE for the sourced response $\delta g_{xy} = A(r)\,\sigma_{xy}$,
\begin{align}
r^2\left(b^3 r^3 - 1\right) A'' + 3 r A' - \left(2 b^3 r^3 + 4\right) A = -4\, b^3 r^4 \,,
\end{align}
whose regular normalizable solution is the AdS$_4$ shear profile. The weighted magnetic components, however, are independent of $A$: $\delta\tilde B_{xx} = -\delta\tilde B_{yy} = -3s/(2b^3)$, with the pseudo-tensor structure $\epsilon^{k}{}_{(i}\sigma_{j)k}$; the orthogonal shear polarization ($\sigma_{xx} = -\sigma_{yy}$) gives the correspondingly rotated components and the same radial profile, confirming isotropy. Mixed flows superpose the two sectors exactly, with no cross term in $\delta\tilde B_{ab}\delta\tilde B^{ab}$.
Thus,
\begin{align}
&\tilde K_\omega = \frac{27\,f^2(br)}{8\, b^6 r^4} \,, \quad
\tilde K_\sigma = \frac{9}{4\, b^6 r^4} \,, \nonumber \\
&\tilde{\mathcal C} = \tilde{\mathcal C}_\omega + \tfrac12 \tilde{\mathcal C}_\sigma = \frac{369}{112\, b^5} = \frac{2624\,\pi^5}{189}\, T^5 \,,
\end{align}
in units $\ell_{\rm AdS} = 1$, with $b = 3/(4\pi T)$, $r_h = 1/b$, and $\tilde{\mathcal C}_\omega = 243/(112 b^5)$, $\tilde{\mathcal C}_\sigma = 9/(4 b^5)$. Both integrands are regular at the horizon ($\tilde K_\omega$ vanishing as $f^2$; $\tilde K_\sigma$ finite) and fall as $r^{-2}$ at conformal infinity. The unit-observer shear computation diverges as $f^{-2}$ at the horizon; the weighting is necessary then. Using the same EF measure, the unit-fluid-observer vorticity coefficient is $9\pi T/10$. This is not a proper-spacelike-volume coefficient.
The $T^5$ scaling follows from the contribution of $r^2(1-f) = 1/(b^3 r)$ which is the source of Weyl curvature and every component of $\delta\tilde B_{ab}$ carries a factor $1/b^3$; contracting, $\delta\tilde B_{ab}\,\delta\tilde B^{ab} \sim \omega_{\rm fl}^2/(b^6 r^4)$; integrating against $r^2\,dr$ gives $\tilde{\mathcal C} \sim 1/(b^6 r_h) = 1/b^5 \propto T^5$.

By design, $\mathcal R_{\rm holo}(v) \to 1$ in the strict hydrodynamic regime, cf.\ \cref{eq:Rholo-diagnostic}. Deviations require control of the higher-gradient, compressibility and temperature corrections described in \cref{subsec:holo-turbulence-relation}.

\subsection{Spacelike slices and the observer choice}
\label{app:fg-spacelike}

With a Cauchy foliation the fluid-gravity correspondence has also been demonstrated~\cite{Bantilan:2012vu,Bantilan:2014sra}. To link black-hole behavior to the asymptotic fluid, one can conveniently choose a horizon-penetrating deformation of the EF slicing,
\begin{align}
    \tau = v - h(r) \,, \qquad 0 < h'(r) < \frac{2}{r^2 f(br)} \,,
    \label{eq:spacelike-deformation}
\end{align}
on the equilibrium brane, for which $\gamma_{rr}=2h'-r^2f h'^2>0$. A smooth finite positive $h'$ at the horizon gives a regular crossing and a regular unit normal $n^a$. The proper volume measure is $r^2\sqrt{\gamma_{rr}}\,dr\,d^2x$. Both this measure and the change from the fluid observer to $n^a$ affect the coefficient; neither is included merely by removing the weight $(r^2f)^2$.

For comparison, keeping the EF measure but normalizing the fluid vector gives
\begin{align}
 K_{U,\omega}&=\frac{\tilde K_\omega}{(r^2f)^2}
 =\frac{27}{8b^6r^8},\nonumber\\
 \mathcal C_{U,\omega}^{\rm EF}
 &=\int_{1/b}^{\infty}K_{U,\omega}r^2dr
 =\frac{27}{40b}=\frac{9\pi}{10}T.
\end{align}
The corresponding shear kernel diverges as $f^{-2}$. These are properties of the unit fluid frame with the EF measure, not a calculation with a regular spacelike normal. A spacelike construction must also remove any zeroth-gradient magnetic field induced by the relative boost of a uniformly moving brane. Its coefficient and temperature scaling require a separate evaluation. The weighted EF functional avoids that additional prescription and is the one used in \cref{eq:Rholo-diagnostic}.

For numerical use of this dictionary, we note the following:
\begin{enumerate}
\item The closed forms above are the leading small-$k$ profiles; for a boundary velocity mode $\mathbf{v}(\mathbf{x}) = \mathbf{v}_0 e^{i\mathbf{k}\cdot\mathbf{x}}$ with $\omega_{\rm fl}(\mathbf{k}) = i \epsilon^{ij} k_i v_{0,j} e^{i\mathbf{k}\cdot\mathbf{x}}$, finite-wavenumber corrections enter at $O(\partial^2 v)$, consistently with \cref{eq:weighted-kernel-integrated}, and require the metric and hydrodynamic data to the corresponding higher derivative order.
\item Extraction at boundary time $t$ pairs $\Omega_{\rm CFT}(t)$ with the bulk integral on the EF slice $v = t$, which reaches conformal infinity at that boundary time.
\end{enumerate}

\section{Quantitative checks: shear, horizon suppression, and the spectral ratio}
\label{app:quantitative-diagnostics}

\rev{This appendix collects the quantitative input behind the near-extremal discussion of \cref{sec:flux-regimes}: the size of the ZAMO shear, the suppression of horizon absorption near extremality, the status of the spectral ratio $\delta\mathcal Z_k/W_k = \omega_k^2$ in the strong-field region, the kinematics of ZDM triads, and the timescale separation needed for nonlinear transfer to act before the modes damp.}

\subsection{Smallness of the ZAMO shear}
\label{app:shear-diagnostic}

{For transverse radiative fields neither shear component couples to the magnetic enstrophy (Proposition~\ref{prop:shear-decoupling}); the coupling enters only through the $O(v)$ misalignment of the ZAMO frame with the propagation direction. It is still useful to quantify the size of the background shear.} In Boyer--Lindquist coordinates the frame-dragging angular velocity is
\begin{align}
    \Omega_{\rm fd}
    =
    -\frac{g_{t\phi}}{g_{\phi\phi}}
    =
    \frac{2Mar}{A},
    \qquad
    A=(r^2+a^2)^2-a^2\Delta\sin^2\theta,
\end{align}
with $\Delta=r^2-2Mr+a^2$. The ZAMO four-velocity is
$u^a=\alpha^{-1}((\partial_t)^a+\Omega_{\rm fd}(\partial_\phi)^a)$,
and the horizon angular velocity is
\begin{align}
    \Omega_H \equiv \Omega_{\rm fd}(r_+) = \frac{a}{r_+^2+a^2} = \frac{a}{2Mr_+}\,,
\end{align}
where $r_+=M+\sqrt{M^2-a^2}$. The non-vanishing orthonormal shear components are given in \cref{eq:app-sigma-rphi-orth,eq:app-sigma-thphi-orth}.

{Define the dimensionless background shear scale
\begin{align}
    \hat\chi_\sigma(\chi,r,\theta)
    \equiv
    \frac{|\sigma_{\hat\theta\hat\phi}^{\rm ZAMO}|}{\omega_R},
    \qquad
    \omega_R = \mathrm{Re}(\omega_{\rm QNM}).
\end{align}
Notice this is not the diagnostic $\chi_\sigma$ of \cref{app:shear-smallness}: the latter includes lapse weighting, tensor contractions, frame mixing and spatial integration, none of which is specified by this pointwise shear scale.
Representative values of $\hat\chi_\sigma$ at} $r=r_++\max(0.1M,3M(1-\chi))$, $\theta=\pi/4$, are given in Table~\ref{tab:shear-diagnostic}; they were obtained using the $l=m=2$, $n=0$ QNM frequencies of Leaver~\cite{Leaver1985} and Yang et al.~\cite{YangZimmermanLehner2015}.

\begin{table*}[t]
\centering
\caption{{Background ZAMO shear scale $\hat\chi_\sigma=|\sigma_{\hat\theta\hat\phi}|/\omega_R$}. Evaluation point: 
$r=r_++\max(0.1M, 3M(1-\chi))$, $\theta=\pi/4$. Last column: 
background timescale ratio $\kappa_+/\Omega_H$ (App.~\ref{app:superradiant-suppression}).}
\label{tab:shear-diagnostic}
\begin{tabular}{ccccccc}
\hline
$\chi$ & $r_+/M$ & $r_{\rm eval}/M$ &
$|\sigma_{\hat r\hat\phi}|M$ & $|\sigma_{\hat\theta\hat\phi}|M$ &
{$\hat\chi_\sigma$} & $\kappa_+/\Omega_H$ \\
\hline
0.50 & 1.866 & 3.366 & 0.02660 & 0.000127 & $2.7\times10^{-4}$ & 1.732 \\
0.70 & 1.714 & 2.614 & 0.07174 & 0.000917 & $1.7\times10^{-3}$ & 1.020 \\
0.90 & 1.436 & 1.736 & 0.21268 & 0.005799 & $8.6\times10^{-3}$ & 0.484 \\
0.95 & 1.312 & 1.462 & 0.28413 & 0.007888 & $1.1\times10^{-2}$ & 0.329 \\
0.99 & 1.141 & 1.241 & 0.34009 & 0.009085 & $1.0\times10^{-2}$ & 0.142 \\
0.999 & 1.045 & 1.145 & 0.36094 & 0.008580 & $9.0\times10^{-3}$ & 0.045 \\
0.9999& 1.014 & 1.114 & 0.36712 & 0.007740 & $7.9\times10^{-3}$ & 0.014 \\
\hline
\end{tabular}
\end{table*}

\rev{Notice that $|\sigma_{\hat r\hat\phi}|$ is much larger than $|\sigma_{\hat\theta\hat\phi}|$: the ratio is about $200$ at $\chi=0.5$ and settles at $35$--$50$ for $\chi\gtrsim0.9$, where the radial shear approaches $0.37\,M^{-1}$.} 
{Neither component couples to purely transverse radiation (Proposition~\ref{prop:shear-decoupling}); through the $O(v)$ misalignment it is the larger radial component $\sigma_{\hat r\hat\phi}$ that enters, weighted by $v$.

As well, $\hat\chi_\sigma$ ranges from $2.7\times10^{-4}$ to 
$1.1\times10^{-2}$ across $\chi\in[0.5,0.9999]$. It plateaus near 
$0.008$--$0.011$ for $\chi\gtrsim0.95$, then drops slightly toward 
extremality: near-extremality does \emph{not} enhance the background shear.}

{Ignoring lapse and profile factors, a local estimate of the residual coupling is
\begin{align}
    \frac{|\delta\dot{\mathcal Z}_\sigma|}{\delta\mathcal Z/\tau_{\rm nl}}
    \sim
    \frac{v\,|\sigma_{\hat r\hat\phi}|}{\epsilon^2\omega_R},
\end{align}
with $v \lesssim 0.2\chi$ (\cref{eq:carter-zamo-boost}). With the values of Table~\ref{tab:shear-diagnostic}, the shear contribution is appreciable at $\epsilon\sim0.1$. It enters the sum that reproduces the linear mode decay, as discussed in \cref{subsec:enstrophy-balance}.}

\subsection{Suppression of horizon absorption}
\label{app:superradiant-suppression}

For a fundamental ZDM the complete quadratic envelope decays at rate $2|\omega_I|\simeq\kappa_+$. Under the reference nonlinear-rate estimate,
\[
 \alpha_{\rm loss}\simeq\frac{\kappa_+}{\epsilon^2\omega_R}.
\]
This controls total linear change. A separate horizon fraction requires the actual radial solution, norm and horizon conversion of \cref{app:horizon-flux}. It cannot be inferred from the threshold factor alone.

The exact ratio is
\begin{align}
    \frac{\kappa_+}{\Omega_H}
    =
    \frac{\sqrt{M^2-a^2}}{a}
    \underset{\chi\to 1}{=}
    \sqrt{2(1-\chi)} + O((1-\chi)^{3/2}).
\end{align}
Numerically (Table~\ref{tab:shear-diagnostic}, last column), $\kappa_+/\Omega_H = 0.1425$ at $\chi=0.99$, $0.0448$ at $\chi=0.999$, and $0.0141$ at $\chi=0.9999$.

\rev{The ratio $\kappa_+/\Omega_H$ depends only on the spin. The fractional loss itself, $\kappa_+/(\epsilon^2\omega_R)$, is the inverse of the timescale ratio $\tau_{\rm damp}/\tau_{\rm nl}$ of \cref{app:timescale-estimate}, so requiring it to be small is the same condition as $\epsilon^2 \gg m^{-1}\sqrt{2(1-\chi)}$ (\cref{eq:epsilon-chi-condition}). }

\subsection{The spectral ratio $\delta\mathcal Z_k/W_k = \omega_k^2$: 
where it holds and where it doesn't}
\label{app:spectral-ratio-proof}

The radiation-zone ratio follows from the additional time derivative in curvature, using the energy normalization of Proposition~\ref{thm:spectral-weighting}. Its extension to strong fields must be tested with independently defined norms.

\subsubsection{Radiation-zone proof}
In an observer-adapted transverse tetrad, $\delta B^2=\tfrac12|\Psi_4|^2$. With the same phase convention and to leading radiative order,
\[
 W_k=\frac{1}{2\omega_k^2}\int_{\mathcal D}|\Psi_4^{(1)}|^2dV,
 \qquad
 \delta\mathcal Z_k=\frac12\int_{\mathcal D}|\Psi_4^{(1)}|^2dV.
\]
Here the energy expression follows from the metric Isaacson norm; using it as a definition would make the ratio a tautology. For complex-frequency envelopes replace $\omega_k^2$ by $|\omega_k^{\rm QNM}|^2$. Higher Schwarzschild overtones are not weakly damped: for $\ell=2$, $n=0,1,2,3$, the factors $|\omega^{\rm QNM}|^2/\omega_R^2$ are approximately $1.06,1.62,3.52,8.86$. Thus they cannot all be represented by an exact real-frequency $\eta=1$ test.

\subsubsection{Near-horizon: what changes}

Near the horizon the Isaacson energy is not the correct energy 
expression. Moreover, the derivative along the ZAMO congruence is
$u^a\nabla_a\sim\alpha^{-1}(\partial_t+\Omega_{\rm fd}\partial_\phi)
+\text{connection}$, so a mode $e^{-i\omega t+im\phi}$ has local 
convective frequency $(\omega-m\Omega_{\rm fd})/\alpha$, not 
$\omega$. Dropping shear and acceleration gives
\begin{align}
    \delta B_{ab}
    \simeq
    \frac{i}{\omega_k}(\mathrm{curl}\,\delta E)_{ab},
\end{align}
which is a radiation-zone relation, not a strong-field identity.

\rev{The correction that matters here is the one to the enstrophy.} The linearized Weyl tensor receives background curvature terms:
\begin{align}
    \delta B_{ab}
    =
    \omega_k^2\,h_{ab}^{(\rm TT)}
    +
    C_{\rm bg}\cdot h_{ab}
    +
    \ldots,
    \qquad
    C_{\rm bg}\sim M/r^3.
\end{align}
For low-$m$ ZDMs, $\omega_k\simeq m\Omega_H\simeq m/(2M)$ and $C_{\rm bg}\sim1/M^2$, so the two terms can be comparable near the horizon.

We can always write $\delta\mathcal Z_k/W_k=c_k\omega_k^2$, but this is only notation; it says nothing about the size of $c_k$. The difficulty is that ZDMs live in the near-horizon throat, where background curvature is as large as the derivative terms for the low-$m$ modes considered here. Two modes concentrated in the same region can still have different $c_k$, so we cannot assume that the ordering of the weights survives there. The ratio of adjacent frequency-squared weights is $(m+1)^2/m^2$, which tends to one as $m$ grows; there is no uniform gap such as $16/9$.

The required calculation compares the curvature and energy norms of the near-horizon modes. Their outgoing tails obey the radiation-zone relation, but the nonlinear couplings probe the field stored near the black hole.

\subsubsection{Schwarzschild sanity check}

For the Schwarzschild $l=2$, $n=0$ QNM ($\omega M=0.37367-0.08896i$), 
an illustrative dimensionless curvature-scale comparison is
\begin{align}
    \frac{C_{\rm bg}(r)}{\omega_k^2}
    =
    \frac{l(l+1)M}{r^3\omega_k^2}.
\end{align}

\begin{table}[!htb]
\centering
\begin{tabular}{ccc}
\hline
$r/M$ & $C_{\rm bg}/\omega_k^2$ & regime \\
\hline
$3$ (photon sphere) & $1.59$ & strong \\
$4$ & $0.67$ & strong \\
$5$ & $0.34$ & moderate \\
$6$ & $0.20$ & moderate \\
$8$ & $0.08$ & weak \\
$10$ & $0.04$ & negligible \\
\hline
\end{tabular}
\caption{Illustrative scale $\ell(\ell+1)M/(r^3\omega_R^2)$ for the 
Schwarzschild $l=2$, $n=0$ QNM.}
\label{tab:Cbg-correction}
\end{table}

Table~\ref{tab:Cbg-correction} compares the curvature and frequency scales. To turn these numbers into corrections to $\eta_k$, one would reconstruct the metric and curvature and evaluate both norms over the same domain.

\FloatBarrier
\subsection{Three-wave resonance is kinematically open}
\label{app:zdm-triad-resonance}

The three-wave resonance requires $\omega_3\approx\omega_1+\omega_2$ 
within the mode linewidths. For ZDMs,
$\omega_{\ell mn}\simeq m\Omega_H-\kappa_+[c_{\ell m}+i(n+\tfrac12)]$,
with $c_{\ell m}$ of order unity~\cite{Yang:2012pj}. For an azimuthally 
allowed triad ($m_1+m_2=m_3$), the $\Omega_H$ terms cancel, and both 
the detuning $\Delta=|{\rm Re}(\omega_1+\omega_2-\omega_3)|$ and the 
combined linewidth $\Gamma=\sum_i|{\rm Im}\,\omega_i|$ scale as 
$\kappa_+$. The ratio $R=\Delta/\Gamma$ has a finite leading near-extremal limit.
Table~\ref{tab:zdm-triads} lists $R$ for the leading $\ell=m$ triads,
computed from Leaver continued fractions\footnote{For 
$\chi \gtrsim 0.999$ the $n=1$ modes must be tracked onto the ZDM 
branch directly, ${\rm Im}\,\omega \simeq -(n+\tfrac12)\kappa_+$: 
continued-fraction continuation in spin can hop onto an adjacent 
overtone-multiplet branch with ${\rm Im}\,\omega \simeq 
-\tfrac52\kappa_+$, which would spuriously lower the $n=1$ ratios.}. 
The values shown are representative near-extremal estimates; the linewidth comparison is kinematic and does not determine a secular growth rate.

\begin{table}[ht]
\centering
\caption{Detuning-to-linewidth ratio $R=\Delta/\Gamma$ for ZDM triads. 
All leading triads lie inside the linewidth ($R<1$).}
\label{tab:zdm-triads}
\begin{tabular}{lc}
\hline
Triad & $R=\Delta/\Gamma$ \\
\hline
$(2,2,0)+(2,2,0)\to(4,4,0)$ & 0.357 \\
$(2,2,0)+(2,2,0)\to(4,4,1)$ & 0.215 \\
$(2,2,0)+(2,2,1)\to(4,4,0)$ & 0.215 \\
$(2,2,0)+(3,3,0)\to(5,5,0)$ & 0.32 \\
$(3,3,0)+(3,3,0)\to(6,6,0)$ & 0.27 \\
\hline
\end{tabular}
\end{table}

Since $R<1$ throughout, the detuning of every leading triad lies within its combined linewidth, and three-wave resonances are kinematically allowed.

\subsection{Conditions for the cascade to develop}
\label{app:timescale-estimate}

\rev{Nonlinear transfer can act before the modes damp only when $\tau_{\rm damp}\gg\tau_{\rm nl}$; what the transfer then does to $\mathcal{Z}$ is set by the dominant interaction (\cref{subsec:fjortoft-interpretation}).} For $l=m=2$ ZDMs:

\begin{align}
    \tau_{\rm nl}^{-1}
    \sim
    \epsilon^2\,\omega_R
    =
    \epsilon^2\,m\Omega_H,
    \qquad
    \tau_{\rm damp}
    \sim
    \kappa_+^{-1}
    \simeq
    \frac{M\sqrt{2}}{\sqrt{1-\chi}}.
\end{align}

The separation ratio is
\begin{align}
    \frac{\tau_{\rm damp}}{\tau_{\rm nl}}
    =
    \frac{\epsilon^2\,m\Omega_H}{\kappa_+}
    \simeq
    \frac{\epsilon^2\,m}{\sqrt{2(1-\chi)}}.
\end{align}

For $l=m=2$, the onset condition $\tau_{\rm damp}/\tau_{\rm nl}\geq 1$ gives
\begin{align}
    \epsilon^2
    \geq
    \frac{\sqrt{2(1-\chi)}}{2}.
\end{align}

Explicit thresholds:
\begin{itemize}
\item $\chi=0.99$: $\epsilon\gtrsim 0.27$
\item $\chi=0.999$: $\epsilon\gtrsim 0.15$
\item $\chi=0.9999$: $\epsilon\gtrsim 0.084$
\end{itemize}

At $\epsilon=0.1$, the ratio is $0.14$ at $\chi=0.99$, $0.45$ at 
$\chi=0.999$, and $1.41$ at $\chi=0.9999$. The last value is only a crossover, with roughly half a quadratic norm lost in one transfer time. A ratio of ten at $\epsilon=0.1$ requires $1-\chi\lesssim2\times10^{-6}$ under the same rate estimate.

\rev{Figure~\ref{fig:timescale-separation} displays this comparison. Above the equality line, the assumed transfer rate exceeds linear damping; small losses over a transfer time require a ratio much greater than one. The alternative cubic estimate and its lower amplitude threshold are given in \cref{subsec:near-extremal-kerr-flux}.}

\bibliographystyle{apsrev4-2}
\bibliography{ref}

@article{Bondi1962,
  author  = {Bondi, H. and van der Burg, M. G. J. and Metzner, A. W. K.},
  title   = {Gravitational waves in general relativity. {VII}. Waves from axi-symmetric isolated systems},
  journal = {Proc. R. Soc. Lond. A},
  volume  = {269},
  pages   = {21--52},
  year    = {1962}
}

@article{AshtekarStreubel1981,
  author  = {Ashtekar, A. and Streubel, M.},
  title   = {Symplectic geometry of radiative modes and conserved quantities at null infinity},
  journal = {Proc. R. Soc. Lond. A},
  volume  = {376},
  pages   = {585--607},
  year    = {1981}
}

@article{Ashtekar1981,
  author  = {Ashtekar, A.},
  title   = {Asymptotic quantization of the gravitational field},
  journal = {Phys. Rev. Lett.},
  volume  = {46},
  pages   = {573--576},
  year    = {1981}
}

@book{Nazarenko2011,
   author    = {Nazarenko, Sergey},
   title     = {Wave Turbulence},
   series    = {Lecture Notes in Physics},
   volume    = {825},
   publisher = {Springer},
   address   = {Berlin, Heidelberg},
   year      = {2011},
   doi       = {10.1007/978-3-642-15942-8}
 }

@article{GaltierNazarenko2017,
   author    = {Galtier, S{\'e}bastien and Nazarenko, Sergey V.},
   title     = {Turbulence of Weak Gravitational Waves in the Early Universe},
   journal   = {Phys. Rev. Lett.},
   volume    = {119},
   pages     = {221101},
   year      = {2017},
   doi       = {10.1103/PhysRevLett.119.221101},
   eprint    = {1703.09069},
   archivePrefix = {arXiv},
   primaryClass  = {gr-qc}
 }

@article{Westernacher-Schneider:2015gfa,
    author = "Westernacher-Schneider, John Ryan and Lehner, Luis and Oz, Yaron",
    title = "{Scaling Relations in Two-Dimensional Relativistic Hydrodynamic Turbulence}",
    eprint = "1510.00736",
    archivePrefix = "arXiv",
    primaryClass = "hep-th",
    doi = "10.1007/JHEP12(2015)067",
    journal = "JHEP",
    volume = "12",
    pages = "067",
    year = "2015"
}

@ARTICLE{1980RPPh...43..547K,
       author = {{Kraichnan}, R.~H. and {Montgomery}, D.},
        title = "{REVIEW ARTICLE: Two-dimensional turbulence}",
      journal = {Reports on Progress in Physics},
         year = 1980,
        month = may,
       volume = {43},
       number = {5},
        pages = {547-619},
          doi = {10.1088/0034-4885/43/5/001},
       adsurl = {https://ui.adsabs.harvard.edu/abs/1980RPPh...43..547K}
}

@article{Maartens1998,
  author  = {Maartens, Roy and Bassett, Bruce A.},
  title   = {Gravito-electromagnetism},
  journal = {Classical and Quantum Gravity},
  volume  = {15},
  pages   = {705--717},
  year    = {1998},
  eprint  = {gr-qc/9704059},
  archivePrefix = {arXiv},
  doi     = {10.1088/0264-9381/15/3/018}
}

@article{Fjortoft1953,
  author  = {Fj{\o}rtoft, Ragnar},
  title   = {On the changes in the spectral distribution of kinetic energy for two-dimensional, nondivergent flow},
  journal = {Tellus},
  volume  = {5},
  number  = {3},
  pages   = {225--230},
  year    = {1953},
  doi     = {10.1111/j.2153-3490.1953.tb01051.x}
}

@article{Kraichnan1967,
  author  = {Kraichnan, Robert H.},
  title   = {Inertial ranges in two-dimensional turbulence},
  journal = {Physics of Fluids},
  volume  = {10},
  pages   = {1417--1423},
  year    = {1967},
  doi     = {10.1063/1.1762301}
}

@article{TeukolskyPress1974,
  author  = {Teukolsky, Saul A. and Press, William H.},
  title   = {Perturbations of a rotating black hole. {III}. {I}nteraction of the hole with gravitational and electromagnetic radiation},
  journal = {The Astrophysical Journal},
  volume  = {193},
  pages   = {443--461},
  year    = {1974},
  doi     = {10.1086/153180}
}

@article{YangZimmermanLehner2015,
  author  = {Yang, Huan and Zimmerman, Aaron and Lehner, Luis},
  title   = {Turbulent black holes},
  journal = {Physical Review Letters},
  volume  = {114},
  pages   = {081101},
  year    = {2015},
  eprint  = {1402.4859},
  archivePrefix = {arXiv},
  primaryClass = {gr-qc},
  doi     = {10.1103/PhysRevLett.114.081101}
}

@article{MaLehnerYangKidderPfeifferScheel2026,
  author  = {Ma, Sizheng and Lehner, Luis and Yang, Huan and Kidder, Lawrence E. and Pfeiffer, Harald P. and Scheel, Mark A.},
  title   = {Emergent Turbulence in Nonlinear Gravity},
  journal = {Physical Review Letters},
  volume  = {136},
  pages   = {061401},
  year    = {2026},
  eprint  = {2508.13294},
  archivePrefix = {arXiv},
  primaryClass = {gr-qc}
}

@article{BizonRostworowski2011,
  author  = {Bizo{\'n}, Piotr and Rostworowski, Andrzej},
  title   = {On weakly turbulent instability of anti-de {S}itter space},
  journal = {Physical Review Letters},
  volume  = {107},
  pages   = {031102},
  year    = {2011},
  eprint  = {1104.3702},
  archivePrefix = {arXiv},
  primaryClass = {gr-qc},
  doi     = {10.1103/PhysRevLett.107.031102}
}

@article{BuchelGreenLehnerLiebling2015,
  author  = {Buchel, Alex and Green, Stephen R. and Lehner, Luis and Liebling, Steven L.},
  title   = {Conserved quantities and dual turbulent cascades in anti--de {S}itter spacetime},
  journal = {Physical Review D},
  volume  = {91},
  pages   = {064026},
  year    = {2015},
  eprint  = {1412.4761},
  archivePrefix = {arXiv},
  primaryClass = {gr-qc},
  doi     = {10.1103/PhysRevD.91.064026}
}

@article{GreenCarrascoLehner2014,
  author  = {Green, Stephen R. and Carrasco, Federico and Lehner, Luis},
  title   = {Holographic path to the turbulent side of gravity},
  journal = {Physical Review X},
  volume  = {4},
  pages   = {011001},
  year    = {2014},
  eprint  = {1309.7940},
  archivePrefix = {arXiv},
  primaryClass = {hep-th},
  doi     = {10.1103/PhysRevX.4.011001}
}

@article{CarrascoLehnerMyersReulaSingh2012,
  author  = {Carrasco, Federico and Lehner, Luis and Myers, Robert C. and Reula, Oscar and Singh, Ajay},
  title   = {Turbulent flows for relativistic conformal fluids in 2+1 dimensions},
  journal = {Physical Review D},
  volume  = {86},
  pages   = {126006},
  year    = {2012},
  eprint  = {1210.6702},
  archivePrefix = {arXiv},
  primaryClass = {hep-th},
  doi     = {10.1103/PhysRevD.86.126006}
}

@book{WaldGR,
  author    = {Wald, Robert M.},
  title     = {General Relativity},
  publisher = {University of Chicago Press},
  address   = {Chicago},
  year      = {1984}
}

@article{Leaver1985,
  author  = {Leaver, Edward W.},
  title   = {An analytic representation for the quasi-normal modes of {K}err black holes},
  journal = {Proceedings of the Royal Society of London A},
  volume  = {402},
  pages   = {285--298},
  year    = {1985},
  doi     = {10.1098/rspa.1985.0119}
}

@article{BelRobinson,
  author  = {Bel, Lluis},
  title   = {Radiation states and the problem of energy in general relativity},
  journal = {General Relativity and Gravitation},
  volume  = {32},
  pages   = {2047--2078},
  year    = {2000},
  note    = {Golden Oldie reprint of the 1962 original},
  doi     = {10.1023/A:1001958805232}
}

@article{Adams:2013vsa,
    author = "Adams, Allan and Chesler, Paul M. and Liu, Hong",
    title = "{Holographic turbulence}",
    journal = "Phys. Rev. Lett.",
    volume = "112",
    number = "15",
    pages = "151602",
    year = "2014",
    eprint = "1307.7267",
    archivePrefix = "arXiv",
    primaryClass = "hep-th",
    doi = "10.1103/PhysRevLett.112.151602",
    adsurl = "https://harvard.edu"
}

@misc{iuliano2024extremalblackholeweather,
      title={Extremal Black Hole Weather}, 
      author={Claudio Iuliano and Stefan Hollands and Stephen R. Green and Peter Zimmerman},
      year={2024},
      eprint={2412.02821},
      archivePrefix={arXiv},
      primaryClass={gr-qc},
      url={https://arxiv.org/abs/2412.02821}, 
}

@misc{marjieh2021enstrophysymmetry,
      title={Enstrophy from symmetry}, 
      author={Raja Marjieh and Natalia Pinzani-Fokeeva and Amos Yarom},
      year={2021},
      eprint={2009.03980},
      archivePrefix={arXiv},
      primaryClass={hep-th},
      url={https://arxiv.org/abs/2009.03980}, 
}

@article{Redondo-Yuste:2022czt,
    author = "Redondo-Yuste, Jaime and Lehner, Luis",
    title = "{Non-linear black hole dynamics and Carrollian fluids}",
    journal = "JHEP",
    volume = "02",
    pages = "240",
    year = "2023",
    eprint = "2212.06175",
    archivePrefix = "arXiv",
    primaryClass = "gr-qc",
    doi = "10.1007/JHEP02(2023)240"
}

@article{Eling:2013sna,
    author = "Eling, Christopher and Oz, Yaron",
    title = "{Holographic Vorticity in the Fluid/Gravity Correspondence}",
    eprint = "1308.1651",
    archivePrefix = "arXiv",
    primaryClass = "hep-th",
    doi = "10.1007/JHEP11(2013)079",
    journal = "JHEP",
    volume = "11",
    pages = "079",
    year = "2013"
}

@article{bardeen1972rotating,title={Rotating black holes: locally nonrotating frames, energy extraction, and scalar synchrotron radiation},author={Bardeen, James M and Press, William H and Teukolsky, Saul A},journal={The Astrophysical Journal},volume={178},pages={347--374},year={1972}}

@article{Senovilla2000,
  author  = {Senovilla, Jos{\'e} M. M.},
  title   = {Super-energy tensors},
  journal = {Classical and Quantum Gravity},
  volume  = {17},
  pages   = {2799--2841},
  year    = {2000},
  eprint  = {gr-qc/9906087},
  archivePrefix = {arXiv},
  doi     = {10.1088/0264-9381/17/14/313}
}

@article{GarciaParrado2008,
  author  = {Garc{\'\i}a-Parrado G{\'o}mez-Lobo, Alfonso},
  title   = {Dynamical laws of superenergy in general relativity},
  journal = {Classical and Quantum Gravity},
  volume  = {25},
  pages   = {015006},
  year    = {2008},
  eprint  = {0707.1475},
  archivePrefix = {arXiv},
  primaryClass = {gr-qc},
  doi     = {10.1088/0264-9381/25/1/015006}
}

@article{NicholsOwenEtAl2011,
  author  = {Nichols, David A. and Owen, Robert and Zhang, Fan and Zimmerman, Aaron and Brink, Jeandrew and Chen, Yanbei and Kaplan, Jeffrey D. and Lovelace, Geoffrey and Matthews, Keith D. and Scheel, Mark A. and Thorne, Kip S.},
  title   = {Visualizing spacetime curvature via frame-drag vortexes and tidal tendexes: General theory and weak-gravity applications},
  journal = {Physical Review D},
  volume  = {84},
  pages   = {124014},
  year    = {2011},
  eprint  = {1108.5486},
  archivePrefix = {arXiv},
  primaryClass = {gr-qc},
  doi     = {10.1103/PhysRevD.84.124014}
}

@article{Yang:2012pj,
    author = "Yang, Huan and Zhang, Fan and Zimmerman, Aaron and Nichols, David A. and Berti, Emanuele and Chen, Yanbei",
    title = "{Branching of quasinormal modes for nearly extremal Kerr black holes}",
    eprint = "1212.3271",
    archivePrefix = "arXiv",
    primaryClass = "gr-qc",
    doi = "10.1103/PhysRevD.87.041502",
    journal = "Phys. Rev. D",
    volume = "87",
    number = "4",
    pages = "041502",
    year = "2013"
}

@article{Bhattacharyya:2007vjd,
    author = "Bhattacharyya, Sayantani and Hubeny, Veronika E and Minwalla, Shiraz and Rangamani, Mukund",
    title = "{Nonlinear Fluid Dynamics from Gravity}",
    eprint = "0712.2456",
    archivePrefix = "arXiv",
    primaryClass = "hep-th",
    reportNumber = "TIFR-TH-07-44, DCPT-07-73, NI07097",
    doi = "10.1088/1126-6708/2008/02/045",
    journal = "JHEP",
    volume = "02",
    pages = "045",
    year = "2008"
}

@article{Newman:1962cia,
    author = "Newman, Ezra and Penrose, Roger",
    title = "{An approach to gravitational radiation by a method of spin coefficients}",
    journal = "J. Math. Phys.",
    volume = "3",
    pages = "566--578",
    year = "1962",
    doi = "10.1063/1.1724257"
}

@article{LoutrelRipleyGiorgiPretorius2021,
  author       = {Loutrel, Nicholas and Ripley, Justin L. and Giorgi, Elena and Pretorius, Frans},
  title        = {Second-order perturbations of Kerr black holes: Formalism and reconstruction of the first-order metric},
  journal      = {Physical Review D},
  volume       = {103},
  number       = {10},
  pages        = {104017},
  year         = {2021},
  doi          = {10.1103/PhysRevD.103.104017},
  eprint       = {2008.11770},
  archivePrefix= {arXiv},
  primaryClass = {gr-qc}
}

@article{RipleyLoutrelGiorgiPretorius2021,
  author       = {Ripley, Justin L. and Loutrel, Nicholas and Giorgi, Elena and Pretorius, Frans},
  title        = {Numerical computation of second order vacuum perturbations of Kerr black holes},
  journal      = {Physical Review D},
  volume       = {103},
  number       = {10},
  pages        = {104018},
  year         = {2021},
  doi          = {10.1103/PhysRevD.103.104018},
  eprint       = {2010.00162},
  archivePrefix= {arXiv},
  primaryClass = {gr-qc}
}

@article{DiasHorowitzSantos2012,
  author       = {Dias, {\'O}scar J. C. and Horowitz, Gary T. and Santos, Jorge E.},
  title        = {Gravitational turbulent instability of anti-de Sitter space},
  journal      = {Classical and Quantum Gravity},
  volume       = {29},
  number       = {19},
  pages        = {194002},
  year         = {2012},
  doi          = {10.1088/0264-9381/29/19/194002},
  eprint       = {1109.1825},
  archivePrefix= {arXiv},
  primaryClass = {hep-th}
}

@article{DiasHorowitzMarolfSantos2012,
  author       = {Dias, {\'O}scar J. C. and Horowitz, Gary T. and Marolf, Donald and Santos, Jorge E.},
  title        = {On the nonlinear stability of asymptotically anti-de Sitter solutions},
  journal      = {Classical and Quantum Gravity},
  volume       = {29},
  number       = {23},
  pages        = {235019},
  year         = {2012},
  doi          = {10.1088/0264-9381/29/23/235019},
  eprint       = {1208.5772},
  archivePrefix= {arXiv},
  primaryClass = {gr-qc}
}

@article{DiasSantos2016,
  author       = {Dias, {\'O}scar J. C. and Santos, Jorge E.},
  title        = {AdS nonlinear instability: moving beyond spherical symmetry},
  journal      = {Classical and Quantum Gravity},
  volume       = {33},
  number       = {23},
  pages        = {23LT01},
  year         = {2016},
  doi          = {10.1088/0264-9381/33/23/23LT01},
  eprint       = {1602.03890},
  archivePrefix= {arXiv},
  primaryClass = {hep-th}
}

@article{Bel1958,
    author = {Bel, L.},
    title = {Sur le tenseur de super-\'energie},
    journal = {Comptes Rendus de l'Acad\'emie des Sciences},
    volume = {247},
    pages = {1094--1096},
    year = {1958}
}

@book{Misner1973,
    author = {Misner, Charles W. and Thorne, Kip S. and Wheeler, John Archibald},
    title = {Gravitation},
    publisher = {W. H. Freeman},
    address = {San Francisco},
    year = {1973},
    isbn = {978-0-7167-0344-0}
}

@ARTICLE{1996CQGra..13.1451F,
       author = {{Friedrich}, Helmut},
        title = "{Hyperbolic reductions for Einstein's equations}",
      journal = {Classical and Quantum Gravity},
         year = 1996,
        month = jun,
       volume = {13},
       number = {6},
        pages = {1451-1469},
          doi = {10.1088/0264-9381/13/6/014},
       adsurl = {https://ui.adsabs.harvard.edu/abs/1996CQGra..13.1451F}
}

@article{Isaacson1968b,
  author  = {Isaacson, Richard A.},
  title   = {Gravitational radiation in the limit of high frequency. {II}. {N}onlinear terms and the effective stress tensor},
  journal = {Physical Review},
  volume  = {166},
  pages   = {1272--1280},
  year    = {1968},
  doi     = {10.1103/PhysRev.166.1272}
}

@article{Wald1973,
  author  = {Wald, Robert M.},
  title   = {On perturbations of a {K}err black hole},
  journal = {Journal of Mathematical Physics},
  volume  = {14},
  pages   = {1453--1461},
  year    = {1973},
  doi     = {10.1063/1.1666203}
}

@article{Zhu:2026mhn,
    author = "Zhu, Hengrui and Pretorius, Frans and Stone, James M.",
    title = "{Trapping, Irregular Waveforms, and Efficient Radiation in Ultra-relativistic Black Hole Encounters}",
    eprint = "2604.26253",
    journal       = "",    
    archivePrefix = "arXiv",
    primaryClass = "gr-qc",
    month = "4",
    year = "2026"
}

@article{Cardoso:2026llh,
    author = "Cardoso, Vitor and Redondo-Yuste, Jaime and Sperhake, Ulrich and Tuncer, Furkan",
    title = "{Nonlinear Dynamics in General Relativity}",
    eprint = "2603.04501",
    journal       = "",    
    archivePrefix = "arXiv",
    primaryClass = "gr-qc",
    month = "3",
    year = "2026"
}

@article{Redondo-Yuste:2025hlv,
    author = "Redondo-Yuste, Jaime and C{\'a}rdenas-Avenda{\~n}o, Alejandro",
    title = "{Perturbative and nonlinear analyses of gravitational turbulence in spacetimes with stable light rings}",
    eprint = "2502.18643",
    archivePrefix = "arXiv",
    primaryClass = "gr-qc",
    doi = "10.1103/hy3r-ww3w",
    journal = "Phys. Rev. D",
    volume = "111",
    number = "12",
    pages = "124009",
    year = "2025"
}

@article{Siemonsen:2025fne,
    author = "Siemonsen, Nils",
    title = "{Weakly Turbulent Saturation of the Nonlinear Scalar Ergoregion Instability}",
    eprint = "2510.07467",
    archivePrefix = "arXiv",
    primaryClass = "gr-qc",
    doi = "10.1103/ytqg-v1j8",
    journal = "Phys. Rev. Lett.",
    volume = "136",
    number = "17",
    pages = "171401",
    year = "2026"
}

@article{Kehle:2026vxm,
    author = "Kehle, Christoph and Moschidis, Georgios",
    title = "{Weakly turbulent dynamics on Schwarzschild-AdS black hole spacetimes}",
    journal       = "",    
    eprint = "2604.12118",
    archivePrefix = "arXiv",
    primaryClass = "gr-qc",
    month = "4",
    year = "2026"
}

@article{Kinnersley:1969zz,
    author = "Kinnersley, William",
    title = "{Type D Vacuum Metrics}",
    journal = "J. Math. Phys.",
    volume = "10",
    pages = "1195--1203",
    year = "1969",
    doi = "10.1063/1.1664958"
}

@article{CardosoLemos2001,
  author  = {Cardoso, Vitor and Lemos, Jos{\'e} P. S.},
  title   = {Quasinormal modes of {Schwarzschild}--anti-de {Sitter} black holes: Electromagnetic and gravitational perturbations},
  journal = {Physical Review D},
  volume  = {64},
  pages   = {084017},
  year    = {2001},
  eprint  = {gr-qc/0105103},
  archivePrefix = {arXiv},
  doi     = {10.1103/PhysRevD.64.084017}
}

@article{Konoplya2002,
  author  = {Konoplya, R. A.},
  title   = {Quasinormal modes of a small {Schwarzschild}--anti-de {Sitter} black hole},
  journal = {Physical Review D},
  volume  = {66},
  pages   = {044009},
  year    = {2002},
  eprint  = {hep-th/0205142},
  archivePrefix = {arXiv},
  doi     = {10.1103/PhysRevD.66.044009}
}

@article{ComteBellot1971,
  author = {Comte-Bellot, G. and Corrsin, S.},
  title = {Simple and modified energy spectra in grid turbulence in an incompressible fluid},
  journal = {Journal of Fluid Mechanics},
  volume = {48},
  number = {2},
  pages = {273--337},
  year = {1971},
  publisher = {Cambridge University Press}
}

@article{Craps_2015,
   title={Renormalization, averaging, conservation laws and AdS (in)stability},
   volume={2015},
   ISSN={1029-8479},
   url={http://dx.doi.org/10.1007/JHEP01(2015)108},
   DOI={10.1007/jhep01(2015)108},
   number={1},
   journal={Journal of High Energy Physics},
   publisher={Springer Science and Business Media LLC},
   author={Craps, Ben and Evnin, Oleg and Vanhoof, Joris},
   year={2015},
   month=Jan }

@article{Ianniccari:2025nkf,
    author = "Ianniccari, A. and Kehagias, A. and Bianco, L. Lo and Riotto, A.",
    title = "{Nonlinear Gravity and Multipole Turbulence}",
    eprint = "2512.08872",
    journal       = "",    
    archivePrefix = "arXiv",
    primaryClass = "gr-qc",
    month = "12",
    year = "2025"
}

@article{ManleyRowe1956,
  author = {Manley, J. M. and Rowe, H. E.},
  title = {Some General Properties of Nonlinear Elements---Part I: General Energy Relations},
  journal = {Proceedings of the IRE},
  volume = {44},
  number = {7},
  pages = {904--913},
  year = {1956},
  doi = {10.1109/JRPROC.1956.275145}
}

@article{Bantilan:2014sra,
    author = "Bantilan, Hans and Romatschke, Paul",
    title = "{Simulation of Black Hole Collisions in Asymptotically Anti{\textendash}de Sitter Spacetimes}",
    eprint = "1410.4799",
    archivePrefix = "arXiv",
    primaryClass = "hep-th",
    doi = "10.1103/PhysRevLett.114.081601",
    journal = "Phys. Rev. Lett.",
    volume = "114",
    number = "8",
    pages = "081601",
    year = "2015"
}

@article{Kehagias:2025zws,
    author = "Kehagias, Alex and Riotto, Antonio",
    title = "{Schwarzschild Black Hole Turbulence: Scalar Probe}",
    eprint = "2512.05003",
    journal       = "",    
    archivePrefix = "arXiv",
    primaryClass = "gr-qc",
    month = "12",
    year = "2025"
}

@article{Bantilan:2012vu,
    author = "Bantilan, Hans and Pretorius, Frans and Gubser, Steven S.",
    title = "{Simulation of Asymptotically AdS5 Spacetimes with a Generalized Harmonic Evolution Scheme}",
    eprint = "1201.2132",
    archivePrefix = "arXiv",
    primaryClass = "hep-th",
    doi = "10.1103/PhysRevD.85.084038",
    journal = "Phys. Rev. D",
    volume = "85",
    pages = "084038",
    year = "2012"
}

@article{Donnelly:2020xgu,
    author = "Donnelly, William and Freidel, Laurent and Moosavian, Seyed Faroogh and Speranza, Antony J.",
    title = "{Gravitational edge modes, coadjoint orbits, and hydrodynamics}",
    eprint = "2012.10367",
    archivePrefix = "arXiv",
    primaryClass = "hep-th",
    doi = "10.1007/JHEP09(2021)008",
    journal = "JHEP",
    volume = "09",
    pages = "008",
    year = "2021"
}

@book{chandrasekhar1983,title     = {The Mathematical Theory of Black Holes},author    = {Chandrasekhar, Subrahmanyan},year      = {1983},publisher = {Oxford University Press},address   = {New York},isbn      = {978-0198512912}}

@incollection{NewmanTod1980,
  author    = {E. T. Newman and K. P. Tod},
  title     = {Asymptotically flat space-times},
  booktitle = {General Relativity and Gravitation},
  publisher = {Plenum Press},
  year      = {1980},
  editor    = {A. Held},
  volume    = {2},
  pages     = {1--36},
  address   = {New York}
}

@article{GayBabichevGaltierNoui2026,
  author = {Gay, Beno{\^i}t and Babichev, Eugeny and Galtier, S{\'e}bastien and Noui, Karim},
  title = {Toward gravitational wave turbulence within the {Hadad--Zakharov} metric},
  journal = {Phys. Rev. D},
  volume = {113},
  pages = {124074},
  year = {2026},
  doi = {10.1103/vc5p-c3w8},
  eprint = {2603.29699},
  archivePrefix = {arXiv},
  primaryClass = {gr-qc}
}

@article{FischerMarsden1973,
  author = {Fischer, Arthur E. and Marsden, Jerrold E.},
  title = {Linearization stability of the {Einstein} equations},
  journal = {Bull. Amer. Math. Soc.},
  volume = {79},
  number = {5},
  pages = {997--1003},
  year = {1973},
  doi = {10.1090/S0002-9904-1973-13299-9}
}

@article{Moncrief1975,
  author = {Moncrief, Vincent},
  title = {Spacetime symmetries and linearization stability of the {Einstein} equations. {I}},
  journal = {J. Math. Phys.},
  volume = {16},
  number = {3},
  pages = {493--498},
  year = {1975},
  doi = {10.1063/1.522572}
}

\end{document}